\documentclass[11pt]{article}
\usepackage{graphicx}
\usepackage{subcaption}
\usepackage{todonotes}
\usepackage{mathtools}
\usepackage{soul,comment,cite}
\usepackage{amsthm}
\usepackage{amsmath}
\usepackage{systeme}

\usepackage{gensymb}
\usepackage[ruled,vlined]{algorithm2e}
\SetKwComment{Comment}{/* }{ */}
\SetKwInOut{Input}{input}\SetKwInOut{Output}{output}
\newtheorem{theorem}{Theorem}
\newtheorem{corollary}{Corollary}
\newtheorem{lemma}{Lemma}
\newtheorem{observation}{Observation}

\newtheorem{definition}{Definition}

\graphicspath{{./figures/}}
\usepackage[absolute]{textpos}

\usepackage{fullpage}
\usepackage[top=0.8in, bottom=0.9in, left=0.9in, right=0.9in]{geometry}

\usepackage[pdftex, plainpages = false, pdfpagelabels, 
                 bookmarks=false,
                 bookmarksopen = true,
                 bookmarksnumbered = true,
                 breaklinks = true,
                 linktocpage,
                 pagebackref,
                 colorlinks = true,  
                 linkcolor = blue,
                 urlcolor  = cyan,
                 citecolor = red,
                 anchorcolor = green,
                 hyperindex = true,
                 hyperfigures
                 ]{hyperref}

\def\calD{\mathcal{D}}
\def\calI{\mathcal{I}}
\def\calC{\mathcal{C}}

\def\calH{\mathcal{H}}

\def\calT{\mathcal{T}}
\def\dt{\mathcal{DT}} 
\def\vd{\mathcal{VD}} 

\def\fvd{\mathcal{F}\!\mathcal{V}\!\mathcal{D}}
\def\calA{\mathcal{A}}
\def\calL{\mathcal{L}}
\def\calI{\mathcal{I}}
\def\fvd{\mathcal{F}\!\mathcal{V}\!\mathcal{D}}

\title{Dominating Set, Independent Set, Discrete $k$-Center, Dispersion, and Related Problems for Planar Points in Convex Position\thanks{A preliminary version of this paper will appear in {\em Proceedings of the 42nd International Symposium on Theoretical Aspects of Computer Science (STACS 2025)}. This research was supported in part by NSF under Grant CCF-2300356.}}
\author{Anastasiia Tkachenko and Haitao Wang}
\date{January 2024}

\author{Anastasiia Tkachenko\thanks{Kahlert School of Computing,
University of Utah, Salt Lake City, UT 84112, USA. {\tt anastasiia.tkachenko@utah.edu}}
\and
Haitao Wang\thanks{Kahlert School of Computing,
University of Utah, Salt Lake City, UT 84112, USA. {\tt haitao.wang@utah.edu}}
}
\date{}

\begin{document}

\maketitle

\vspace{-0.2in}
\begin{abstract}
Given a set $P$ of $n$ points in the plane, its unit-disk graph $G(P)$ is a graph with $P$ as its vertex set such that two points of $P$ are connected by an edge if their (Euclidean) distance is at most $1$. We consider several classical problems on $G(P)$ in a special setting when points of $P$ are in convex position. These problems are all NP-hard in the general case. We present efficient algorithms for these problems under the convex position assumption. 
\begin{itemize}
\item For the problem of finding the smallest dominating set of $G(P)$, we present an $O(kn\log n)$ time algorithm, where $k$ is the smallest dominating set size. We also consider the weighted case in which each point of $P$ has a weight and the goal is to find a dominating set in $G(P)$ with minimum total weight; our algorithm runs in $O(n^3\log^2 n)$ time. In particular, for a given $k$, our algorithm can compute in $O(kn^2\log^2 n)$ time a minimum weight dominating set of size at most $k$ (if it exists). 

\item For the discrete $k$-center problem, which is to find a subset of $k$ points in $P$ (called {\em centers}) for a given $k$, such that the maximum distance between any point in $P$ and its nearest center is minimized. 
We present an algorithm that solves the problem in $O(\min\{n^{4/3}\log n+kn\log^2 n,k^2 n\log^2n\})$ time, which is $O(n^2\log^2 n)$ in the worst case when $k=\Theta(n)$. For comparison, the runtime of the current best algorithm for the continuous version of the problem where centers can be anywhere in the plane is $O(n^3 \log n)$. 

\item For the problem of finding a maximum independent set in $G(P)$, we give an algorithm of $O(n^{7/2})$ time and another randomized algorithm of $O(n^{37/11})$ expected time, which improve the previous best result of $O(n^6\log n)$ time. Our algorithms can be extended to compute a maximum-weight independent set in $G(P)$ with the same time complexities when points of $P$ have weights. 

\begin{itemize}
    \item  If we are looking for an (unweighted) independent set of size $3$, we derive an algorithm of $O(n\log n)$ time; the previous best algorithm runs in $O(n^{4/3}\log^2 n)$ time (which works for the general case where points of $P$ are not necessarily in convex position). 
\item 
If points of $P$ have weights and are not necessarily in convex position, we present an algorithm that can find a maximum-weight independent set of size $3$ in $O(n^{5/3+\delta})$ time for an arbitrarily small constant $\delta>0$. By slightly modifying the algorithm, a maximum-weight clique of size $3$ can also be found within the same time complexity. 
\end{itemize}

\item 
For the dispersion problem, which is to find a subset of $k$ points from $P$ for a given $k$, such that the minimum pairwise distance of the points in the subset is maximized. We present an algorithm of $O(n^{7/2}\log n)$ time and another randomized algorithm of $O(n^{37/11}\log n)$ expected time, which improve the previous best result of $O(n^6)$ time. 

\begin{itemize}
\item  If $k=3$, we present an algorithm of $O(n\log^2 n)$ time and another randomized algorithm of $O(n\log n)$ expected time; the previous best algorithm runs in $O(n^{4/3}\log^2 n)$ time (which works for the general case where points of $P$ are not necessarily in convex position). 
\end{itemize}
\end{itemize}
\end{abstract}


\section{Introduction}
\label{sec:intro}

Let $P$ be a set of $n$ points in the plane. The \textit{unit-disk graph} of $P$, denoted by $G(P)$, is the graph with $P$ as its vertex set such that two points are connected by an edge if their (Euclidean) distance is at most $1$. Equivalently, $G(P)$ is the intersection graph of congruent disks with radius $1/2$ and centered at the points in $P$ (i.e., two disks have an edge if they intersect). This model is particularly useful in applications such as wireless sensor networks, where connectivity is determined by signal ranges, represented by unit disks~\cite{ref:ClarkUn90, ref:PerkinsHi94, ref:PerkinsAd99, ref:BalisterCo05}.

\subsection{Our results}

We consider several classical problems on $G(P)$. These problems are all NP-hard. However, little attention has been given to special configurations of points, such as when the points are in convex position, despite the potential for significant algorithmic simplifications in such cases. In this paper, we systematically study these problems under the condition that the points of $P$ are in convex position (i.e., every point of $P$ appears as a vertex in the convex hull of $P$) and present efficient algorithms. We hope our results can lead to efficient solutions to other problems in this setting. 

\paragraph{Dominating set.} 
A \textit{dominating set} of $G(P)$ is a subset $S$ of vertices of $G(P)$ such that each vertex of $G(P)$ is either in $S$ or adjacent to a vertex in $S$. The dominating set problem, which seeks a dominating set of smallest size, is a classical NP-hard problem~\cite{ref:ClarkUn90,ref:HsuEa79, ref:KarpRe72}. In the weighted case, each point of $P$ has a weight and the problem is to find a dominating set of minimum total weight. The dominating set problem has been widely studied, with various approximation algorithms proposed~\cite{ref:GibsonAl10, ref:HernandezAp22, ref:MustafaIm10, ref:DeGe23}. 

To the best of our knowledge, we are not aware of any previous work under the convex position assumption. For the unweighted case, we present an algorithm of $O(kn\log n)$ time, where $k$ is the smallest dominating set size of $G(P)$. For the weighted case, we derive an algorithm of $O(n^3\log^2 n)$ time. In particular, given any $k$, our algorithm can compute in $O(kn^2 \log^2 n)$ time a minimum-weight dominating set of size at most $k$. 
    
\paragraph{Discrete $k$-center.} 
A closely related problem is the \textit{discrete $k$-center} problem. Given a number $k$, the problem is to compute a subset of $k$ points in $P$ (called {\em centers}) such that the maximum distance between any point in $P$ and its nearest center is minimized. The problem, which is NP-hard~\cite{ref:VaziraniAp01}, is also a classical problem with applications in clustering, facility locations, and network design. 
An algorithm for the dominating set problem can be used to solve the {\em decision version} of the discrete $k$-center problem: Given a value $r$ and $k$, decide whether there exists a subset of $k$ centers such that the distance from any point in $P$ to its nearest center is at most $r$. Indeed, if we define the unit-disk graph of $P$ with respect to $r$, then a dominating set of size $k$ in the graph is a discrete $k$-center of $P$ for $r$, and vice versa. 

For the convex position case, we are not aware of any previous work. We propose an algorithm whose runtime is $O(\min\{n^{4/3}\log n+kn\log^2 n,k^2 n\log^2n\})$. In particular, if $k=O(1)$, then the runtime is $O(n\log^2 n)$. 

\paragraph{Independent set.}
An {\em independent set} of $G(P)$ is a subset of the vertices such that no two vertices have an edge. 
The {\em maximum independent set} problem is to find an independent set of the largest cardinality. The problem of finding a maximum independent set in $G(P)$ is NP-hard~\cite{ref:ClarkUn90}. Many approximation algorithms for the problem have been developed in the literature, e.g.,~\cite{ref:DasDi20,ref:DasAp15,ref:MaratheSi95,ref:MatsuiAp98}.

Under the convex position assumption, using the technique of Singireddy, Basappa, and Mitchell~\cite{ref:SingireddyAl23} for a dispersion problem (more details to be discussed later), one can find a maximum independent set in $G(P)$ in $O(n^6\log n)$ time. We give a new algorithm of $O(n^{7/2})$ time,\footnote{Throughout the paper, the algorithm runtime is deterministic unless otherwise stated.} and another randomized algorithm of $O(n^{37/11})$ expected time using the recent randomized result of Agarwal, Ezra, and Sharir~\cite{ref:AgarwalSe24}. Furthermore, our algorithms can be extended to compute a maximum-weight independent set of $G(P)$ within the same time complexities when points of $P$ have weights; specifically, a {\em maximum-weight independent set} is an independent set whose total vertex weight is maximized. 
Since the vertices of a graph other than those in an independent set form a vertex cover, our algorithm can also compute a minimum-weight vertex cover of $G(P)$ in $O(n^{7/2})$ time or in randomized $O(n^{37/11})$ expected time. 

Furthermore, we consider a small-size case that is to find an (unweighted) independent set of size $3$ in $G(P)$. If $P$ is not necessarily in convex position, the problem has been studied by Agarwal, Overmars, and Sharir~\cite{ref:AgarwalCo06}, who presented an $O(n^{4/3}\log^2n)$ time algorithm. We consider the convex position case and derive an algorithm of $O(n\log n)$ time. Note that finding an independent set of size $2$ is equivalent to computing a farthest pair of points of $P$, which can be done in $O(n\log n)$ time using the farthest Voronoi diagram~\cite{ref:ShamosCl75}.

In addition, we consider a more general small-size case that is to find a maximum-weight independent set of size $3$ in $G(P)$ when points of $P$ have weights and are not necessarily in convex position. Our algorithm runs in $O(n^{5/3+\delta})$ time; throughout the paper, $\delta$ refers to an arbitrarily small positive constant. Our technique can also be used to find a maximum-weight clique of size $3$ in $G(P)$ within the same time complexity. In addition, we show that a maximum-weight independent set or clique of size $2$ can be found in $n^{4/3}2^{O(\log^*n)}$ time. 
All these algorithms also work for computing the minimum-weight independent set or clique. 
We are not aware of any previous work on these weighted problems. As mentioned above, the problem of finding an (unweighted) independent set of size $3$ can be solved in $O(n^{4/3}\log^2n)$ time~\cite{ref:AgarwalCo06}. It is also known that finding an (unweighted) clique of size $3$ in a disk graph (not necessarily unit-disk graph) can be done in $O(n\log n)$ time~\cite{ref:KaplanTr19}. 

\paragraph{The dispersion problem.}
As mentioned above, a related problem is the dispersion problem (also called {\em maximally separated set problem}~\cite{ref:AgarwalCo06}). Given $P$ and a number $k$, the problem is to find a subset of $k$ points from $P$ so that their minimum pairwise distance is maximized. The problem is NP-hard~\cite{ref:WangAs88}. An algorithm for the independent set problem of $G(P)$ can be used as a decision algorithm for the dispersion problem: Given a value $r$, we can decide whether $P$ has a subset of $k$ points whose minimum pairwise distance is larger than $r$ using the independent set algorithm (i.e., by defining an edge for two points in the graph if their distance is at most $r$). 

Under the convex position assumption, Singireddy, Basappa, and Mitchell~\cite{ref:SingireddyAl23} previously gave an $O(n^4k^2)$ time algorithm for the problem. Using our independent set algorithm as a decision procedure and doing binary search among the interpoint distances of $P$, we present a new algorithm that can solve the problem in $O(n^{7/2}\log n)$ time, or in randomized $O(n^{37/11}\log n)$ expected time. For a special case where $k=3$, the algorithm of \cite{ref:AgarwalCo06} solves the problem in $O(n^{4/3}\log^3 n)$ time even if the points of $P$ are not in convex position. Our new algorithm, which works on the convex-position case only, runs in $O(n\log^2 n)$ time. This is achieved using parametric search~\cite{ref:ColeSl87,ref:MegiddoAp83} with our independent set algorithm as a decision algorithm. In addition, with our decision algorithm and Chan's randomized technique~\cite{ref:ChanGe99}, we can obtain a randomized algorithm of $O(n\log n)$ expected time. We note that a recent work~\cite{ref:KobayashiAn22} proposed another algorithm of $O(n^2)$ time, apparently unaware of the result in \cite{ref:AgarwalCo06}.

\subsection{Related Work}
Unit-disk graphs are a fundamental model in wireless networks, particularly where coverage and connectivity are governed by proximity~\cite{ref:ClarkUn90, ref:PerkinsHi94, ref:PerkinsAd99, ref:BalisterCo05}. However, many classical graph problems, including coloring, vertex cover, independent set, and dominating set, remain NP-hard even when restricted to unit-disk graphs~\cite{ref:ClarkUn90}. One exception is that finding a maximum clique in a unit-disk graph can be done in polynomial time~\cite{ref:ClarkUn90,ref:EppsteinGr09,ref:EspenantFi23} and the current best algorithm runs in $O(n^{2.5}\log n)$ time~\cite{ref:EspenantFi23} (see~\cite{ref:KeilTh24} for a comment about improving the runtime to $O(n^{7/3+o(1)})$). 

The assumption that points are in convex position can simplify certain problems that are otherwise NP-hard for general point sets in the plane. 
This has motivated the exploration of other computational problems under similar assumptions. For example, the {\em continuous} $k$-center problem where centers can be anywhere in the plane is NP-hard for arbitrary points but become polynomial time solvable under the convex position assumption \cite{ref:ChoiEf23}. 
The convex position constraint was even considered for classical problems that are already polynomial time solvable in the general case. For instance, the renowned result of Aggarwal, Guibas, Saxe, and Shor~\cite{ref:AggarwalA89} gives a linear time algorithm for computing the Voronoi diagram for a set of planar points in convex position. Refer to \cite{ref:LingasOn86, ref:RichardsRe90, ref:ChazelleIm93} for more work for points in convex position.

The $k$-center problem under a variety of constraints has received much attention. Particularly, when $k$, the number of centers, is two and the centers can be anywhere in the plane (referred to as the {\em continuous $2$-center problem}), several near-linear time algorithms have been developed~\cite{ref:SharirA97,ref:ChanMo99,ref:EppsteinFa97, ref:WangOn22}, culminating in an optimal $O(n\log n)$ time~\cite{ref:ChoOp24}. 
The problem variations under other constraints were also considered. 
For example, the $k$-center problem can be solved in $O(n\log n)$ time if centers are required to lie on the same line~\cite{ref:WangLi16, ref:ChenEf15} or two lines~\cite{ref:BhattacharyaGe15}. 
The continuous one-center problem is the classical smallest enclosing circle problem and can be solved in linear time~\cite{ref:MegiddoLi83}. 

For the convex position case of the continuous $k$-center problem, 
Choi, Lee, and Ahn~\cite{ref:ChoiEf23} proposed an $O(\min\{k, \log n\}\cdot n^2  \log n + k^2n \log n)$ time algorithm. Hence, the worst-case runtime of their algorithm is cubic, while our discrete $k$-center algorithm runs in near quadratic time. 

The discrete $2$-center problem also gets considerable attention. Agarwal, Sharir, and Welzl gave the first subquadratic  $O(n^{4/3} \log^5 n)$ time algorithm~\cite{ref:AgarwalTh98}; the logarithmic factor was slightly improved by Wang~\cite{ref:WangUn23}.
As the continuous two-center problem can be solved in $O(n\log n)$ time~\cite{ref:ChoOp24} while the current best discrete two-center algorithm runs in $\Omega(n^{4/3})$ time~\cite{ref:AgarwalTh98,ref:WangUn23}, the discrete problem appears more challenging than the continuous counterpart. This makes our discrete $k$-center algorithm even more interesting because it is almost a linear factor faster than the continuous $k$-center algorithm in \cite{ref:ChoiEf23}. Therefore, it is an intriguing question whether the algorithm in \cite{ref:ChoiEf23} can be further improved. 

Other variations of the discrete $k$-center problem for small $k$ were recently studied by Chan, He, and Yu~\cite{ref:ChanOn23}, improving over previous results~\cite{ref:BespamyatnikhRe99, ref:BespamyatnikhRec99, ref:KatzDi00}. 

The dispersion problem and some of its variants have also been studied before. The general planar dispersion problem can be solved by an exact algorithm in $n^{O(\sqrt{k})}$ time~\cite{ref:AgarwalCo06}. 
If all points of $P$ lie on a single line, Araki and Nakano~\cite{ref:ArakiMa22} gave an algorithm of $O((2k^2)^k\cdot n)$ time (assuming that the points are not given sorted), which is $O(n)$ for a constant $k$. For a circular case where all points of $P$ lie on a circle and the distance between two points is measured by their distance along the circle, the problem is solvable in $O(n)$ time~\cite{ref:TsaiOp97}, provided that the points are given sorted along the circle. We note that this implies that the line case problem, which can be viewed as a special case of the circular case, is also solvable in $O(n)$ time after the points are sorted on the line. 

\subsection{Our approach}
\label{sec:approach}

The weighted dominating set problem reduces to the following problem: Given any $k$, find a minimum weight dominating set of size at most $k$. This is equivalent to finding a minimum weight subset of at most $k$ points of $P$ such that the union of the unit disks centered at these points covers $P$. Let $S$ be an optimal solution for the problem (points of $S$ are called {\em centers}). If we consider $P$ as a cyclic list of points along the convex hull of $P$, then for each center $p\in S$, 
its unit disk $D_p$ may cover multiple maximal contiguous subsequences (called {\em sublists}) of $P$. Roughly speaking, we prove that it is possible to assign at most two such sublists to each center $p\in S$ such that (1) $p$ belongs to at least one of these sublists; (2) the union of the sublists assigned to all centers is $P$; (3) for every two centers $p_i,p_j\in S$, the sublists of the points assigned to $p_i$ can be separated by a line from the sublists assigned to $p_j$. Using these properties, we further obtain the following structural property (called {\em ordering property}; see Figure~\ref{fig:asgn}) about the optimal solution $S$: There exists an ordering of the centers of $S$ as $p_{i_1},p_{i_2},\ldots,p_{i_k}$ such that (1) $p_{i_1}$ (resp., $p_{i_k}$) is only assigned one sublist; (2) if a center $p_{i_j}$, $1<j < k$, is assigned two sublists, then one of them is on $P_1$, the portion of $P$ from $p_{i_1}$ to $p_{i_k}$ clockwise, and the other is on $P_2$, the portion of $P$ from $p_{i_1}$ to $p_{i_k}$ counterclockwise; (3) the order of the centers of the sublists along $P_1$ (resp., $P_2$) from $p_{i_1}$ to $p_{i_k}$ is a (not necessarily contiguous) subsequence of the above ordering. 

\begin{figure}[t]
\begin{minipage}[h]{\textwidth}
\begin{center}
\includegraphics[height=2.2in]{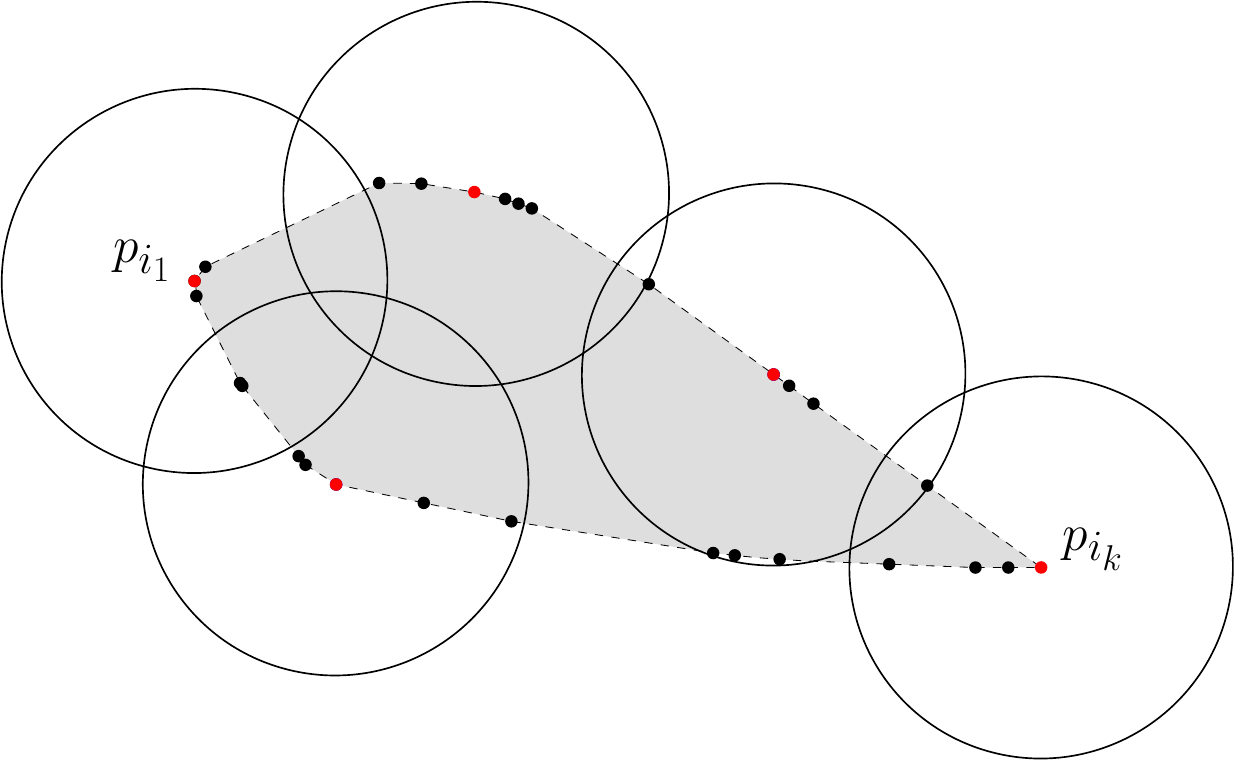}
\caption{\footnotesize Illustrating the ordering property of $S$ (the centers of the disks).}
\label{fig:asgn}
\end{center}
\end{minipage}
\vspace{-0.1in}
\end{figure}

The above ordering property is crucial to the success of our method. 
Using the property, we develop a dynamic programming algorithm for the problem and the runtime is $O(kn^2\log^2 n)$. 
Setting $k=n$ leads to an $O(n^3\log^2 n)$ time algorithm for the original weighted dominating set problem. 

These properties are also applicable to the unweighted case, which is essentially a special case of the weighted problem. Using an additional greedy strategy, the runtime of the algorithm can be improved by roughly a linear factor for the unweighted case. 

To solve the discrete $k$-center problem, as already discussed above, the algorithm for the unweighted dominating set problem can be used to solve the decision problem: Given any value $r$, determine whether $r\geq r^*$, where $r^*$ is the optimal objective value, i.e., the minimum value for which there exist $k$ centers such that the maximum distance from any point of $P$ to its closest center is at most $r^*$. Observe that $r^*$ must be equal to the distance of two points of $P$. As such, by doing binary search on the pairwise distances of points of $P$ and applying the distance selection framework in~\cite{ref:WangIm23} with our unweighted dominating set algorithm, we can compute $r^*$ in $O(n^{4/3}\log n + kn \log^2 n)$ time. Furthermore, using parametric search~\cite{ref:ColeSl87, ref:MegiddoAp83}, we develop another algorithm of $O(k^2 n\log^2 n)$ time, which is faster than the first algorithm when $k = o(n^{1/6}/\sqrt{\log n})$. 

For the independent set problem, our algorithm is a dynamic program, which is in turn based on the observation that the Voronoi diagram of a set of points in convex position forms a tree~\cite{ref:AggarwalA89}. The (unweighted) size-3 case is solved by new observations and developing efficient data structures. As discussed above, we tackle the dispersion problem by using the independent set algorithm as a decision procedure. 
For computing a maximum-weight independent set of size $3$ for points in arbitrary position, our algorithm relies on certain interesting observations and a {\em tree-structured biclique partition} of $P$. Biclique partition has been studied before, e.g., \cite{ref:KatzAn97,ref:WangIm23}. However, to the best of our knowledge, tree-structured biclique partitions have never been introduced before. Our result may find applications elsewhere. 

\paragraph{Outline.} The rest of the paper is organized as follows. After introducing notation in Section~\ref{sec:pre}, we present our algorithms for the dominating set, the discrete $k$-center, the independent set, and the dispersion problems for points in convex position in Sections~\ref{sec:convex}, \ref{sec:discretecenter}, \ref{sec:isconvex} and \ref{sec:dispersion}, respectively. As the only problem for points in arbitrary position studied in this paper, the size-3 weighted independent set problem is discussed in Section~\ref{sec:generalis3weight}.

\section{Preliminaries}
\label{sec:pre}
We introduce some notations that will be used throughout the paper, in addition to those already defined in Section~\ref{sec:intro}, e.g., $P$, $n$, $G(P)$. 

A {\em unit disk} refers to a disk with radius $1$; the boundary of a unit disk is a {\em unit circle}.  For any point $p$ in the plane, we use $D_p$ to denote the unit disk centered at $p$. For any two points $p$ and $q$ in the plane, we use $|pq|$ to denote their (Euclidean) distance and use $\overline{pq}$ to denote the line segment connecting them. Let $\overrightarrow{pq}$ to denote the directed segment from $p$ to $q$. 

For any compact region $R$ in the plane, we use $\partial R$ to denote its boundary and use $\overline{R}$ to denote the complement region of $R$ in the plane. In particular, for a disk $D$ in the plane, $\partial D$ is its bounding circle, and $\overline{D}$ refers to the region of the plane outside $D$. 

Let $\calH(P)$ be the convex hull of $P$. If the points in $P$ are in convex position, then we can consider $P$ as a cyclic sequence. Specifically, let $P = \langle p_1, p_2, \ldots, p_n \rangle$ represent a cyclic list of the points ordered counterclockwise along $\calH(P)$. We use a {\em sublist} to refer to a contiguous subsequence of $P$. Multiple sublists are said to be \textit{consecutive} if their concatenation is also a sublist. 
For any two points $p_i$ and $p_j$ in $P$, we define $P[i,j]$ as the sublist of $P$ from $p_i$ counterclockwise to $p_j$, inclusive, i.e., if $i \leq j$, then $P[i,j] = \langle p_i, p_{i+1}, \ldots, p_j \rangle$; otherwise, $P[i,j] = \langle p_i, p_{i+1}, \ldots, p_n, p_1, \ldots, p_j \rangle$. We also denote by $P(i,j]$ the sublist $P[i,j]$ excluding $p_i$, 
and similarly for other variations, e.g., $P[i,j)$ and $P(i,j)$. 

For simplicity of the discussion, we make a general position assumption that no three points of $P$ are collinear and no four points lie on the same circle. This assumption is made without loss of generality as degenerate cases can be handled through perturbations.


\section{The dominating set problem}
\label{sec:convex}

In this section, we present our algorithms for the dominating set problem on a set $P$ of $n$ points in convex position. 
The weighted and unweighted cases are discussed in Sections~\ref{sec:domwgt} and \ref{sec:domunwgt}, respectively. 
Before presenting these algorithms, we first prove in Section~\ref{sec:structure} the structural properties that our algorithms rely on. 

\subsection{Structural properties}
\label{sec:structure}
In this section, we examine the structural properties of the dominating sets in the unit-disk graph $G(P)$. As discussed in Section~\ref{sec:approach}, the success of our method hinges on these properties. 
We introduce these properties for the weighted case, which are also applicable to the unweighted case. 

Let $\calA $ represent a partition of $P$ into consecutive, nonempty, and disjoint sublists. Suppose $S \subseteq P$ is a dominating set of $G(P)$; points of $S$ are called {\em centers}. It is not difficult to see that the union of the collection $\calD$ of unit disks centered at the points in $S$ covers $P$. 

We say that a collection $\calD$ of unit disks {\em covers} $\calA$ if every sublist $\alpha\in \calA$ is covered by at least one disk from $\calD$. An \textit{assignment} $\phi: \calA \rightarrow S$ is a mapping from sublists in $\calA$ to points in $S$, such that each sublist $\alpha$ is assigned to exactly one center $p_i \in S$ with $\alpha\subseteq D_{p_i}$. 
For each $p_i \in S$, we define $G_{p_i}$ as the set of points in the sublists $\alpha \in \calA$ that are assigned to $p_i$; $G_{p_i}$ is called the {\em group} of $p_i$. Depending on the context, $G_{p_i}$ may also represent the collection of sublists assigned to $p_i$. 
By definition, the groups of two centers of $S$ are disjoint.

An assignment $\phi$ is said to be \textit{line separable} if, for every two groups of $\phi$, there exists a line $\ell$ that separates the points from the two groups, that is, the points of one group lie on one side of $\ell$ or on $\ell$  while those of the other group lie strictly on the other side of $\ell$. 

As discussed in Section~\ref{sec:approach}, our main target is to prove the ordering property. This is achieved by proving a series of lemmas. We start with the following lemma that proves a line separable property. 


\begin{lemma}
    \label{lem:linesep}
    Let $S$ be a dominating set of $G(P)$. There exist a partition $\calA$ of $P$ and a line-separable assignment $\phi: \calA\rightarrow S$ such that for any center $p_i\in S$, $p_i\in G_{p_i}$, meaning that a sublist of $p_i$ contains $p_i$.
\end{lemma}
\begin{proof}
Let the centers of $S$ be $p'_i$, $1\leq i\leq k$, with $k=|S|$. For each center $p_i'\in S$, to simplify the notation, let $D_i$ represent $D_{p_i'}$, i.e., the unit disk centered at $p_i'$. 

Consider two disks $D_i, D_j$ centered at two points $p'_i$ and $p'_j$ of $S$, respectively. Define $\ell$ as the bisector of $p'_i$ and $p'_j$. If $D_i\cap D_j\neq \emptyset$, then let $u_{ij}, v_{ij}$ be the intersections of the boundaries of $D_i$ and $D_j$. 
Without loss of generality, we assume that $\ell$ is vertical and $p'_i$ lies to the left of $p'_j$ (see Figure~\ref{fig:line_sep_cover}).     
Define $D_i^j$ to be the region of $D_i$ strictly to the right of $\ell$, and in particular, no point of the segment $\overline{u_{ij}v_{ij}}$ is in $D_i^j$ (let $D_i^j=\emptyset$ if $D_i\cap D_j=\emptyset$). Define $D_j^i$ symmetrically. Note that $D_j^i\cup D_i^j \cup \overline{u_{ij}v_{ij}} = D_i\cap D_j$ and $D_j^i\cap D_i^j=\emptyset$. 
    For any $p_i\in S$, define $D_i'=D_i - \bigcup_{j=1, j\neq i}^k D_i^j$. 
    
    We first prove the following {\em claim:} $P\subseteq\bigcup_{i=1}^k D_i'$.  
     As $P\subseteq\bigcup_{i=1}^k D_i$, it is sufficient to show $\bigcup_{i=1}^k D_i=\bigcup_{i=1}^k D_i'$. We prove it by induction on $k$. If $k=1$, it is obviously true that $D_1=D_1'$. 

    \begin{figure}[t]
\begin{minipage}[t]{\textwidth}
\begin{center}
\includegraphics[height=1.7in]{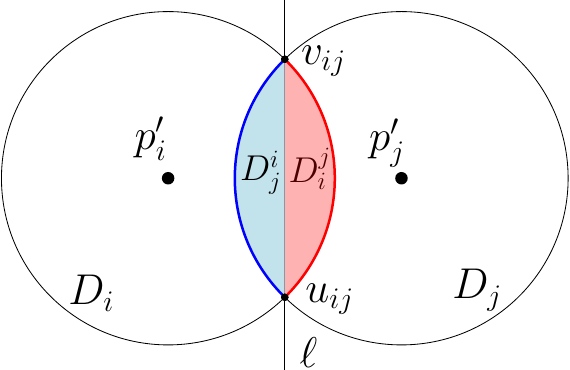}
\caption{\footnotesize Illustrating $D_j^i$ and $D_i^j$, which are the blue and red regions, respectively, excluding $\overline{u_{ij}v_{ij}}$.}
\label{fig:line_sep_cover}
\end{center}
\end{minipage}
\vspace{-0.1in}
\end{figure}
       
       Assume that the statement holds for $k-1$, that is, $\bigcup_{i=1}^{k-1} D_i=\bigcup_{i=1}^{k-1} D_i'$. Note that in the case of $k-1$, $D_i'=D_i - \bigcup_{j=1, j\neq i}^{k-1} D_i^j$. To avoid confusion in the notation, for any $1\leq i\leq k-1$, let $D_i''=D_i - \bigcup_{j=1, j\neq i}^{k-1} D_i^j$. Hence, we have $\bigcup_{i=1}^{k-1} D_i=\bigcup_{i=1}^{k-1} D_i''$ by assumption.
       We still define $D_i'=D_i - \bigcup_{j=1, j\neq i}^k D_i^j$ with respect to $k$. Then, we have $D_i'=D_i''-D_i^k$ for any $1\leq i\leq k-1$. To see this, we derive 
       \begin{align*}
    D_i' &= D_i - \bigcup_{j=1, j\neq i}^{k} D_i^j=(D_i - \bigcup_{j=1, j\neq i}^{k-1} D_i^j)\cap (D_i-D_i^k)=D_i''-\overline{D_i-D_i^k}\\
    & =D_i''-(\overline{D_i}\cup D_i^k)=(D_i''-\overline{D_i})\cap (D_i''-D_i^k).
\end{align*}       
        As $D_i''\subseteq D_i$, $D_i''-\overline{D_i}=D_i''$. 
        Hence, $(D_i''-\overline{D_i})\cap (D_i''-D_i^k)=D_i''\cap (D_i''-D_i^k)=D_i''-D_i^k$. We thus obtain $D_i'=D_i''-D_i^k$.

        
        We now prove that the statement holds for $k$, i.e., $\bigcup_{i=1}^k D_i=\bigcup_{i=1}^k D_i'$.         
        Since $\bigcup_{i=1}^{k-1} D_i=\bigcup_{i=1}^{k-1} D_i''$ by assumption, it suffices to show that $\bigcup_{i=1}^{k-1} D_i'' \cup D_k =\bigcup_{i=1}^k D_i'$. We first make the following two observations.
        \begin{enumerate}
            \item
            $\bigcup_{i=1}^{k-1} D_i'' \cup D_k=\bigcup_{i=1}^{k-1} D_i' \cup D_k$. To see this, for any $1\leq i\leq k-1$, by definition, $D_i^k\subseteq D_k$, and thus $D_i''\cup D_k=(D_i''- D_i^k)\cup D_k=D_i'\cup D_k$.                  
            Applying this iteratively, we obtain $\bigcup_{j=1}^{k-1} D_i'' \cup D_k=\bigcup_{j=1}^{k-2} D_i'' \cup(D_{k-1}''\cup D_k)=\bigcup_{j=1}^{k-2} D_i''\cup D_k\cup D_{k-1}'=\dots =\bigcup_{j=1}^{k-1} D_i' \cup D_k$.

            \item
            $\bigcup_{i=1}^{k-1} D_k^i\subseteq\bigcup_{i=1}^{k-1}D_i'$.
            Assume to the contrary that this is not true. Then there exists a point $q$ such that $q\in \bigcup_{i=1}^{k-1} D_k^i$ but $q\notin \bigcup_{i=1}^{k-1} D_i'$. For any region $D_i^j$, by definition, $|q'p'_{j}|<|q'p'_{i}|$ holds for any point $q'\in D_{i}^{j}$. Let $p'_{j}$ be the closest point to $q$ among $\{p'_1,p'_2,\ldots, p'_{k-1}\}$. As $q\in \bigcup_{i=1}^{k-1} D_k^i$, $q\in D_k^t$ for some $t\in [1,k-1]$. Hence, $|qp_t|<|qp_k|$. As $|qp_j|\leq |qp_t|$, we have $|qp_j|<  |qp_k|$. Therefore, $|qp_j|\leq |qp_i|$ holds for all $i\in [1,k]$. 

            On the other hand, since $q\in D_k^t$, we have $q\in D_t$ and $|qp_t|\leq 1$. Since $|qp_j|\leq |qp_t|$, it holds that $|qp_j|\leq 1$ and thus $q\in D_j$. Since $q\notin \bigcup_{i=1}^{k-1} D_i'$, we have $q\notin D_j'$. Hence, it follows that $q\in\bigcup_{i=1,i\neq j}^{k} D_j^i $. Therefore, $q\in D_j^{l}$ for some $l$ with $l\in [1,k]$ and $l\neq j$, and thus $|qp_l|<|qp_j|$. However, this incurs a contradiction since $|qp_j|\leq |qp_i|$ for all $i\in [1,k]$. 
        \end{enumerate}
        In light of the above observations, we can derive 
        \begin{align*}
         \bigcup_{i=1}^{k-1} D_i'' \cup D_k & =\bigcup_{i=1}^{k-1} D_i' \cup D_k= \bigcup_{i=1}^{k-1} D_i' \cup \left(\bigcup_{i=1}^{k-1} D_k^i\cup (D_k-\bigcup_{i=1}^{k-1} D_k^i)\right)\\
         &=\bigcup_{i=1}^{k-1} D_i' \cup (D_k-\bigcup_{i=1}^{k-1} D_k^i)=\bigcup_{i=1}^{k-1} D_i' \cup D_k'=\bigcup_{i=1}^{k} D_i'.   
        \end{align*}
                
        This proves the claim that $P\subseteq\bigcup_{i=1}^k D_i'$.         
              
 
    Based on the claim, we now construct a line separable assignment $\phi$ as follows. Starting from set $D_1'$, assign to $G_{p_1'}$ the maximal sublists of points of $P$ that are inside $D_1'$, and then remove these points from further consideration. Next, for $D_2'$, assign to $G_{p_2'}$ the maximal sublists of $P$ consisting of unassigned points that are inside $D_2'$. Repeat this process, until all points of $P$ are assigned. We argue that resulting groups are pairwise line separable, i.e., for any two groups $G_{p_i'}$ and $G_{p_j'}$ with $i<j$, they can be separated by a line.  

    Consider the bisector $\ell$ of $p'_i$ and $p'_j$, the centers of $D_i$ and $D_j$. We argue that $G_{p_i'}$ and $G_{p_j'}$ can be separated by $\ell$. Indeed, by definition, $D_i'$ and $D_j'$ lie on the two different sides of $\ell$. Since $i<j$, $G_{p_i'}$ is constructed earlier than $G_{p_j'}$. Therefore, if there are points of $P$ on $\ell$ inside both $D_i'$ and $D_j'$, all these points will not be assigned to $G_{p_j'}$. Since $D_i'$ and $D_j'$ lie on the two different sides of $\ell$, we obtain that $G_{p_i'}$ and $G_{p_j'}$ are separated by $\ell$. 

    We finally show that for each $p_i'\in S$, the group $G_{p_i'}$ contains $p_i'$. Indeed, for any two centers $p'_i, p'_j \in S$, by definition, $p_i'$ cannot be in $D_i^j\cup \overline{u_{ij}v_{ij}}$. Therefore, $p_i'$ must be in $D_i'$ and cannot be in $D_j'$ for any $j\neq i$. According to our construction of $\phi$, $p_i'$ cannot be assigned to $G_{p_j'}$ for any $j\neq i$ because only points in $D_{p_j'}$ can be assigned to $G_{p_j'}$. Therefore, $p_i'$ must be in $G_{p_i'}$. This proves the lemma. 
\end{proof}

For the assignment $\phi$ from Lemma~\ref{lem:linesep}, for each center $p_i\in S$, we refer to the sublist of $p_i$ that contains $p_i$ as the {\em main sublist} of $p_i$ while all other sublists (if any) are called {\em secondary sublists} of $p_i$. 

\begin{lemma}
\label{lem:balassign}
    Let $S$ be a dominating set for $G(P)$. Then there exist a partition $\calA$ of $P$ and an assignment $\phi: \calA\rightarrow S$ with the following properties: 
    \begin{enumerate}
        \item $\phi$ is line separable.
        \item Each center of $S$ is assigned at most two sublists, one of which is a main sublist.
        \item For any center $p_i\in S$ that has a secondary sublist $\beta_i$, there exists a point $p_t\in \beta_i$ such that $p_t\not\in D_{q_1}$ and $p_t\not\in D_{q_2}$ for some centers $q_1$ and $q_2$ of $S$ with $q_1\in P(i,t)$ and $q_2\in P(t,i)$. 
    \end{enumerate}    
\end{lemma}
\begin{proof}
    Let $\phi$ be the assignment obtained from Lemma~\ref{lem:linesep}. In the following, we show that $\phi$ can be adjusted so that the group $G_{p_i}$ of each center $p_i\in S$ has at most two sublists and one of them must contain $p_i$. 

    For each center $p_i \in S$, we concatenate the consecutive sublists of $\calA$ that assigned to the group $G_{p_i}$. As a result, every group will consist of only nonconsecutive sublists.
    
    Consider a center $p_i\in S$. If $G_{p_i}$ has at most two sublists, then we do not need to do anything for $G_{p_i}$. Otherwise, $G_{p_i}$ has at least three sublists, and more specifically, $G_{p_i}$ has one main sublist and at least two secondary sublists. By the line separable property of $\phi$ and due to the convexity of $P$, for any other group $G_{p_j}$ that contains a sublist between two sublists of $G_{p_i}$, its center $p_j$ is also between these sublists of $G_{p_i}$. The converse is also true: for any center $p_j\in S$ lying between two sublists of $G_{p_i}$, all sublists of $G_{p_j}$ are between these sublists.
    

    Consider two secondary sublists $\beta_1', \beta_2'$ of $G_{p_i}$. Let $\alpha_i$ be the main sublist of $G_{p_i}$. Removing $\beta_1'$, $\beta_2'$, and $\alpha_i$ from $P$ will leave three {\em gaps} in the cyclic list of $P$.  
    As sublists of $G_{p_i}$ are disjoint, there is at least one point of $P$ in each gap. Since $S$ is a dominating set, due to the line separable property of $\phi$ and the convexity of $P$, there is at least one center of $S$ in each gap. Let $q_1$ be a point of $S$ lying between $\alpha_i$ and $\beta_1'$. Let $q_2, q_3$ be centers of $S$ lying between $\beta_1'$ and $\beta_2'$; if there is only one such center exists, then let $q_2$ and $q_3$ refer to the same point. Let    
    $q_4$ be a center of $S$ lying between $\beta_2'$ and $\alpha'$. Without loss of generality, we assume that $\alpha_i$, $q_1$, $\beta_1'$, $q_2$, $q_3$, $\beta_2'$, and $q_4$ are ordered counterclockwise along $P$ (see Figure~\ref{fig:listpoint}).

\begin{figure}[t]
\begin{minipage}[t]{\textwidth}
\begin{center}
\includegraphics[height=1.7in]{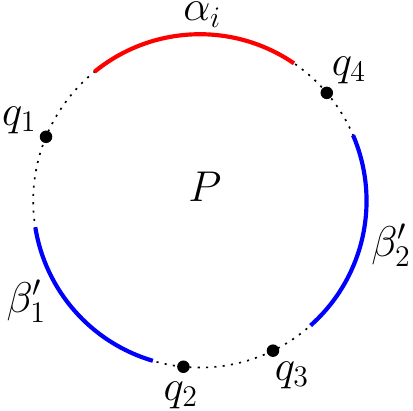}
\caption{\footnotesize Illustrating a schematic view of $P$ (i.e., the dotted circle) and the relative positions of $\alpha_i$, $q_1$, $\beta_1'$, $q_2$, $q_3$, $\beta_2'$, and $q_4$. The sublists $\alpha_i$, $\beta_1'$, and $\beta_2'$ are illustrated with solid arcs.}
\label{fig:listpoint}
\end{center}
\end{minipage}
\vspace{-0.1in}
\end{figure}
    
    We claim that $\beta_1'\subseteq D_{q_1}\cup D_{q_2}$ or $\beta_2'\subseteq D_{q_3}\cup D_{q_4}$. Assume to the contrary that this is not true. Then, there exist a point $c_1\in \beta_1'$ and a point $c_2\in \beta_2'$ such that the distances $|c_1q_1|,|c_1q_2|,|c_2q_3|,|c_2q_4|$ are all greater than $1$; see Figure~\ref{fig:balassign}. Define angles (see Figure~\ref{fig:balassign}): $\xi_1=\angle c_1q_1p_i$, $\zeta_1=\angle q_2c_1q_1$, $\gamma_1=\angle c_1q_2p_i$, $\gamma_2=\angle p_iq_3c_2$, $\zeta_2=\angle q_3c_2q_4$, $\xi_2=\angle c_2q_4p_i$, $\theta_1=\angle q_1p_ic_1$, $\theta_2=\angle c_1p_iq_2$, $\theta_3=\angle q_3p_ic_2$, $\theta_4=\angle c_2p_iq_4$, and $\theta=\angle q_1p_iq_4$. Obviously, 
    \begin{equation}
    \label{eq:angles}
    \xi_1+\zeta_1+\gamma_1+\gamma_2+\zeta_2+\xi_2+\theta_1+\theta_2+\theta_3+\theta_4=4\cdot 180\degree
    \end{equation}
    Consider the triangle $\triangle q_1c_1p_i$. As $c_1\in \beta_1'\subseteq D_{p_i}$, we obtain that $|c_1p_i|\leq 1$. Since 
    $|q_1c_1|>1$, we have $|c_1p_i|<|q_1c_1|$ and thus $\xi_1<\theta_1$. Using the same argument, we can derive $\gamma_1< \theta_2$, $\gamma_2< \theta_3$, $\xi_2< \theta_4$. Consequently, by \eqref{eq:angles}, we have
    $$\zeta_1+\zeta_2+2\theta_1+2\theta_2+2\theta_3+2\theta_4>4\cdot180\degree$$
    Notice that $\theta_1+\theta_2+\theta_3+\theta_4\leq \theta$ (the equality occurs when $q_2=q_3$), where $\theta\leq 180\degree$ due to the convexity of $P$. Therefore:
    $$\zeta_1+\zeta_2+2\cdot180\degree\geq \zeta_1+\zeta_2+2\theta> 4\cdot180\degree$$
    Thus, we obtain $\zeta_1+\zeta_2> 360\degree$, which means that at least one of the angles is greater than $180\degree$. This incurrs a contradiction since $q_1,c_1,q_2,q_3,c_2,q_4, p_i$ are in convex position. Therefore, the claim is proved. In what follows, we adjust $\phi$ by using the claim. 

        \begin{figure}[t]
\begin{minipage}[t]{0.5\textwidth}
\begin{center}
\includegraphics[height=1.6in]{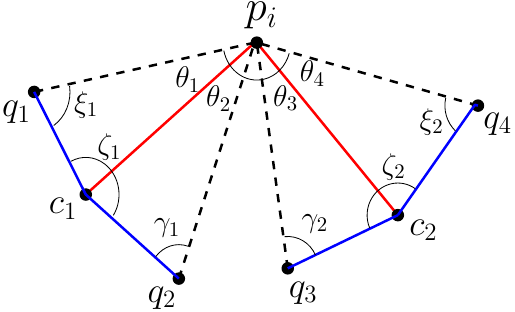}
\caption{\footnotesize  Illustrating angles for Lemma~\ref{lem:balassign}. 
         Solid red segments have lengths at most $1$ while
         solid blue segments have lengths greater than $1$.}
\label{fig:balassign}
\end{center}
\end{minipage}
\hspace{0.05in}
\begin{minipage}[t]{0.48\textwidth}
\begin{center}
\includegraphics[height=1.6in]{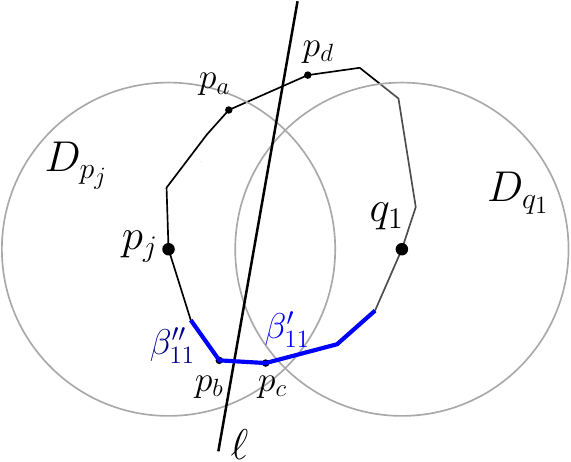}
\caption{\footnotesize $\beta'_{11}$ consists of the points in the blue part, and its portion left of $\ell$ is $\beta''_{11}$.}
\label{fig:balassign2}
\end{center}
\end{minipage}
\vspace{-0.15in}
\end{figure}

    In the above discussion, the centers $q_1, q_2, q_3, q_4\in S$ were chosen arbitrarily from the corresponding gaps between $\alpha_i$, $\beta_1'$, and $\beta_2'$. We now choose some particular centers. Let $q_1$ be the center of $S$ such that its group $G_{q_1}$ contains the sublist of $\calA$ immediately following $\beta_1'$ clockwise, $q_2$ the center such that its group $G_{q_2}$ contains the sublist of $\calA$ immediately following $\beta_1'$ counterclockwise, $q_3$ the center such that its group  $G_{q_3}$ contains the sublist of $\calA$ immediately following $\beta_2'$ clockwise, and $q_4$ the center such that its group $G_{q_4}$ contains the sublist of $\calA$ immediately following $\beta_2'$ counterclockwise. 

    According to the above claim, $\beta_1'\subseteq D_{q_1}\cup D_{q_2}$ or $\beta_2'\subseteq D_{q_3}\cup D_{q_4}$. Without loss of generality, we assume $\beta_1'\subseteq D_{q_1}\cup D_{q_2}$. Let $q_1=p_{n_1}$ and $q_2=p_{n_2}$. By definition, $\beta_1'\subseteq P[n_1,n_2]$. Since points of $P[n_1,n_2]$ are in convex position, $\beta_1'$ is split into at most two sublists $\beta'_{11}, \beta'_{12}$ that are separated by the bisector $\ell$ of $q_1$ and $q_2$. Without loss of generality, we assume that $\beta'_{11}$ is in the side of $\ell$ that contains $q_1$ and $\beta'_{12}$ is in the side of $\ell$ that contains $q_2$.
    By definition, $G_{q_1}$ contains a sublist $\alpha'_1$ adjacent to $\beta_1'$ and thus to $\beta'_{11}$. We concatenate $\alpha'_1$ with $\beta'_{11}$ to obtain a new sublist for $G_{q_1}$. Similarly, by definition, $G_{q_2}$ contains a sublist $\alpha'_2$ adjacent to $\beta_2'$ and thus to $\beta'_{12}$. We concatenate $\alpha'_2$ with $\beta'_{12}$ to obtain a new sublist for $G_{q_2}$. The sublist $\beta_1'$ is then removed from $G_{p_i}$. We refer to this as a {\em reassignment procedure}. Note that the procedure does not change the number of sublists of $G_{q_1}$ or the number of sublists of $G_{q_2}$.
    
    We argue that after the reassignment procedure the line separable property and the main sublist property of the groups (i.e., each group contains a main sublist) still hold. First of all, as the procedure only reassigned a secondary sublist of $G_{p_i}$, the main sublist property trivially holds. Next we argue that the line separable property is also preserved. 
    It suffices to show that the new group $G_{q_1}$ (resp., $G_{q_2}$) can be separated from every other group by a line. We only discuss $G_{q_1}$ since the case of $G_{q_2}$ can be argued analogously. 
    
    The only change of the reassignment is to reassign $\beta_1'$ to $q_1$ and $q_2$. Due to the convexity, $G_{q_1}$ can still be separated from $G_{p_i}$ and also from groups whose centers are in the gap between $\beta_1'$ and $\beta_2'$ as well as the gap between $\alpha_i$ and $\beta_2'$. What remains to show is that the new $G_{q_1}$ can be separated from groups whose centers are in the same gap as $q_1$, i.e., the gap between $\alpha_i$ and $\beta_1'$.  Consider any other group $G_{p_j}$ in that gap. The original $G_{q_1}$ is separable from $G_{p_j}$ by some line $\ell$ (see Figure~\ref{fig:balassign2}). Due to the convexity, the line splits $P$ into two sublists: $P[a, b]$ and $P[c,d]$ for some points $p_a,p_b,p_c,p_d$. 
    Without loss of generality, we assume that $G_{p_j}\subseteq P[a,b]$ and $G_{q_1}\subseteq P[c,d]$ (see Figure~\ref{fig:balassign2}). Assume that after the reassignment, $\ell$ does not separate $G_{p_j}$ and $G_{q_1}$, which is due to the fact that a portion of $\beta'_{11}$ (denoted by $\beta_{11}''$) lies in the same side of $\ell$ as $G_{p_j}$, i.e., is in $P[a,b]$. Since $\beta_{11}''$ does not contain any point of $G_{p_j}$, $G_{p_j}\subseteq P[a,b]\setminus \beta''_{11}$. Due to the convexity, $G_{p_j}$ and new $G_{q_1}$ can still be separated by a line, e.g., the line that contains the two endpoints of $P[a,b]\setminus \beta''_{11}$.


\paragraph{Handling other secondary sublists of $p_i$.}
    For the center $p_i\in S$, we can repeat the above procedure, each time reducing the number of secondary sublists of $p_i$ by one until at most one secondary sublist remains. Thus, the number of sublists of $p_i$ can be reduced to two while the number of sublists of every other center does not change. 

\paragraph{Handling other centers.} 
    By applying the algorithm to each center $p_i\in S$ with more than two sublists, we will eventually reduce the number of sublists of each center to at most two —  one main sublist and at most one secondary sublist, while ensuring that the groups remain line separable and each group contains a main sublist. The current assignment $\phi$ satisfies the first two properties in the lemma. 
    
\paragraph{Proving the third property of the lemma.}
    To make it satisfy the third property, we just need to apply the reassignment procedure on each remaining secondary arc. Specifically, for each center $p_i\in S$ that has a secondary arc $\beta_i$, we do the following. Let $q_1$ be the center of $S$ such that its group $G_{q_1}$ contains the sublist of $\calA$ immediately following $\beta_i$ clockwise, and $q_2$ the center such that its group $G_{q_2}$ contains the sublist of $\calA$ immediately following $\beta_i$ counterclockwise. If $\beta_i\subseteq D_{q_1}\cup D_{q_2}$, then we can apply the reassignment procedure on $\beta_i$, after which $p_i$ does not have a secondary sublist anymore; otherwise, $\beta_i$ has a point $p_t$ such that $p_t\not\in D_{q_1}\cup D_{q_2}$. By definition, one of $q_1$ and $q_2$ is in $P(i,t)$ and the other is in $P(t,i)$. As such, the third property of the lemma holds.        
\end{proof}

For any center $p_i$ in the assignment of Lemma~\ref{lem:balassign}, we often use $\alpha_i$ to denote its main sublist and use $\beta_i$ to denote its secondary sublist if it has one. 
    
\begin{figure}[t]
\begin{minipage}[t]{\textwidth}
\begin{center}
\includegraphics[height=1.7in]{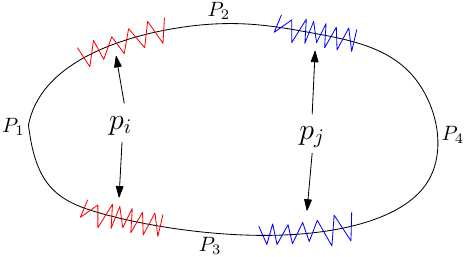}
\caption{\footnotesize Illustrating Lemma~\ref{lem:3disks}. The cycle represents the cyclic list of $P$. The two red (resp., blue) crossed portions are sublists of $p_i$ (resp., $p_j$). Removing these sublists leave the cyclic list of $P$ into four portions $P_1,P_2,P_3,P_4$. 
For any other center of $S$ whose group has two sublists, only the following three cases are possible: (1) both sublists are in $P_1$; (2) both sublists are in $P_4$; (3) one sublist is in $P_2$ while the other sublist is in $P_3$.}
\label{fig:groups}
\end{center}
\end{minipage}
\vspace{-0.1in}
\end{figure}

\begin{lemma}
\label{lem:3disks}
    Let $S$ be an optimal dominating set and $\phi:\calA\rightarrow S$ be the assignment from Lemma~\ref{lem:balassign}. For any two centers $p_i, p_j \in S$ whose groups have two sublists, removing these four sublists from $P$ resulting in four portions of the cyclic list of $P$: one portion is bounded by the two sublists of $p_i$, one portion is bounded by the two sublists of $p_j$, and two portions each of which is bounded by a sublist of $p_i$ and a sublist of $p_j$. 
    Then, for any other center of $S$ whose group has two sublists, either both sublists of the group are in one of the first two portions, or the two sublists are in the last two portions, respectively (see Figure~\ref{fig:groups}). 
\end{lemma}
\begin{proof}
    Assume to the contrary that the lemma statement is not true. Then,  
    due to the line separable property, both sublists of a center $p_l\in S$ are in one of the last two portions of $P$ (i.e., in either $P_2$ or $P_3$ in Figure~\ref{fig:groups}).

    Since the groups $G_{p_i}$, $G_{p_j}$, and $G_{p_l}$ all have two sublists,       
    by the third property of Lemma~\ref{lem:balassign}, the secondary sublist $\beta_i$ of $p_i$ has a point $p_a\in P$ such that $P(a, i)$ has a center whose distance from $p_a$ is greater than $1$.  
    Similarly, $\beta_l$ has a point $p_b\in P$ such that $P(b, l)$ has a center whose distance from $p_b$ is greater than $1$, and $\beta_j$ has a point $p_c$ such that $P(j, c)$ has a center whose distance from $p_j$ is greater than $1$. 
    We assume that $p_i,p_b,p_l,p_j,p_c, p_a$ are in the counterclockwise order along $P$ (see Figure~\ref{fig:3disks}); other cases can be argued analogously.

\begin{figure}[t]
\begin{minipage}[t]{\textwidth}
\begin{center}
\includegraphics[height=2.5in]{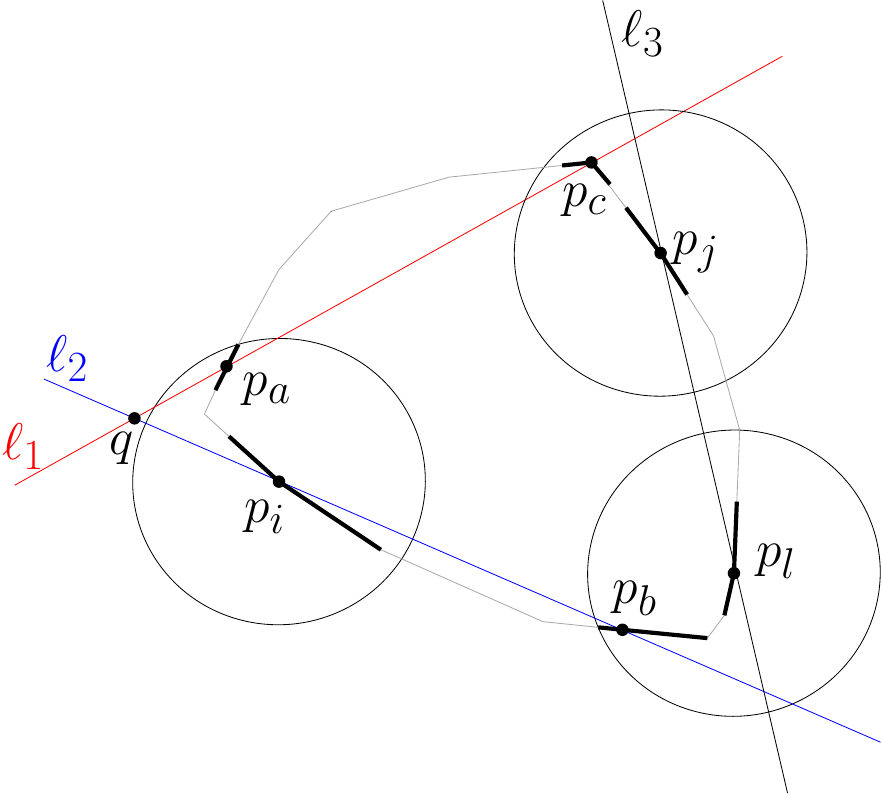}
\caption{\footnotesize Illustrating the centers $p_i, p_l, p_j$, and points $p_a,p_b,p_c$ in their secondary sublists, respectively. The black segments represents sublists of the three centers.}
\label{fig:3disks}
\end{center}
\end{minipage}
\vspace{-0.15in}
\end{figure}
    
    As $S$ is an optimal dominating set, we claim that $P(a, i)$ has at least one point that is outside the disk $D_{p_i}$. Indeed, if this is not true, then $D_{p_i}$ covers all points of $P(a,i)$. As not all points of $P(a,i)$ are assigned to $G_{p_i}$ in $\phi$, $S$ must have at least one center $p_g$ whose group $G_{p_g}$ contains at least a point of $P(a,i)$. Furthermore, due to the line separable property, all points of $G_{p_g}$ including $p_g$ must be in $P(a,i)$, since otherwise $G_g$ cannot be separated from $G_{p_i}$ by a line due to the convexity. But since $D_{p_i}$ covers all points of $P(a,i)$, $S- \{p_g\}$ still forms a dominating set, which contradicts with the optimality of $S$. Therefore, $P(a, i)$ has at least one point that is outside $D_{p_i}$. 
    
    Similarly, $P(b, j)$ has at least one point that is outside $D_{p_j}$ and $P(j, c)$ has at least one point that is outside $D_{p_l}$. 

    Let $\ell_1$ be the line through $p_a$ and $p_c$, $\ell_2$ the line through $p_i$ and $p_b$, and $\ell_3$ the line through $p_j$ and $p_l$ (see Figure~\ref{fig:3disks}). Let $q$ be the intersection of $\ell_1$ and $\ell_2$. Due to the convexity, all points of $P(a,i)$ are in the interior of the triangle $\triangle p_ap_iq$. We claim that the angle $\angle p_aqp_i$ is less than $60\degree$. Indeed, as $p_a\in \beta_i\subseteq D_{p_i}$, we have $|p_ap_i|\leq 1$. As discussed above, $P(a,i)$ has a point $p$ such that $|p_ap|>1$. As $p\in P(a,i)\subseteq \triangle p_ap_iq$, it holds that $|p_ap|<\max\{|p_ap_i|,|p_aq|\}$. Since $|p_ap|>1$ and $|p_ap_i|\leq 1$, we obtain $|p_aq|>1$. In addition, we have proved above that $P(a,i)$ has a point $p'$ outside $D_{p_i}$, meaning that $|p'p_i|>1$. Following a similar argument, we can obtain $|p_iq|>1$. Since  $|p_aq|>1$, $|p_iq|>1$, and $|p_ip_a|\leq 1$, it follows that $\angle p_aqp_i<60\degree$. 
    
    Similarly, the angle of the wedge formed by $\ell_2$ and $\ell_3$ and containing $P(b, l)$ is smaller than $60\degree$, so is  the angle of the wedge formed by $\ell_3$ and $\ell_1$ and containing $P(j, c)$. However, this leads to a contradiction, as the sum of the above three angles is equal to $180\degree$. The lemma is thus proved. 
\end{proof}



\begin{lemma}
    \label{lem:corners}
    Let $S$ be an optimal dominating set and $\phi:\calA\rightarrow S$ be the assignment given by Lemma~\ref{lem:balassign}. There exists a pair of centers $(p_i,p_j)$ in $S$, called a {\em decoupling pair}, such that the following hold: (1) each of $p_i$ and $p_j$ has only one sublist; (2) for any center $S$ that has two sublists, one sublist is in $P(i,j)$ while the other is in $P(j,i)$.
\end{lemma}



\begin{proof}
    We assume that there is at least one center $p_{i_1}\in S$ that has two sublists; otherwise every two centers of $S$ form a decoupling pair. Hence, the group $G_{i_1}$ contains a main sublist $\alpha_{i_1}$ and a secondary sublist $\beta_{i_1}$. 
    Let $\ell_1$ be a line connecting an endpoint of $\alpha_{i_1}$ and an endpoint of $\beta_{i_1}$, and $\ell_2$ a line connecting the other two endpoints of $\alpha_{i_1}$ and $\beta_{i_1}$, in such a way that both sublists $\alpha_{i_1}$ and $\beta_{i_1}$ are in the same wedge formed by $\ell_1$ and $\ell_2$ (see Figure~\ref{fig:corner}). Without loss of generality, we assume that both sublists are in the right side of $\ell_1$ and in the left side of $\ell_2$. 
    Let $P_L$ represent the sublist of $P$ to the left of $\ell_1$, and $P_R$ the sublist of $P$ to the right of $\ell_2$ (Figure~\ref{fig:corner}). In the following, we argue that there is a decoupling pair $(p_i,p_j)$ with $p_i\in P_L$ and $p_j\in P_R$. We first show how to identify $p_i$ in $P_L$.

\begin{figure}[t]
\begin{minipage}[h]{\textwidth}
\begin{center}
\includegraphics[height=1.7in]{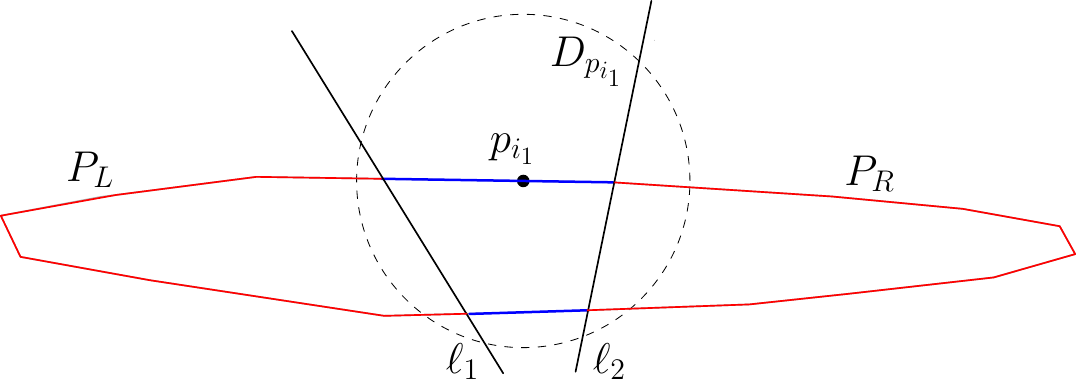}
\caption{\footnotesize Illustrating the proof for lemma~\ref{lem:corners}; blue segments represent the two sublists of $p_{i_1}$.}
\label{fig:corner}
\end{center}
\end{minipage}
\vspace{-0.1in}
\end{figure}
    
Note that $P_L$ should contain at least one center of $S$. Indeed, consider a sublist of $\phi$ that is in $P_L$. Let $p$ be the center whose group contains the sublist. Due to the line separable property of $\phi$, all sublists of $p$ must be in $P_L$. In particular, the main sublist of $p$ is in $P_L$. Thus, $p\in P_L$. On the other hand, for every center $p\in S$ that is in $P_L$, again due to the line separable property, all sublists of $p$ must be in $P_L$. The same observations hold for $P_R$ as well. In addition, for any center $p_{i'}\in S$ in $P_L$ and any center $p_{j'}\in S$ in $P_R$, one of the two sublists of $p_{i_1}$ is in $P(i',j')$ and the other is in  $P(j',i')$.
 
If $P_L$ has only one center, then let $p_i$ be that center and $D_{p_i}$ must cover $P_L$. As such, $P_L$ is the only sublist of $p_i$.
If $P_L$ contains a center $p_{i_2}\in S$ that has two sublists, then as discussed above, both sublists must be in $P_L$. As such, we repeat the process described above to shrink $P_L$: The new $P_L$ is the connected portion of old $P_L$ lying between the two sublists of $G_{i_2}$. We repeat this process until every center in $P_L$ has a single sublist. Then, we choose an arbitrary center in $P_L$ as $p_i$. We apply the same process to $P_R$ such that every center in $P_R$ has a single sublist. We choose an arbitrary center in $P_R$ as $p_j$. In what follows, we argue that $(p_i,p_j)$ is decoupling pair. 

\begin{figure}[t]
\begin{minipage}[h]{\textwidth}
\begin{center}
\includegraphics[height=1.7in]{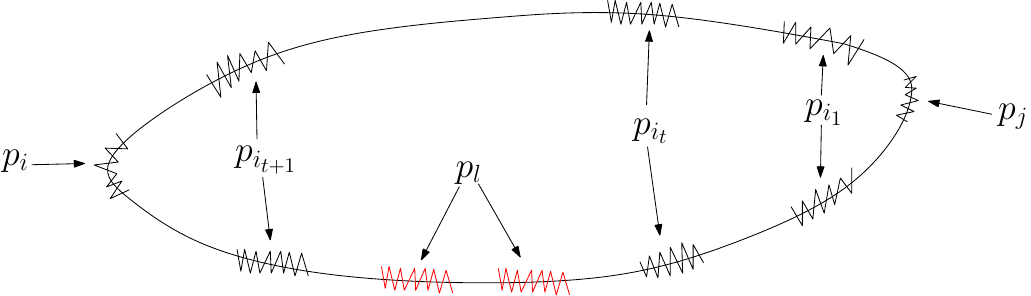}
\caption{\footnotesize Illustrating the relative positions of the sublists of $p_l$, $p_{i_t}$, and $p_{i_{t+1}}$. The cycle represents $P$ and the crossed portions are sublists.}
\label{fig:jam}
\end{center}
\end{minipage}
\vspace{-0.1in}
\end{figure}

Consider a center $p_l\in S$ with two sublists. Our goal is to show that one sublist of $p_l$ is in $P(i,j)$ and the other is in $P(j,i)$. 
If $p_l$ is already considered during our construction for $P_L$ and $P_R$, i.e., $p_l$ is one of the centers like $p_{i_2}$ discussed above, then according to our construction, it is true that one sublist of $p_l$ is in $P(i,j)$ and the other is in $P(j,i)$.
We now consider the case where $p_l$ is not considered during the construction. 
Assume to the contrary that the above statement is not true. Then, the sublists of $p_l$ are either both in $P(i,j)$ or both in $P(j,i)$. 
Without loss of generality, we assume that both sublists of $p_l$ are in $P(i, j)$. By our construction, every center in $P_L\cup P_R$ has only one sublist. Therefore, the two sublists of $p_l$ belongs to the portion of $P(i,j)$ jammed between two centers $p_{i_t}$ and $p_{i_{t+1}}$ that have been considered in our construction (see Figure~\ref{fig:jam}). As discussed above, each of the two centers $p_{i_t}$ and $p_{i_{t+1}}$ has two sublists, one is in $P(i,j)$ and the other in $P(j,i)$. As such, we have the following situation: the two sublists of $p_l$ are between the sublist of $p_{i_t}$ and the sublist of $p_{i_{t+1}}$ on $P(i,j)$. According to Lemma~\ref{lem:3disks}, this is not possible. 

This proves that $(p_i,p_j)$ is decoupling pair. The lemma thus follows. 
\end{proof}

Let $S$ be an optimal dominating set and $\phi:\calA\rightarrow S$ be the assignment given by Lemma~\ref{lem:balassign}.
Let $(p_i,p_j)$ be a decoupling pair from Lemma~\ref{lem:corners}. For any center of $S\setminus\{p_i,p_j\}$, each of its sublist  must be either entirely in $P(i,j)$ or in $P(j,i)$, and
by Lemma~\ref{lem:corners}, the center has at most one sublist in $P(i,j)$ and at most one sublist in $P(j,i)$. 
We finally prove in the following lemma the ordering property discussed in Section~\ref{sec:approach}.


\begin{lemma}\label{lem:ordering}
{\em (The ordering property)}
Let $S$ be an optimal dominating set and $\phi:\calA\rightarrow S$ be the assignment given by Lemma~\ref{lem:balassign}. Let $(p_{i},p_{j})$ be a decoupling pair from Lemma~\ref{lem:corners}. Then, there exists an ordering of all centers of $S$ as $p_{i_1},p_{i_2},\ldots,p_{i_k}$ with $k=|S|$ such that (see Figure~\ref{fig:asgn})
\begin{enumerate}
    \item $p_{i}=p_{i_1}$ and $p_{j}=p_{i_k}$, i.e., $p_{i_1}$ and $p_{i_k}$ are the first and last points in the ordering, respectively.
    \item The sequence of the centers of $S$ that have at least one sublist in $P[i,j]$ (resp., $P[j,i]$) ordered by the points of their sublists appearing in $P[i,j]$ (resp., $P[j,i]$) from $p_i$ to $p_j$ is a  (not necessarily contiguous) subsequence of the ordering.
    \item For any $t$, $1\leq t\leq k$, the sublists of the first $t$ centers in the ordering are consecutive (i.e., their union, which is $\bigcup_{l=1}^{t}G_{i_l}$, is a sublist of $P$).
    \item For any $t$, $2\leq t\leq k$, $\bigcup_{l=1}^{t-1}G_{i_l}\subseteq \bigcup_{l=1}^{t}G_{i_l}$. 
    \item For any $t$, $1\leq t\leq k$, each sublist of $p_{i_t}$ appears at one end of $\bigcup_{l=1}^{t}G_{i_l}$ (if $t=k$, then $\bigcup_{l=1}^{t}G_{i_l}$ becomes the cyclic list of $P$; for convenience, we view $\bigcup_{l=1}^{t}G_{i_l}$ as a list by cutting it right after the clockwise endpoint of $G_{i_k}$). 
\end{enumerate}
\end{lemma}
\begin{proof}
First of all, notice that the first two properties imply the last three. Therefore, it suffices to prove the first two properties.     

Let $S_1$ (resp., $S_2$) be the subset of centers of $S\setminus\{p_i,p_j\}$ that have a sublist in $P(i,j)$ (resp., $P(j,i)$). 
By Lemma~\ref{lem:corners}, each center of $S_1$ has at most one sublist in $P(i,j)$ and each center of $S_2$ has at most one sublist in $P(j,i)$. We add $p_i$ and $p_j$ to both $S_1$ and $S_2$. We sort all centers of $S_1$ as a sequence (called {\em the sorted sequence} of $S_1$) by the points of their sublists appearing in $P[i,j]$ from $p_i$ to $p_j$ (and thus $p_i$ is the first center and $p_j$ is the last one in the sequence). Similarly, we sort all centers of $S_2$ as a sequence (called {\em the sorted sequence} of $S_2$) by the points of their sublists appearing in $P[j,i]$ from $p_i$ to $p_j$ (and thus $p_i$ is the first center and $p_j$ is the last one in the sequence). 
To prove the first two properties of the lemma, it suffices to show the following {\em statement}: There exists an ordering of $S$ such that (1) $p_i$ is the first one in the ordering and $p_j$ is the last one; (2) the sorted sequence of $S_1$ (resp., $S_2$) is a subsequence of the ordering.

We say that two centers $p_{j_1},p_{j_2}\in S$ are {\em conflicting} if $p_{j_1}$ appears in front of $p_{j_2}$ in the sorted sequence of $S_1$ while $p_{j_2}$ appears in front of $p_{j_1}$ in the sorted sequence of $S_2$. It is not difficult to see that if no two centers of $S$ are conflicting then the above statement holds. Assume to the contrary that there exist two centers $p_{j_1},p_{j_2}\in S$ that are conflicting. Then, since the two centers are in both $S_1$ and $S_2$, each of them has two sublists. Since they are conflicting, by the definition of the sorted sequences of $S_1$ and $S_2$ and due to the convexity of $P$, the group of $p_{j_1}$ cannot be separated from the group of $p_{j_2}$ by a line, a contradiction with the line separable property of $\phi$. 
\end{proof}

\subsection{The weighted dominating set problem}
\label{sec:domwgt}
We now present our algorithm for computing a minimum-weight dominating set of $G(P)$. 

For each point $p_i\in P$, let $w_i$ denote its weight. 
We assume that each $w_i>0$, since otherwise $p_i$ could always be included in the solution. 
For any subset $S\subseteq P$, let $w(S)$ denote the total weight of all points of $S$. 
We mainly consider the following {\em bounded size problem}: Given a number $k$, compute a dominating set $S$ of minimum total weight with $|S|\leq k$ in the unit-disk graph $G(P)$. If we have an algorithm for this problem, then applying the algorithm with $k=n$ can compute a minimum weight dominating set for $G(P)$. 
Let $S^*$ denote an optimal dominating set for the above bounded size problem. Define $W^*=w(S^*)$. 

In what follows, we first describe the algorithm, then explain why it is correct, and finally discuss how to implement the algorithm efficiently. 

\subsubsection{Algorithm description and correctness}


We begin by introducing the following definition. 

\begin{definition}
\label{def:ab}
For two points $p_i,p_j\in P$ ($p_i=p_j$ is possible), define $a_{i}^j$ as the index of the first point $p$ of $P$ counterclockwise from $p_j$ such that $|p_ip|> 1$, and $b_{i}^j$ the index of the first point $p$ of $P$ clockwise from $p_j$ such that $|p_ip|> 1$ (if $|p_ip_j|>1$, then $a_{i}^j=b_{i}^j=j$). If $|p_ip| \leq 1$ for all points $p\in P$, then let $a_{i}^j=b_{i}^j=0$.   
\end{definition}



For a subset $P'\subseteq P$, let $\calD(P')$ denote the union of the unit disks centered at the points of $P'$. Note that a subset $S\subseteq P$ is a dominating set if and only if $P\subseteq \calD(S)$. 

Our algorithm has $k$ iterations. In each $t$-th iteration with $1\leq t\leq k$, we compute a set $\calL_t$ of $O(n^2)$ sublists  of $P$, and each sublist $L\in \calL_t$ is associated with a weight $w'(L)$ and a set $S_L\subseteq P$ of at most $t$ points. Our algorithm maintains the following invariant: For each sublist $L\in \calL_t$, $w(S_L)\leq w'(L)$ and points of $L$ are all covered by $\calD(S_L)$. 
Suppose that there exists a set $S\subseteq P$ of $k$ points such that $P\subseteq \calD(S)$. Then we will show that $\calL_k$ contains a sublist $L$ that is $P$ and $w'(L)\leq W^*$. As such, after $k$ iterations, we only need to find all sublists of $\calL_k$ that are $P$ and then return the one with the minimum weight. The details are described below. 

In the first iteration, for each point $p_i\in P$, we compute the two indices $a_i^i$ and $b_i^i$; we will show later in Lemma~\ref{lem:firstout} that this can be done in $O(\log n)$ time 
after $O(n\log n)$ time preprocessing. Then, let $\calL_1=\{P(b_i^i,a_i^i)\ |\ p_i\in P\}$. For each sublist $L=P(b_i^i,a_i^i)$ of $\calL_1$, we set $S_L=\{p_i\}$ and $w'(L)=w_i$. Clearly, the algorithm invariant holds on all sublists of $\calL_1$. This finishes the first iteration. It should be noted that although $|\calL_1|=O(n)$, as will be seen next, $|\calL_t|=O(n^2)$ for all $t\geq 2$. 

In general, suppose that we have a set $\calL_{t-1}$ of $O(n^2)$ sublists and each sublist $L\in \calL_{t-1}$ is associated with a weight $w'(L)$ and a set $S_L\subseteq P$ of at most $t-1$ points such that the algorithm invariant holds, i.e., $w(S_L)\leq w'(L)$ and points of $L$ are all covered by $\calD(S_L)$. We now describe the $t$-th iteration of the algorithm. 

For each point $p_i\in P$, we perform a {\em counterclockwise processing procedure} as follows.
For each point $p_j\in P$, we do the following. Compute the minimum weight sublist from $\calL_{t-1}$ that contains $P[a_i^i,j]$; we call this step a {\em minimum-weight enclosing sublist query}. We will show later in Section~\ref{sec:weightimple} that each such query can be answered in $O(\log^2 n)$ time after $O(n^2\log n)$ time preprocessing on the sublists of $\calL_{t-1}$. Let $P[j_{i1},j_{i2}]$ be the sublist computed above. Then, we compute the index $a_{i}^{j_{i2}+1}$. By definition, the union of the following three sublists is a sublist of $P$: $P(b_i^i,a_i^i)$, $P[j_{i1},j_{i2}]$, and $P(j_{i2},a_i^{j_{i2}+1})$; denote by $L$ the sublist. 
We set $S_L=S_{L'}\cup \{p_i\}$ and $w'(L)=w'(L')+w_i$, where $L'=P[j_{i1},j_{i2}]$. We add $L$ to $\calL_t$. We next argue that the algorithm invariant holds for $L$, i.e., points of $L$ are covered by $\calD(S_L)$, $|S_L|\leq t$, and $w(S_L)\leq w'(L)$. 
Indeed, by definition, all the points of $P(b_i^i,a_i^i)\cup P(j_{i2},a_i^{j_{i2}+1})$ are covered by the disk $D_{p_i}$. Since the sublist $L'$ is from $\calL_{t-1}$, by our algorithm invariant, $L'$ is covered by $\calD(S_{L'})$, $|S_{L'}|\leq t-1$, and $w(S_{L'})\leq w'(L')$. Therefore, $L$ is covered by $\calD(S_{L'}\cup \{p_i\})$ and $|S_L|\leq t$. In addition, we have $w(S_L)\leq w(S_{L'})+w_i\leq w'(L')+w_i=w'(L)$. As such, the algorithm invariant holds on $L$.

The above counterclockwise processing procedure for $p_i$ will add $O(n)$ sublists to $\calL_t$. Symmetrically, we perform a {\em clockwise processing procedure} for $p_i$, which will also add $O(n)$ sublists to $\calL_t$. We briefly discuss it. Given $p_i\in P$, for each point $p_j\in P$, we compute the minimum weight sublist from $\calL_{t-1}$ that contains $P[j,b_i^i]$. Let $P[j_{i3},j_{i4}]$ be the sublist computed above. Then, we compute the index $b_{i}^{j_{i3}-1}$. Let $L$ be the sublist that is the union of the following three sublists: $P(b_i^i,a_i^i)$, $P[j_{i3},j_{i4}]$, and $P(b_i^{j_{i3}-1},j_{i3})$. We let $S_L=S_{L'}\cup \{p_i\}$ and $w'(L)=w'(L')+w_i$, where $L'=P[j_{i3},j_{i4}]$. As above, the algorithm invariant holds on $L$. We add $L$ to $\calL_t$. 

In this way, the $t$-th iteration computes $O(n^2)$ sublists in $\calL_t$. 

After the $k$-th iteration, we find from all sublists of $\calL_k$ that are $P$ the one $L^*$ whose weight $w'(L^*)$ is the minimum. 

\paragraph{Algorithm correctness.}
The next lemma shows that $S_{L^*}$ is an optimal dominating set. 
\begin{lemma}\label{lem:correctweight}
$S_{L^*}$ is an optimal dominating set and $W^*=w'(L^*)$.
\end{lemma}
\begin{proof}
According to our algorithm invariant, the points of $L^*$, which is $P$, are covered by $\calD(S_{L^*})$, $|S_{L^*}|\leq k$, and $w(S_{L^*})\leq w'(L^*)$. As such, we obtain that $W^*\leq w(S_{L^*})\leq w'(L^*)$. To prove the lemma, it now suffices to show that $w'(L^*)\leq W^*$. For this, we resort to the ordering property in Lemma~\ref{lem:ordering}.

Let $S$ be an optimal dominating set and $\phi:\calA\rightarrow S$ be the assignment given by Lemma~\ref{lem:balassign}.
Let $p_{i_1},p_{i_2},\ldots,p_{i_k}$ be the ordering of $S$ from Lemma~\ref{lem:ordering}. As such, $(p_{i_1},p_{i_k})$ is a decoupling pair. By Lemma~\ref{lem:ordering}, for any $1\leq t\leq k$, the union of the sublists of the first $t$ centers in the ordering is a sublist of $P$, denoted by $L^t$, and $L^t$ contains the only sublist of $p_{i_1}$, which is the main sublist of $p_{i_1}$. By Lemma~\ref{lem:ordering}, $L^{t-1}\subseteq L^t$ for any $2\leq t\leq k$. Define $W_t$ as the total weight of the first $t$ centers in the ordering. By definition, $W_k=W^*$. 

To prove $w'(L^*)\leq W^*$, we show by induction that for any $1\leq t\leq k$, $\calL_t$ must contain a sublist $L$ with $L^t\subseteq L$ and $w'(L)\leq W_t$. 

As the base case, for $t=1$, $L^1$ is the only sublist of $p_{i_1}$ and $W_1=w_i$. 
Since the main sublist of $p_{i_1}$ contains $p_{i_1}$, by definition, the main sublist of $p_{i_1}$ must be contained in $P(b_{i_1}^{i_1},a_{i_1}^{i_1})$, which is a sublist in $\calL_1$. According to our algorithm, $w'(P(b_{i_1}^{i_1},a_{i_1}^{i_1}))=w_{i_1}$. Thus, we have $L^1\subseteq P(b_{i_1}^{i_1},a_{i_1}^{i_1})$ and $w'(P(b_{i_1}^{i_1},a_{i_1}^{i_1}))\leq w_{i_1}= W_1$. This proves the induction statement for the base case. 

Now assume that the induction statement holds for $t-1$, i.e., $\calL_{t-1}$ must contain a sublist $L'$ with $L^{t-1}\subseteq L'$ and $w'(L')\leq W_{t-1}$. We next prove that $\calL_{t}$ must contain a sublist $L$ with $L^{t}\subseteq L$ and $w'(L)\leq W_t$. By Lemma~\ref{lem:ordering}, one of the end sublists of $L^t$ must be the main sublist of $p_{i_t}$, and if $p_{i_t}$ has two sublists, then they are the two end sublists of $L^t$. Without loss of generality, we assume that the clockwise end sublist of $L^t$ is the main sublist of $p_{i_t}$. Define $p_j$ as the counterclockwise endpoint of $L^{t-1}$. Let $L$ be the sublist of $\calL_t$ computed when our algorithm considers $p_j$ during the counterclockwise processing procedure of $p_{i_t}$. In the following, we argue that $L^{t}\subseteq L$ and $w'(L)\leq W_{t}$. To make the notation consistent with our algorithm description, let $i=i_t$. 

First of all, since the main sublist $\alpha_i$ of $p_i$ contains $p_i$, by definition, $\alpha_i\subseteq P(b_i^i,a_i^i)$.
Because $p_j$ is the counterclockwise endpoint of $L^{t-1}$, and $\alpha_i$ and $L^{t-1}$ are consecutive, we have $\alpha_i\cup L^{t-1}\subseteq P(b_i^i,a_i^i)\cup P[a_i^i,j]$. We claim that $p_{a_i^i}$ must be in $L^{t-1}$. Indeed, assume to the contrary, this is not true. Then, since  $\alpha_i$ and $L^{t-1}$ are consecutive  and $\alpha_i\subseteq P(b_i^i,a_i^i)$, we have $L^{t-1}\subseteq P(b_i^i,a_i^i)$. By the line separable property and the convexity of $P$, all centers of $S$ in $L^{t-1}$ can be removed from $S$ and the remaining centers still form a dominating set, a contradiction with the optimality of $S$. Therefore, $p_{a_i^i}\in L^{t-1}$ holds. Consequently, we have $P[a_i^i,j]\subseteq L^{t-1}$. Since $L^{t-1}\subseteq L'$, we have $P[a_i^i,j]\subseteq L'$. According to our algorithm, a minimum-weight sublist $P[j_{i1},j_{i2}]$ from $\calL_{t-1}$ that contains $P[a_i^i,j]$ is computed. As $P[a_i^i,j]\subseteq L'\in \calL_{t-1}$, we obtain that $w'(P[j_{i1},j_{i2}])\leq w'(L')$. 
Since $w'(L')\leq W_{t-1}$ by assumption, we obtain $w'(P[j_{i1},j_{i2}])\leq W_{t-1}$.

According to our algorithm, $L$ is the union of the following three sublists: $P(b_i^i,a_i^i)$, $P[j_{i1},j_{i2}]$, and $P(j_{i2},a_i^{j_{i2}+1})$. Depending on whether $p_i$ has one or two sublists in $\phi$, there are two cases. 

\begin{itemize}
\item 
If $p_i$ has only one sublist, then $L^t=\alpha_i\cup L^{t-1}$. Since  $\alpha_i\cup L^{t-1}\subseteq P(b_i^i,a_i^i)\cup P[a_i^i,j]$, we have $L^t\subseteq P(b_i^i,a_i^i)\cup P[a_i^i,j]$. As $P[a_i^j,j]\subseteq P[j_{i1},j_{i2}]$, it follows that $L^t\subseteq P(b_i^i,a_i^i)\cup P[j_{i1},j_{i2}]\subseteq L$. Note that $W_t=W_{t-1}+w_i$. According to our algorithm, $w'(L)=w'(P[j_{i1},j_{i2}])+w_i$. Since $w'(P[j_{i1},j_{i2}])\leq W_{t-1}$, we obtain that $w'(L)\leq W_t$.

\item 
If $p_i$ has two sublists, then $L^t=\alpha_i\cup L^{t-1}\cup \beta_i$, where $\beta_i$ is the secondary sublist of $p_i$. Since $\alpha_i$ is the clockwise end sublist of $L^t$, $\beta_i$ is the counterclockwise end sublist of $L^t$. 
As above, it holds that $\alpha_i\cup L^{t-1}\subseteq P(b_i^i,a_i^i)\cup P[a_i^i,j]\subseteq P(b_i^i,a_i^i)\cup P[j_{i1},j_{i2}]$.

If $L^t\subseteq P(b_i^i,a_i^i)\cup P[j_{i1},j_{i2}]$, then we still have $L^t\subseteq L$. Otherwise, it must be the case that $p_{j_{i2}}\in \beta_i$. Consequently, since all points of $\beta_i$ are within distance $1$ from $p_i$ and $p_{j+1}$ is the clockwise endpoint of $\beta_i$, it follows that $\beta_i\subseteq P[j+1,a_i^{j_{i2}+1})$. As $\alpha_i\cup L^{t-1}\subseteq P(b_i^i,a_i^i)\cup P[a_i^i,j]$, we obtain $L^t=\alpha_i\cup L^{t-1}\cup \beta_i\subseteq P(b_i^i,a_i^i)\cup P[a_i^i,j]\cup \beta_i\subseteq P(b_i^i,a_i^i)\cup P[a_i^i,j]\cup P[j+1,a_i^{j_{i2}+1})$. Since $P[a_i^i,j]\subseteq P[j_{i1},j_{i2}]$ by definition, we have $P[a_i^i,j]\cup P[j+1,a_i^{j_{i2}+1})\subseteq P[j_{i1},j_{i2}]\cup P(j_{i2},a_i^{j_{i2}+1})$. Hence, $P(b_i^i,a_i^i)\cup P[a_i^i,j]\cup P[j+1,a_i^{j_{i2}+1})\subseteq P(b_i^i,a_i^i)\cup P[j_{i1},j_{i2}]\cup P(j_{i2},a_i^{j_{i2}+1})=L$. Therefore, $L^t\subseteq L$ holds.

Following the same analysis as the above case, we can obtain  $w'(L)\leq W_t$. 
\end{itemize}

In summary, for both cases, we have $L^t\subseteq L$ and  $w'(L)\leq W_t$. This proves the induction statement for $t$.

Applying the induction statement to $k$ obtains that $\calL_k$ must contain a sublist $L$ with $L^k\subseteq L$ and $w'(L)\leq W_k$. As $L^k=P$ and $W_k=W^*$, we obtain that $L=P$ and $w'(L)\leq W^*$. According to our algorithm, $L^*$ is the sublist of $\calL_k$ such that $L^*=P$ and $w'(L^*)$ is the minimum. Therefore, we have $w'(L^*)\leq w'(L)$, and thus $w'(L^*)\leq W^*$.

The lemma thus follows. 
\end{proof}

\paragraph{Time analysis.}
For the time analysis, in each iteration, we perform $O(n^2)$ operations for computing indices $a_i^j$ and $b_i^j$, and perform $O(n^2)$ minimum-weight enclosing sublist queries. We will show later in Section~\ref{sec:weightimple} that each of these operations takes $O(\log^2 n)$ time after $O(n^2\log n)$ time preprocessing. 
As such, each iteration of the algorithm takes $O(n^2\log^2 n)$ time and the total time of the algorithm is thus $O(kn^2\log^2 n)$. 

\subsubsection{Algorithm implementation}
\label{sec:weightimple}

The following lemma provides a data structure for computing the indices $a_i^j$ and $b_i^j$. 

\begin{lemma}
\label{lem:firstout}
We can construct a data structure for $P$ in $O(n\log n)$ time such that the indices $a_i^j$ and $b_i^j$ can be computed in $O(\log n)$ time for any two points $p_i, p_j\in P$.
\end{lemma}
\begin{proof}
We will first present a method based on farthest Voronoi diagrams that can compute $a_i^j$ and $b_i^j$ in $O(\log^2 n)$ time. We will then improve the algorithm to $O(\log n)$ time by using fractional cascading~\cite{ref:ChazelleFi86}. The preprocessing time of both methods is $O(n\log n)$. One reason we present the first method is that it will be used later in our parametric searching algorithm for the discrete $k$-center problem in Section~\ref{sec:discretecenter}. 

\paragraph{The farthest Voronoi diagram method.}
We build a complete binary tree $T$ whose leaves from left to right store $p_1,p_2,\ldots,p_n$, respectively. For each node $v\in T$, let $P_v$ denote the subset of points of $P$ in the leaves of the subtree rooted at $v$. We construct the farthest Voronoi diagram on $P_v$, denoted by $\fvd(v)$. We also build a point location data structure on $\fvd(v)$ so that given a query point the cell of $\fvd(v)$ containing the point can be found in $O(\log |P_v|)$ time~\cite{ref:EdelsbrunnerOp86,ref:KirkpatrickOp83}. 
Since points of $P_v$ are in convex position, 
constructing $\fvd(v)$ can be done in $O(|P_v|)$ time~\cite{ref:AggarwalA89}. Constructing the point location data structure takes linear time in the size of $\fvd(v)$ (which is $O(|P_v|)$)~\cite{ref:EdelsbrunnerOp86,ref:KirkpatrickOp83}. Doing this for all nodes $v\in T$ takes $O(n\log n)$ time in total.

Given two points $p_i,p_j\in P$, we can use $T$ to compute the index $a_i^j$ in $O(\log^2 n)$ time as follows. 
The main idea is similar to operations in finger search trees~\cite{ref:BrodalFi04,ref:GuibasA77,ref:Tsakalidis85}. We first check whether $a_i^j=0$, i.e., whether the distance from $p_i$ to every point of $P$ is at most $1$. For this, we find the farthest point $p'$ of $p_i$ in $P$, which can be done by finding the cell of $\fvd(v)$ containing $p_i$ with $v$ as the root of $T$; the latter task can be done by a point location query on $\fvd(v)$, which takes $O(\log n)$ time. Notice that $a_i^j=0$ if and only if $|p'p_i|\leq 1$. Next, we check if $a_i^j=j$ by testing whether $|p_ip_j|>1$ holds. In the following, we assume that $a_i^j\neq 0$ and $a_i^j\neq j$. Note that either $a_i^j>j$ or $a_i^j<j$ is possible. We first consider the case $a_i^j>j$ and the other case can be handled similarly with a slight modification.

Let $v_j$ be the leaf of $T$ storing $p_j$. Starting from $v_j$, we go up on $T$ until the first node $v$ whose right child $u$ has the following property: The distance between $p_i$ and its farthest point in $P_u$ is greater than $1$. Let $v^*$ denote such a node $v$. 
Specifically, starting from $v=v_j$, if $v$ does not have a right child, then set $v$ to its parent. Otherwise, let $u$ be the right child of $v$. Using a point location query on $\fvd(P_u)$, we find the farthest point $p'$ of $p_i$ in $P_u$. If $|p'p_i|\leq 1$, then we set $v$ to its parent; otherwise, we have found $v^*$ that is current vertex $v$. Since $a_i^j>j$, $v^*$ is guaranteed to be found. Next, starting from the right child of $v^*$, we perform a top-down search process. For each node $v$, let $u$ be its left child. We find the farthest point $p'$ of $p_i$ in $P_u$. If $|p'p_i|> 1$, then we set $v$ to $u$; otherwise, we set $v$ to its right child. The process will eventually reach a leaf $v$; we return the index of the point stored at $v$ as $a_i^j$. The total time is $O(\log^2 n)$ as the search process calls point location queries $O(\log n)$ times and each point location query takes $O(\log n)$ time. 

We now consider the case $a_i^j<j$. In this case, the distance between $p_i$ and the point in any leaf to the right of $v_j$ is at most $1$ and thus the bottom-up procedure in the above algorithm will eventually reach the root. Note that $a_i^j$ is the first point in $P[1,j-1]$ whose distance from $p_i$ is larger than $1$. We proceed with the following. If $|p_1p_i|> 1$, then $a_i^j=1$ and we can stop the algorithm. Otherwise, starting from $v=v_1$, the leftmost leaf of $T$, we apply the same algorithm as in the above case, i.e., first run a bottom-up procedure and then a top-down one. The total time of the algorithm is $O(\log^2 n)$. 

Computing $b_i^j$ can be done similarly in $O(\log^2 n)$ time.

\paragraph{The fractional cascading method.}
For any subset $P'\subseteq P$, let $\calI(P')$ denote the common intersection of the unit disks centered at the points of $P'$. 

We still construct a complete binary tree $T$ as above. For each node $v$, instead of constructing $\fvd(v)$, we compute $\calI(v)$. Our approach is based on the following observation: For any point $p_i$, $P_v\subseteq D_{p_i}$ if and only if $p_i\in \calI(v)$. With the observation, to determine whether $P_v\subseteq D_{p_i}$, instead of using $\fvd(v)$ to find the farthest point of $P_v$ from $p_i$, we determine whether $p_i\in \calI(v)$. Computing $\calI(v)$ for all nodes $v\in T$ can be done in $O(n\log n)$ time in a bottom-up manner (because $\calI(v)$ can be obtained from $\calI(u)$ and $\calI(w)$ in linear time for the two children $u$ and $w$ of $v$~\cite{ref:HershbergerFi91,ref:WangOn22}). Finally, we construct a fractional cascading data structure on the vertices of the boundary of $\calI(v)$ for all nodes $v\in T$. This takes $O(n\log n)$ time as the total number of such boundary vertices is $O(n\log n)$~\cite{ref:ChazelleFi86}. This finishes the preprocessing, whose total time is $O(n\log n)$.

Given two points $p_i,p_j\in P$, we compute the index $a_i^j$ in $O(\log n)$ time as follows. We again first consider the case where $a_i^j>j$. We follow the algorithmic scheme as in the above first method. Recall that we need to find the vertex $v^*$. To this end, we preform an additional step first. Let $\pi_j$ be the path in $T$ from the root to the leaf $v_j$ that stores $p_j$. Let $V$ be the set of nodes $v$ of $T$ whose parents are in $\pi_j$ such that $v$ is the right child of its parent. We wish to determine whether $p_i$ is in $\calI(v)$ for all $O(\log n)$ nodes $v\in V$. This can be done in $O(\log n)$ time using the fractional cascading data structure~\cite{ref:ChazelleFi86}. We now search $v^*$ in a bottom-up manner as before. Starting from $v=v_j$, if $v$ does not have a right child, then set $v$ to its parent. Otherwise, let $u$ be the right child of $v$. We already know whether $p_i\in \calI(u)$. If so, then all points of $P_u$ are inside the unit disk $D_{p_i}$ and we set $v$ to its parent. Otherwise, we have found $v^*$ that is current vertex $v$. Next, starting from the right child of $v^*$, we perform a top-down search process. For each node $v$, let $u$ be its left child. We determine whether $p_i\in \calI(u)$, which takes $O(1)$ time due to the fractional cascading data structure (we spend $O(\log n)$ time doing binary search at the root and then $O(1)$ time per node subsequently)~\cite{ref:ChazelleFi86}. If $p_i\not\in \calI(u)$, then we set $v$ to $u$; otherwise, we set $v$ to its right child. The process will eventually reach a leaf $v$; we return the index of the point stored at $v$ as $a_i^j$. The total time is $O(\log n)$.

For the case $a_i^j<j$, the algorithm follows the scheme in the first method but instead uses the fractional cascading data structure. The total time is also $O(\log n)$. 

Computing $b_i^j$ can be done similarly in $O(\log n)$ time. 
\end{proof}

\paragraph{Minimum-weight enclosing sublist queries.}
\label{sec:minwgtenclarc}

We now present our data structure for the minimum-weight enclosing sublist queries. Given a set $\calL$ of $m$ sublists of $P$, each sublist has a weight. We wish to build a data structure to answer the following minimum-weight enclosing sublist queries: Given a sublist $L$, compute the minimum weight sublist of $\calL$ that contains $L$. 

\begin{lemma}\label{lem:enclosequery}
We can construct a data structure for $\calL$ in $O(m\log m)$ time, with $m=|\calL|$, so that each minimum-weight enclosing sublist query can be answered in $O(\log^2 m)$ time.     
\end{lemma}
\begin{proof}  
We first consider a special case in which $i\leq j$ holds for each sublist $P[i,j]\in \calL$. We will show later that the general case can be reduced to this case. 

In the special case, the sublists of $\calL$ become 1D intervals, and it is well known that the problem can be reduced to 2D orthogonal range searching~\cite{ref:deBergCo08}. Indeed, create a point with coordinate $(i,j)$ in the plane for each sublist $P[i,j]\in \calL$; define the weight of the point the same as that of $P[i,j]$. Let $Q$ denote the set of all points. Then, for a query sublist $P[i',j']$, the sublists of $\calL$ containing $P[i',j']$ correspond to exactly the points of $Q$ to the northwest of the point $(i',j')$. As such, by constructing a 2D range tree for $Q$ in $O(m\log m)$ time, each query can be answered in $O(\log^2 m)$ time~\cite{ref:deBergCo08}. 

In the general case, there may exist sublists $P[i,j]\in \calL$ with $i>j$. The problem now becomes handling cyclic intervals. We can reduce the problem to the 1D intervals as follows. For each sublists $P[i,j]\in \calL$ with $i>j$, we create two sublists $P[i,n+j]$ and $P[1,j]$. Let $\calL'$ denote the set of all these new sublists as well as those sublists $P[i,j]\in \calL$ with $i\leq j$. Clearly, $\calL'$ has at most $2m$ sublists $P[i,j]$ with $1\leq i\leq j\leq 2n-1$. 

We can use $\calL'$ to answer queries as follows. Consider a query sublist $P[i',j']$. 
\begin{itemize}
    \item 
If $i'\leq j'$, then for any sublist $P[i,j]\in \calL$ containing $P[i',j']$, it must correspond to a sublist in $\calL$ that contains $P[i',j']$ (specifically, if $i\leq j$, then $P[i,j]$ is also in $\calL'$; otherwise, one of the two sublists in $\calL'$ created from $P[i,j]$ contains $P[i',j']$). On the other hand, any sublist of $\calL'$ containing $P[i',j']$ must correspond to a sublist of $\calL$ containing $P[i',j']$. As such, a minimum-weight sublist of $\calL'$ containing $P[i',j']$ corresponds to a minimum-weight sublist of $\calL$ containing $P[i',j']$. 
\item
If $i'>j'$, then a minimum-weight sublist of $\calL'$ containing $P[i',j'+n]$ can give our answer. Indeed, for each sublist $P[i,j]\in \calL'$ containing $P[i',j'+n]$, we have $i\leq n<j$ and thus $P[i,j]$ is defined by an old interval $P[i,j-n]$ in $\calL$ and $P[i,j-n]$ contains $P[i',j']$. On the other hand, if a sublist $P[i,j]\in \calL$ contains $P[i',j']$, then $P[i,j]$ defines a new sublist $P[i,j+n]\in \calL'$ that also contains $P[i',j'+n]$. As such, a minimum-weight sublist of $\calL'$ containing $P[i',j'+n]$ corresponds to a minimum-weight sublist of $\calL$ containing $P[i,j]$. 
\end{itemize}

Therefore, we can apply the 2D range tree method on $\calL'$ to build a data structure in $O(m\log m)$ time that can answer each query in $O(\log^2 m)$ time. 
\end{proof}


\subsubsection{Putting it all together}

The following theorem summarizes our result.

\begin{theorem}\label{thm:domsetwgt}
Given a number $k$ and a set $P$ of $n$ weighted points in convex position in the plane, we can find in $O(kn^2\log^2 n)$ time a minimum-weight dominating set of size at most $k$ in the unit-disk graph $G(P)$, or report no such dominating set exists. 
\end{theorem}

Applying Theorem~\ref{thm:domsetwgt} with $k=n$ leads to the following result. 
\begin{corollary}
Given a set $P$ of $n$ weighted points in convex position in the plane, we can compute a minimum-weight dominating set in the unit-disk graph $G(P)$ in $O(n^3\log^2 n)$ time. 
\end{corollary}



\subsection{The unweighted case}
\label{sec:domunwgt}

In this section, we consider the unweighted dominating set problem. The goal is to compute the smallest dominating set in the unit-disk graph $G(P)$. Note that all properties for the weighted case are also applicable here. 
In particular, by setting the weights of all points to $1$ and applying Theorem~\ref{thm:domsetwgt}, one can solve the unweighted problem in $O(n^3\log^2 n)$ time. Here, we provide an improved algorithm of $O(kn\log n)$ time, where $k$ is the smallest dominating set size.

\subsubsection{Algorithm description and correctness}
We follow the iterative algorithmic scheme of the weighted case, but incorporate a greedy strategy by taking advantage of the property that all points of $P$ have the same weight.

In each $t$-th iteration of the algorithm, $t\geq 1$, we compute a set $\calL_t$ of $O(n)$ sublists and each list $L\in \calL_t$ is associated with a set $S_L\subseteq P$ of at most $t$ points. Our algorithm maintains the following invariant: For each sublist $L\in \calL_t$, all points of $L$ are covered by $\calD(S_L)$, i.e., the union of the unit disks centered at the points of $S_L$. If $k$ is the smallest dominating set size, we will show that after $k$ iterations, $\calL_k$ is guaranteed to contain a sublist that is $P$. As such, if a sublist that is $P$ is computed for the first time, then we can stop the algorithm.

Initially, we compute the indices $a_i^i$ and $b_i^i$ for all points $p_i\in P$. By Lemma~\ref{lem:firstout}, this takes $O(\log n)$ time after $O(n\log n)$ time preprocessing. In the first iteration, we have $\calL_1=\{P(b_i^i,a_i^i)\ |\ p_i\in P\}$. For each sublist $L=P(b_i^i,a_i^i)\in \calL_1$, we set $S_L=\{p_i\}$.
Clearly, the algorithm invariant holds. 

In general, suppose that we have a set $\calL_{t-1}$ of $O(n)$ sublists such that the algorithm invariant holds. We assume that no sublist of $\calL_{t-1}$ is $P$. Then, the $t$-th iteration of the algorithm works as follows. For each point $p_i\in P$, we perform the following {\em counterclockwise processing procedure}. We first compute the sublist of $\calL_{t-1}$ that contains $p_{a_i^i}$ and has the most counterclockwise endpoint. This is done by a {\em counterclockwise farthest enclosing sublist query}. We will show later in Section~\ref{sec:unwtimple} that each such query takes $O(\log n)$ time after $O(n\log n)$ time preprocessing for $\calL_{t-1}$. 
Let $P[j_{i1},j_{i2}]$ be the sublist computed above. Then, we compute the index $a_{i}^{j_{i2}+1}$ in $O(\log n)$ time by Lemma~\ref{lem:firstout}. Note that the union of the following three sublists is a sublist $L$ of $P$: $P(b_i^i,a_i^i)$, $P[j_{i1},j_{i2}]$, and $P(j_{i2},a_i^{j_{i2}+1})$.
We add $L$ to $\calL_t$ and set $S_L=S_{L'}\cup \{p_i\}$ with $L'=P[j_{i1},j_{i2}]$.
By our algorithm invariant, points of $L'$ are covered by $\calD(S_{L'})$. By definition, points of $P(b_i^i,a_i^i)\cup P(j_{i2},a_i^{j_{i2}+1})$ are covered by $D_{p_i}$. Therefore, all points of $L$ are covered by $\calD(S_{L})$. Hence, the algorithm invariant holds for $L$. In addition, if $L$ is $P$, then we stop the algorithm and return $S_L$ as an optimal dominating set. 

Symmetrically, we perform a {\em clockwise processing procedure} for $p_i$. We compute the sublist from $\calL_{t-1}$ that contains $b_i^i$ and has the most clockwise endpoint; this is done by a {\em clockwise farthest enclosing sublist query}. Let $P[j_{i3},j_{i4}]$ be the sublist computed above. Then, we compute the index $b_{i}^{j_{i3}-1}$. Let $L$ be the sublist that is the union of the following three sublists: $P(b_i^i,a_i^i)$, $P[j_{i3},j_{i4}]$, and $P(j_{i3},b_i^{j_{i3}-1})$. Let $S_L=S_{L'}\cup \{p_i\}$ with $L'=P[j_{i3},j_{i4}]$. As above, the algorithm invariant holds on $L$. We add $L$ to $\calL_t$. If $L$ is $P$, then we stop the algorithm and return $S_L$ as an optimal dominating set. 

\paragraph{Algorithm correctness.}
The following lemma proves the correctness of the algorithm. 
\begin{lemma}
If the algorithm first time computes a sublist $L$ that is $P$, then $S_L$ is the smallest dominating set of $G(P)$. 
\end{lemma}
\begin{proof}
Let $k$ be the smallest dominating set size. By our algorithm invariant, it suffices to show that the algorithm will stop within $k$ iterations. To this end, we resort to the ordering property in Lemma~\ref{lem:ordering}. Many arguments are similar to the weighted case proof in Section~\ref{lem:correctweight}.

Let $S$ be an optimal dominating set and $\phi:\calA\rightarrow S$ be the assignment given by Lemma~\ref{lem:balassign}. 
Let $p_{i_1},p_{i_2},\ldots,p_{i_k}$ be the ordering of $S$ from Lemma~\ref{lem:ordering}. As such, $(p_{i_1},p_{i_k})$ is a decoupling pair. By Lemma~\ref{lem:ordering}, for any $1\leq t\leq k$, the union of the sublists of the first $t$ centers in the ordering is a sublist of $P$, denoted by $L^t$, and $L^t$ contains the only sublist of $p_{i_1}$. Also, $L^{t-1}\subseteq L^t$ for any $2\leq t\leq k$. 
By definition, $L^k=P$.

In the following, we show by induction that if our algorithm runs for $k$ iterations, then for any $1\leq t\leq k$, $\calL_t$ must contain a sublist $L$ with $L^t\subseteq L$.  

As the base case, for $t=1$, $L^1$ is the only sublist of $p_{i_1}$, which is the main sublist by Lemma~\ref{lem:ordering}. Since the main sublist of $p_{i_1}$ contains $p_{i_1}$, by definition, the main sublist of $p_{i_1}$ must be contained in $P(b_{i_1}^{i_1},a_{i_1}^{i_1})$, which is a sublist in $\calL_1$. 

Now assume that $\calL_{t-1}$ must contain a sublist $L'$ with $L^{t-1}\subseteq L'$. We argue that $\calL_{t}$ must contain a sublist $L$ with $L^{t}\subseteq L$. By Lemma~\ref{lem:ordering}, one of the end sublists of $L^t$ must be the main sublist of $p_{i_t}$, and if $p_{i_t}$ has two sublists, then they are the two end sublists of $L^t$. Without loss of generality, we assume that the clockwise end sublist of $L^t$ is the main sublist of $p_{i_t}$. Define $p_j$ as the  counterclockwise endpoint of $L^{t-1}$. Let $L$ be the sublist of $\calL_t$ computed by our algorithm during the counterclockwise processing procedure for $p_{i_t}$. In the following, we show that $L^{t}\subseteq L$. To make the notation consistent with our algorithm description, let $i=i_t$. 

First of all, since the main sublist $\alpha_i$ of $p_i$ contains $p_i$, by definition, $\alpha_i\subseteq P(b_i^i,a_i^i)$. Since $p_j$ is the counterclockwise endpoint of $L^{t-1}$, and $\alpha_i$ and $L^{t-1}$ are consecutive, we have $\alpha_i\cup L^{t-1}\subseteq P(b_i^i,a_i^i)\cup P[a_i^i,j]$. We claim that the point $p_{a_i^i}$ must be in $L^{t-1}$. Indeed, assume to the contrary this is not true. Then, since $\alpha_i$ and $L^{t-1}$ are consecutive and $\alpha_i\subseteq P(b_i^i,a_i^i)$, we have $L^{t-1}\subseteq P(b_i^i,a_i^i)$. By the line separable property, all centers of $S$ in $L^{t-1}$ can be removed from $S$ and the remaining centers still form a dominating set, a contradiction with the optimality of $S$. Therefore, $p_{a_i^i}\in L^{t-1}$. Consequently, we have $P[a_i^i,j]\subseteq L^{t-1}$. Since $L^{t-1}\subseteq L'$, it follows that $P[a_i^i,j]\subseteq L'$. Recall that the algorithm computes a sublist $P[j_{i1},j_{i2}]$ from $\calL_{t-1}$ that contains $p_{a_i^i}$ and has the farthest counterclockwise endpoint. Since $p_{a_i^i}\in L^{t-1}\subseteq L'\in \calL_{t-1}$ and $\alpha_i\subseteq P(b_i^i,a_i^i)$, we obtain that $P[a_i^i,j]\subseteq P[j_{i1},j_{i2}]$ and $\alpha_i\cup L^{t-1}\subseteq P(b_i^i,a_i^i)\cup P[j_{i1},j_{i2}]$. 
The rest of the proof follows the similar argument to the weighted case proof of Lemma~\ref{lem:correctweight}. 

According to our algorithm, $L$ is the union of $P(b_i^i,a_i^i)$, $P[j_{i1},j_{i2}]$, and $P(j_{i2},a_i^{j_{i2}+1})$. Depending on whether $p_i$ has one or two sublists in $\phi$, there are two cases. 

\begin{itemize}
    \item 
If $p_i$ has only one sublist, then $L^t=\alpha_i\cup L^{t-1}$. As discussed above, we have $L^t\subseteq P(b_i^i,a_i^i)\cup P[j_{i1},j_{i2}]\subseteq L$. 

\item 
If $p_i$ has two sublists, then $L^t=\alpha_i\cup L^{t-1}\cup \beta_i$, where $\beta_i$ is the secondary sublist of $p_i$. If $L^t\subseteq P(b_i^i,a_i^i)\cup P[j_{i1},j_{i2}]$, then we still have $L^t\subseteq L$. Otherwise, it must be the case that $p_{j_{i2}}\in \beta_i$. Since all points of $\beta_i$ are within distance $1$ from $p_i$ and $p_{j+1}$ is the clockwise endpoint of $\beta_i$, it follows that $\beta_i\subseteq P[j+1,a_i^{j_{i2}+1})$. 
As $\alpha_i\cup L^{t-1}\subseteq P(b_i^i,a_i^i)\cup P[a_i^i,j]$, we obtain $L^t=\alpha_i\cup L^{t-1}\cup \beta_i\subseteq P(b_i^i,a_i^i)\cup P[a_i^i,j]\cup \beta_i\subseteq P(b_i^i,a_i^i)\cup P[a_i^i,j]\cup P[j+1,a_i^{j_{i2}+1})$.
Since $P[a_i^i,j]\subseteq P[j_{i1},j_{i2}]$ as already shown above, we have $P[a_i^i,j]\cup P[j+1,a_i^{j_{i2}+1})\subseteq P[j_{i1},j_{i2}]\cup P(j_{i2},a_i^{j_{i2}+1})$. Hence, $P(b_i^i,a_i^i)\cup P[a_i^i,j]\cup P[j+1,a_i^{j_{i2}+1})\subseteq P(b_i^i,a_i^i)\cup P[j_{i1},j_{i2}]\cup P(j_{i2},a_i^{j_{i2}+1})=L$. Therefore, $L^t\subseteq L$ holds.
\end{itemize}

In summary, for both cases, we have $L^t\subseteq L$. This proves the induction statement for $t$. 

By the induction statement, after the $k$-th iteration, $\calL_k$ has a sublist $L$ with $L^k\subseteq L$. Since $L^k=P$, the algorithm will stop within $k$ iterations. The lemma thus follows. 
\end{proof}

\paragraph{Time analysis.}
For the time analysis, in each iteration, we perform $O(n)$ operations for computing indices $a_i^j$ and $b_i^j$ and $O(n)$ counterclockwise/clockwise farthest enclosing sublist queries. Computing indices $a_i^j$ and $b_i^j$ takes $O(\log n)$ time by Lemma~\ref{lem:firstout}. 
We show in Section~\ref{sec:unwtimple} that each counterclockwise/clockwise farthest enclosing sublist query can be answered in $O(\log n)$ time after $O(n\log n)$ time preprocessing. As such, each iteration runs in $O(n\log n)$ time and the total time of the algorithm is $O(kn\log n)$, where $k$ is the smallest dominating set size. 

\subsubsection{Algorithm implementation}
\label{sec:unwtimple}

It remains to describe the data structure for answering counterclockwise/clockwise farthest enclosing sublist queries. We only discuss the counterclockwise case as the clockwise case can be handled analogously. Given a set $\calL$ consisting of $n$ sublists of $P$, the goal is to build a data structure to answer the following counterclockwise farthest enclosing sublist queries: Given a point $p\in P$, find a sublist in $\calL$ that contains $p$ with the farthest counterclockwise endpoint from $p$.


\begin{lemma}
We can construct a data structure for $\calL$ in $O(n\log n)$ time such that each counterclockwise farthest enclosing sublist query can be answered in $O(\log n)$ time.
\end{lemma}
\begin{proof}
We first consider the case where $i\leq j$ holds for every sublist $P[i,j]\in \calL$. 
For each sublist $P[i,j] \in \calL$, we create a point with coordinate $(i,j)$ in the plane. Let $Q$ denote the set of all these points. For a query point $p_{t}$, the sublists of $\calL$ containing $p_{t}$ correspond exactly to the points in $Q$ that are to the northwest of the point $(t,t)$. The sublist $L$ with the farthest counterclockwise endpoint from $p_t$ corresponds to the highest point in the above range. As such, we can find $L$ as follows. We first find the highest point $q\in Q$ among all points to the left of the vertical line through $(t,t)$. If the $y$-coordinate of $q$ is smaller than $t$, then no sublist of $\calL$ contains $p_t$; otherwise, the sublist defining $q$ is the answer to our query. 
As such, our problem reduces to the query of computing $q$. To answer such queries, we can use an augmented binary search tree that store all points of $Q$ from left to right in the leaves of the tree. Each node of the tree stores the highest point among the points in the leaves of the subtree rooted at the node. The preprocessing time is thus $O(n\log n)$ and the query time is $O(\log n)$. 

If $\calL$ has sublists $P[i,j]$ with $i>j$, then we can reduce the problem to the above special case in the same way as in the proof of Lemma~\ref{lem:enclosequery}. Specifically, we create a new set $\calL'$ of at most $2n-1$ sublists on the indices $1,2,\ldots,2n-1$, as follows. 
For each sublist $P[i,j] \in \calL$, if $i\leq j$, then we add $P[i,j]$ to $\calL'$; otherwise, we create two sublists $P[i,j+n]$ and $P[1,j]$ for $\calL'$. As such, $\calL'$ contains at most $2n$ sublists $P[i,j]$ with $1 \leq i \leq j \leq 2n-1$.

Following the same argument as in the proof of Lemma~\ref{lem:enclosequery}, we can show that the sublist $\calL'$ with the farthest counterclockwise endpoint and containing a query point $p_t$ corresponds to the sublist of $\calL$ with the farthest counterclockwise endpoint and containing $p_t$. As such, after $O(n\log n)$ time preprocessing, each query can be answered in $O(\log n)$ time. 
\end{proof}

\subsubsection{Putting it all together}

We conclude with following theorem.

\begin{theorem}\label{thm:domsetunwgt}
Given a set $P$ of $n$ points in convex position in the plane, the smallest dominating set of the unit-disk graph $G(P)$ can be computed in $O(kn\log n)$ time, where $k$ is the size of the smallest dominating set. 
\end{theorem}

The following corollary will be used in Section~\ref{sec:discretecenter} to solve the discrete $k$-center problem. 

\begin{corollary}\label{crl:genraddomsetunwgt}
Given $k$, $r$, and a set $P$ of $n$ points in convex position in the plane, one can do the following in $O(kn\log n)$ time: determine whether there exists a subset $S\subseteq P$ of at most $k$ points such that the distance from any point of $P$ to its closest point in $S$ is at most $r$, and if so, find such a subset $S$.   
\end{corollary}
\begin{proof}
We redefine the unit-disk graph of $P$ with a parameter $r$, where two points in $P$ are connected by an edge if their distance is at most $r$. We then apply the algorithm of Theorem~\ref{thm:domsetunwgt}. If the algorithm finds a sublist $L$ that is $P$ within $k$ iterations, then we return $S=S_L$; otherwise such a subset $S$ as in the lemma statement does not exist. Since we run the algorithm for at most $k$ iterations, the total time of the algorithm is $O(kn\log n)$.
\end{proof}

\section{The discrete $k$-center problem}
\label{sec:discretecenter}
In this section, we present our algorithm for the discrete $k$-center problem. Let $P$ be a set of $n$ points in convex position in the plane. Given a number $k$, the goal is to compute a subset $S \subseteq P$ of at most $k$ points (called {\em centers}) so that the maximum distance between any point in $P$ and its nearest center is minimized. Let $r^*$ denote the optimal objective value.

Given a value $r$, the \textit{decision problem} is to determine whether $r \geq r^*$, or equivalently, whether there exist a set of $k$ centers in $P$ such that the distance from any point of $P$ to its closest center is at most $r$. By Corollary~\ref{crl:genraddomsetunwgt}, the problem can be solved in $O(kn\log n)$ time. Clearly, $r^*$ is equal to the distance of two points of $P$, that is $r^*\in R$, where $R$ is defined as the set of all pairwise distances between points in $P$.
If we explicitly compute $R$ and then perform a binary search on $R$ using the algorithm of Corollary~\ref{crl:genraddomsetunwgt} as a \textit{decision algorithm}, then $r^*$ can be computed in $O(n^2+kn\log^2 n)$ time. We can improve the algorithm by using the distance selection algorithms, which can find the $k$-th smallest value in $R$ in $O(n^{4/3}\log n)$ time for any given $k$~\cite{ref:KatzAn97,ref:WangIm23}. In fact, by applying the algorithmic framework of Wang and Zhao~\cite{ref:WangIm23} with our decision algorithm, $r^*$ can be computed in $O(n^{4/3}\log n + nk\log^2 n)$ time. 


In the following, we present another algorithm of time complexity $O(k^2n\log^2 n)$ using the parametric search technique~\cite{ref:ColeSl87,ref:MegiddoAp83}. This algorithm is faster than the above one when $k = o(n^{1/6}/\sqrt{\log n})$. 


We simulate the decision algorithm over the unknown optimal value $r^*$. The algorithm maintains an interval $(r_1, r_2]$ that contains $r^*$. Initially, $r_1=-\infty$ and $r_2=\infty$. During the course of the algorithm, the decision algorithm is invoked on certain \textit{critical values} $r$ to determine whether $r \geq r^*$; based on the outcome, the interval $(r_1, r_2]$ is shrunk accordingly so that the new interval still contains $r^*$. Upon completion, we will show that $r^* = r_2$ must hold.


\paragraph{Algorithm overview.}
For any $r$, certain variables in our decision algorithm are now defined with respect to $r$ as the radius of unit disks and therefore may be considered as functions of $r$. For example, we use $a_{i}^j(r)$ to represent $a_i^j$ when the unit disk radius is $r$. The algorithm has $k$ iterations. We wish to compute the sublist set $\calL_t(r^*)$ in each $t$-th iteration, $1\leq t\leq k$. Specifically, the set $\calL_1(r^*)$ relies on $a^i_i(r^*)$ and $b^i_i(r^*)$ for all points $p_i\in P$. As such, in the first iteration, we will compute $a^i_i(r^*)$ and $b^i_i(r^*)$ for all $p_i\in P$. The computation process will generate certain critical values $r$, call the decision algorithm on these values, and shrink the interval $(r_1,r_2]$ accordingly. After that, $\calL_1(r^*)$ can be computed. 

In a general $t$-th iteration, our goal is to compute the set $\calL_t(r^*)$. We assume that the set $\calL_{t-1}(r^*)$ is already available with an interval $(r_1,r_2]$ containing $r^*$. Then, for each $p_i \in P$, we perform a counterclockwise processing procedure. We first compute the sublist of $\calL_{t-1}(r^*)$ with the farthest counterclockwise endpoint and containing $a_i^i(r^*)$. This procedure depends solely on $a_i^i(r^*)$ and $\calL_{t-1}(r^*)$, which are already available, and thus no critical values are generated. Suppose that $P[j_{i1}(r^*),j_{i2}(r^*)]$ is the sublist computed above. The next step is to compute $a_i^{j_{i2}(r^*)+1}(r^*)$. This step will again generate critical values and shrink the interval $(r_1,r_2]$. After that, we add to $\calL_{t}(r^*)$ the sublist that is the union of the following three sublists: $P(b_i^i(r^*),a_i^i(r^*))$, $P[j_{i1}(r^*),j_{i2}(r^*)]$, and $P(j_{i2}(r^*),a_i^{j_{i2}(r^*)+1})$. Similarly, we perform a clockwise processing procedure for each point $p_i\in P$. After that, the set $\calL_t(r^*)$ is computed. The details are given below. 

\paragraph{The first iteration: Computing $a_i^i(r^*)$ and $b_i^i(r^*)$.} We now discuss how to compute $a_i^i(r^*)$ and $b_i^i(r^*)$ for all points $p_i\in P$. We will parameterize the farthest Voronoi diagram method of Lemma~\ref{lem:firstout}, referred to as the {\em FVD algorithm}. Initially, we set $r_1=-\infty$ and $r_2=\infty$. As such, we have $r^*\in (r_1,r_2]$. 

We first construct a binary search tree $T$ on $P$, including building the farthest Voronoi diagram $\fvd(v)$ and the point location data structure on $\fvd(v)$ for all nodes $v \in T$. This process is independent of $r^*$ and takes $O(n \log n)$ time in total.

Next, for each point $p_i \in P$, we compute $a_i^i(r^*)$ by traversing a path $\pi_i^i(r^*)$ in $T$ (i.e., the nodes in the bottom-up procedure and then the top-down procedure in the FVD algorithm), which consists of $O(\log n)$ nodes. At each node $v \in \pi_i^i(r^*)$, we locate the farthest point $p'$ from $p_i$ within $P_v$ using a point location query in $\fvd(P_v)$. The farthest point $p'$ and the point location algorithm for finding $p'$ are independent of $r^*$. Next, we need to compare $r = |p'p_i|$ with $r^*$ to determine whether $r \geq r^*$; here $r$ is a critical value. If $r \leq r_1$, then since $r^*\in (r_1,r_2]$, we have $r < r^*$, and the comparison is resolved. If $r \geq r_2$, then $r \geq r^*$ and the comparison is also resolved. In both cases, we can resolve the comparison without invoking the decision algorithm. If $r \in (r_1, r_2)$, then we apply the decision algorithm on $r$ to check whether $r \geq r^*$. 
If so, we update $r_2 = r$; otherwise, we update $r_1 = r$. After this, we still have $r^* \in (r_1, r_2]$. Once the comparison is resolved, the algorithm proceeds accordingly, following the FVD algorithm. 

Our algorithm maintains the following {\em invariant}: For all $r\in (r_1,r_2)$, the FVD algorithm running with $r$ so far behaves the same because none of the critical values tested so far is in $(r_1,r_2)$ (and the behavior of the FVD algorithm so far depends on the critical values that have been tested). This means that if $r^*\neq r_2$, then $r^*\in (r_1,r_2)$ and thus for any $r\in (r_1,r_2)$ the FVD algorithm running with $r$ so far behaves the same as that running with $r^*$. In this way, after all nodes of the path $\pi_i(r^*)$ are visited, we will reach a leaf $v\in T$ and obtain an interval $(r_1,r_2]$ containing $r^*$ and the FVD algorithm running with any $r\in (r_1,r_2)$ behaves the same; therefore, if $r^*\neq r_2$, the algorithm always reaches the same leaf $v$ when running with any $r\in (r_1,r_2)$ and $a_i^i(r^*)$ is the index of the point stored at $v$. Since $\pi_i^i(r^*)$ has $O(\log n)$ nodes, the algorithm calls the decision algorithm $O(\log n)$ times. 

If we run the above algorithm for each $p_i\in P$ individually, then we would need to call the decision algorithm $O(n\log n)$ times. To improve the algorithm, we adopt the parametric search framework by running the above algorithm in parallel for all points $p_i\in P$. At each parallel step, for each $p_i\in P$, we have at most one critical value $r_i$ to compare with $r^*$. Instead of resolving the comparison individually as above, we first pick the median $r$ of all such $r_i$'s, $1\leq i\leq n$, and then resolve the comparison between $r$ and $r^*$. After that, half of the comparisons between $r^*$ and all $r_i$'s can be resolved. As such, after calling the decision algorithm $O(\log n)$ times, we can resolve all $n$ comparisons for all $r_i$'s, $1\leq i\leq n$. We then continue with the second parallel step. After $O(\log n)$ parallel steps, we obtain an interval $(r_1,r_2]$ and for each $p_i\in P$, we reach a leaf $v^i$, 
such that $r^*\in (r_1,r_2]$, and if $r^*\neq r_2$, the algorithm for any $r\in (r_1,r_2)$ behaves the same as for $r^*$, i.e., for any $r\in (r_1,r_2)$, the algorithm always reaches the same leaf $v^i$ for all $p_i\in P$ and the index of the point stored at $v^i$ is $a_i^i(r^*)$. 

The above algorithm calls the decision algorithm $O(\log^2 n)$ times because the algorithm has $O(\log n)$ parallel steps and calls the decision algorithm $O(\log n)$ times in each step. We can further improve the algorithm so that it suffices to call the decision algorithm $O(\log n)$ times in total. This can be achieved by applying Cole's technique~\cite{ref:ColeSl87}. Indeed, the technique is applicable here because for each $p_i\in P$, we are searching along the nodes in the path $\pi_i^i(r^*)$, which satisfies the bounded fan-in/fan-out condition of the technique~\cite{ref:ColeSl87}. Using the technique, it suffices to call the decision algorithm $O(\log n)$ times, and in a total of $O(kn\log^2 n)$ time we can compute an interval $(r_1,r_2]$ and an index $a_i'$ for each $p_i\in P$, such that $r^*\in (r_1,r_2]$, and if $r^*\neq r_2$, the algorithm for any $r\in (r_1,r_2)$ behaves the same as for $r^*$, i.e., for any $r\in (r_1,r_2)$, $a_i^i(r)=a_i^i(r^*)=a_i'$ for all $p_i\in P$. 

With the interval $(r_1,r_2]$, we next apply a similar algorithm as above to compute $b_i^i(r^*)$. Similarly, by calling the decision algorithm $O(\log n)$ times in a total of $O(kn\log^2 n)$ time, we can compute a shrunken interval $(r_1,r_2]$ and a point index $b_i'$ for each $p_i\in P$, such that $r^*\in (r_1,r_2]$, and if $r^*\neq r_2$, the algorithm for any $r\in (r_1,r_2)$ behaves the same as for $r^*$, i.e., for any $r\in (r_1,r_2)$, $a_i^i(r)=a_i^i(r^*)=a_i'$ and $b_i^i(r)=b_i^i(r^*)=b_i'$ for all $p_i\in P$.

For each $p_i\in P$, we add the sublist $P(b_i^i(r^*),a_i^i(r^*))$ to $\calL_1(r^*)$. According to the above discussion, if $r^*\neq r_2$, then $\calL_1(r)=\calL_1(r^*)$ for all $r\in (r_1,r_2)$. This finishes the first iteration. 

\paragraph{The subsequent iterations.}
For each $t$-th iteration, $t\geq 2$, we assume that we have the set $\calL_{t-1}(r^*)$ and an interval $(r_1,r_2]$ containing $r^*$ such that if $r^*\neq r_2$, then $\calL_{t-1}(r)=\calL_{t-1}(r^*)$ holds for any $r\in (r_1,r_2)$.
The goal of the $t$-th iteration is to compute  a sublist set $\calL_{t}(r^*)$ and shrink the interval $(r_1,r_2]$ such that it still contains $r^*$ and if $r^*\neq r_2$, then $\calL_{t}(r)=\calL_{t}(r^*)$ holds for any $r\in (r_1,r_2)$.
The details are given below. 

For each point $p_i\in P$, we perform a counterclockwise processing procedure as follows. First, we compute the sublist in $\calL_{t-1}(r^*)$ containing $a_i^i(r^*)$ with the farthest counterclockwise endpoint. Note that since the set $\calL_{t-1}(r^*)$ is fixed, this step does not generate any critical values and thus does not need to call the decision algorithm. Let $P[j_{i1}(r^*),j_{i2}(r^*)]$ denote the sublist computed above. The next step is then to compute the index $a_{i}^{j_{i2}(r^*)+1}(r^*)$. As with the computation of $a_i^i(r^*)$, we trace a path of $O(\log n)$ nodes in the tree $T$. Processing each node in the path generates at most one critical value. 
The algorithm is similar to that for computing $a_i^i(r^*)$ discussed above.
Overall, it suffices to call the decision algorithm $O(\log n)$ times, and in a total of $O(kn\log^2 n)$ time we can compute an interval $(r_1,r_2]$ and an index $a_i'$ for each $p_{i}\in P$, such that $r^*\in (r_1,r_2]$, and if $r^*\neq r_2$, 
then for any $r\in (r_1,r_2)$, $a_i^{j_{i2}(r)+1}(r)=a_i^{j_{i2}(r^*)+1}(r^*)=a_{i}'$ for all $p_i\in P$. 
For each $p_i\in P$, we add to $\calL_t(r^*)$ the sublist that is the union of the following three sublists: $P(b_i^i(r^*),a_i^i(r^*))$, $P[j_{i1}(r^*),j_{i2}(r^*)]$, and $P(j_{i2}(r^*),a_i^{j_{i2}(r^*)+1}(r^*))$. 

Similarly, we perform a clockwise processing procedure for all points $p_i\in P$. After that, we obtain $\calL_{t}(r^*)$ and $(r_1,r_2]$ such that $r^*\in (r_1,r_2]$ and if $r^*\neq r_2$, then $\calL_{t}(r)=\calL_{t}(r^*)$ for any $r\in (r_1,r_2)$. The total time of the $t$-th iteration is thus $O(kn\log^2 n)$. 

After the $k$-th iteration, we stop the algorithm with an interval $(r_1,r_2]$. 
The next lemma argues that $r^*=r_2$. 

\begin{lemma}\label{lem:optvalue}
Suppose that $(r_1,r_2]$ is the resulting interval after the $k$-th iteration of the algorithm. Then, it holds that $r^*=r_2$.
\end{lemma}
\begin{proof}
Assume to the contrary that $r^*\neq r_2$. As $r^*\in (r_1,r_2]$, we have $r^*\in (r_1,r_2)$. Therefore, by our algorithm invariant, $\calL_k(r)=\calL_k(r^*)$ for any $r\in (r_1,r_2)$. According to our decision algorithm, $\calL_k(r^*)$ must contain a sublist that is $P$. Hence, $\calL_k(r)$ contains a sublist that is $P$ for any $r\in (r_1,r_2)$. Let $r'$ be any value in $(r_1,r^*)$. Therefore, $\calL_k(r')$ contains a sublist that is $P$. 
This means that if we apply the decision algorithm with $r=r'$, then the decision algorithm will find  a set of $k$ centers in $P$ such that the distance from any point of $P$ to its closest center is at most $r'$. But this contradicts with the definition of $r^*$ as $r'<r^*$. Therefore, $r^*=r_2$ must hold. 
\end{proof}

By Lemma~\ref{lem:optvalue}, the above algorithm correctly computes $r^*$. 
Since each iteration of the algorithm takes $O(kn\log^2 n)$ time, the total time of the algorithm is $O(k^2n\log^2 n)$. 
Combining with the $O(n^{4/3}\log n + kn\log^2 n)$ time algorithm discussed earlier in this section, we obtain the following result. 

\begin{theorem}
Given a set $P$ of $n$ points in convex position in the plane and a number $k$, we can compute in $O(\min\{n^{4/3}\log n+kn\log^2 n,k^2 n\log^2n\})$ time a subset $S \subseteq P$ of size at most $k$, such that the maximum distance from any point of $P$ to its nearest point in $S$ is minimized. 
\end{theorem}


\section{The independent set problem}
\label{sec:isconvex}

In this section, we present our algorithms for the independent set problem, assuming that the points of $P$ are in convex position.
In Section~\ref{sec:isweight}, we present the algorithm for computing a maximum-weight independent set. Section~\ref{sec:isconvex3} gives the algorithm for computing an (unweighted) independent set of size $3$. 

\subsection{The maximum-weight independent set problem}
\label{sec:isweight}
Instead of first solving the maximum independent set problem and then extending it to the weighted case, we give an algorithm for finding a maximum-weight independent set directly. 

Recall that $P=\langle p_1, p_2, \ldots, p_n\rangle$ is a cyclic list ordered along $\calH(P)$ in counterclockwise order.
For each point $p_i\in P$, let $w_i$ denote its weight. We assume that each $w_i>0$ since otherwise $p_i$ can be simply ignored, which would not affect the optimal solution. For any subset $P'\subseteq P$, let $w(P')$ denote the total weight of all points of $P'$.

For any three points $p_1,p_2,p_3$, let $D(p_1,p_2,p_3)$ denote the disk whose boundary contains them. Thus, $\partial D(p_1,p_2,p_3)$ is the unique circle through these points.

In what follows, we first describe the algorithm and explain why it is correct, and then discuss how to implement the algorithm efficiently. 

\subsubsection{Algorithm description and correctness}
To motivate our algorithm and demonstrate its correctness, we first examine the optimal solution structure and then develop a recursive relation on which our dynamic programming algorithm is based. 

Let $S$ be a maximum-weight independent set of $G(P)$, or equivalently, $S$ is a maximum-weight subset of $P$ such that the minimum pairwise distance of the points of $S$ is larger than $1$. 
Let $\dt(S)$ denote the Delaunay triangulation of $S$. If $(p,q)$ is the closest pair of points of $S$, then $\overline{pq}$ must be an edge of $\dt(S)$ and in fact the shortest edge of $\dt(S)$~\cite{ref:ShamosCl75}. As such, finding a maximum-weight independent set of $G(P)$ is equivalent to finding a maximum-weight subset $S\subseteq P$ such that the shortest edge of $\dt(S)$ has length larger than $1$. The algorithm in the previous work~\cite{ref:SingireddyAl23} is based on this observation, which also inspires our algorithm. 

\begin{figure}[t]
\begin{minipage}[t]{\textwidth}
\begin{center}
\includegraphics[height=1.7in]{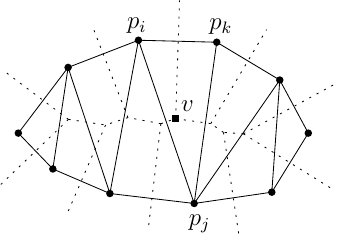}
\caption{\footnotesize Illustrating $\dt(S)$, the solid segments, and $\vd(S)$, the dotted segments. }
\label{fig:dt}
\end{center}
\end{minipage}
\vspace{-0.1in}
\end{figure}

Consider a triangle $\triangle p_ip_jp_k$ of $\dt(S)$ such that the three points $p_i,p_j,p_k$ are in the counterclockwise order of $P$ (i.e., they are counterclockwise on the convex hull $\calH(P)$). Due to the property of Delaunay triangulations, the disk $D(p_i,p_j,p_k)$ does not contain any point of $S\setminus \{p_i,p_j,p_k\}$~\cite{ref:ShamosCl75}. Since the points of $S$ are in convex position, we have the following observation. 

\begin{observation}\label{obser:dt}
$\dt(S)$ does not contain an edge connecting two points from any two different subsets of $\{P(i,j),P(j,k),P(k,i)\}$; see Figure~\ref{fig:dt}.
\end{observation}
\begin{proof}
Let $S[i,j]=P[i,j]\cap S$, $S[j,k]=P[j,k]\cap S$, and $S[k,i]=P[k,i]\cap S$. Similarly, let $S(i,j)=P(i,j)\cap S$, $S(j,k)=P(j,k)\cap S$, and $S(k,i)=P(k,i)\cap S$.

Since the points of $S$ are in convex position, the edges of the Voronoi diagram $\vd(S)$ of $S$ form a tree~\cite{ref:AggarwalA89}. Note that $\dt(S)$ is a dual graph of $\vd(S)$~\cite{ref:ShamosCl75}. Since $\triangle p_ip_jp_k$ is a triangle of $\dt(S)$, the center $v$ of the disk $D(p_i,p_j,p_k)$ is a vertex of $\vd(S)$; see Figure~\ref{fig:dt}. 
Due to our general position assumption, $v$ has three incident edges in $\vd(S)$, which are dual to the three edges of $\triangle p_ip_jp_k$. We consider $v$ the root of the tree $\vd(S)$,
which has three subtrees, defined by $S[i,j]$, $S[j,k]$, and $S[k,i]$ in the following sense~\cite{ref:AggarwalA89}: For each subtree, there is exactly one subset $S'$ of $S[i,j]$, $S[j,k]$, and $S[k,i]$ such that every edge of the subtree is neighboring to the Voronoi cells in $\vd(S)$ of two points of $S'$. As such, $\vd(S)$ does not have an edge neighboring to the Voronoi cells of two points from two different subsets of $S(i,j)$, $S(j,k)$, and $S(k,i)$. 
For each edge $\overline{pq}$ of $\dt(S)$, $\vd(S)$ must have an edge neighboring to the Voronoi cells of $p$ and $q$. Therefore, $\dt(S)$ cannot contain an edge connecting two points from two different subsets of $S(i,j)$, $S(j,k)$, and $S(k,i)$. The observation thus follows. 
\end{proof}

Observation~\ref{obser:dt} implies the following: To find an optimal solution $S$, if we know that $\triangle p_ip_jp_k$ is a triangle in $\dt(S)$, 
since no point of $S\setminus\{p_i,p_j,p_k\}$ lies in the disk $D(p_i,p_j,p_k)$, we can independently search $P(i,j)\cap \overline{D(p_i,p_j,p_k)}$, $P(j,k)\cap \overline{D(p_i,p_j,p_k)}$, and $P(k,i)\cap \overline{D(p_i,p_j,p_k)}$, respectively (recall that $\overline{D(p_i,p_j,p_k)}$ denote the region outside $D(p_i,p_j,p_k)$). This idea forms the basis of our dynamic program. 

Let $W^*$ denote the total weight of a maximum-weight independent set of $G(P)$. 

For any pair of indices $(i,j)$ with $|p_i p_j| > 1$, we call $(i,j)$ a {\em canonical pair} and define $f(i,j)$ as the total weight of a maximum-weight subset $P'$ of $P(i,j)$ such that $P'\cup \{p_i,p_j\}$ forms an independent set of $G(P)$; if no such subset $P'$ exists, then $f(i,j)=0$. Computing $f(i,j)$ is a subproblem in our dynamic program. For simplicity, we let $f(i,j)=-(w_i+w_j)$ if $(i,j)$ is not canonical, i.e., $|p_ip_j|\leq 1$. The following lemma explains why we are interested in $f(i,j)$. 

\begin{lemma}\label{lem:opt}
$W^*=\max_{1\leq i,j\leq n}(f(i,j)+w_i+w_j)$.    
\end{lemma}
\begin{proof}
Consider a pair of indices $(i,j)$. We show that $f(i,j)+w_i+w_j\leq W^*$. If $|p_ip_j|\leq 1$, then $f(i,j)=-(w_i+w_j)$ and thus $f(i,j)+w_i+w_j\leq W^*$ is obviously true. Now suppose that $|p_ip_j|> 1$. 
Let $P'$ be a maximum-weight subset of $P(i,j)$ such that $P'\cup \{p_i,p_j\}$ is an independent set. 
By definition, $f(i,j)=w(P')$. Since $P'\cup \{p_i,p_j\}$ is an independent set, we have $w(P'\cup \{p_i,p_j\})\leq W^*$.  Therefore, $f(i,j)+w_i+w_j=w(P'\cup \{p_i,p_j\})\leq W^*$. This proves $\max_{1\leq i,j\leq n}(f(i,j)+w_i+w_j)\leq W^*$.   

We next prove $W^*\leq \max_{1\leq i,j\leq n}(f(i,j)+w_i+w_j)$. Let $S$ be a maximum-weight independent set of $G(P)$. Let $\overline{p_ip_j}$ be an edge of the convex hull $\calH(S)$ of $S$. Without loss of generality, we assume that $\calH(S)$ is in the halfplane right of $\overrightarrow{p_ip_j}$. It is not difficult to see that $S\setminus\{p_i,p_j\}\subseteq P(i,j)$. Since $S$ is an independent set containing both $p_i$ and $p_j$, we have $w(S\setminus\{p_i,p_j\})\leq f(i,j)$, or equivalently, $W^*=w(S)\leq f(i,j)+w_i+w_j$. This proves $W^*\leq \max_{1\leq i,j\leq n}(f(i,j)+w_i+w_j)$.
\end{proof}

In light of Lemma~\ref{lem:opt}, to compute $W^*$, it suffices to compute $f(i,j)$ for all pairs of indices $1\leq i,j\leq n$ and the one with the largest $f(i,j)+w_i+w_j$ leads to the optimal solution. In order to compute $f(i,j)$, we define another type of subproblems that will be used in our algorithm. 

For any three points $p_i,p_j,p_k$ such that they are ordered counterclockwise in $P$ and their minimum pairwise distance is larger than $1$, we call $(i,j,k)$ a {\em canonical triple}.

For a canonical triple $(i,j,k)$, by slightly abusing the notation, we define $f(i,j,k)$ as the total weight of a maximum-weight subset $P'$ of $P(i,j)\cap \overline{D(p_i,p_j,p_k)}$ such that $P'\cup \{p_i,p_j\}$ is an independent set; if no such subset $P'$ exists, then $f(i,j,k)=0$.  For any canonical pair $(i,j)$, if we consider $p_0$ a dummy point to the left of $\overrightarrow{p_ip_j}$ and infinitely far from the supporting line of $\overline{p_ip_j}$ so that $D(p_i,p_j,p_0)$ becomes the left halfplane of $\overrightarrow{p_ip_j}$, then $f(i,j,0)$ following the above definition is exactly $f(i,j)$; for convenience, we also consider $(i,j,0)$ a canonical triple. In the following, to make the discussion more concise, we often use $f(i,j,0)$ instead of $f(i,j)$ because the way we compute $f(i,j,0)$ is consistent with the way we compute  $f(i,j,k)$ for $k\neq 0$.

For any canonical triple $(i,j,k)$, define $P_k(i,j)=\{p\ |\ p\in P(i,j), p\not\in D(p_i,p_j,p_k), |pp_i|> 1, |pp_j|> 1\}$. For any canonical pair $(i,j)$, define $P_0(i,j)=\{p\ |\ p\in P(i,j), |pp_i|> 1, |pp_j|> 1\}$. Note that $P_0(i,j)$ is consistent with $P_k(i,j)$ if we consider $p_0$ a dummy point as defined above. 
Observe also that $P_k(i,j)=P_0(i,j)\cap \overline{D(p_i,p_j,p_k)}$ for any canonical triple $(i,j,k)$.
By definition, $f(i,j,k)$ (including the case $k=0$) is the total weight of a maximum-weight independent set $P'\subseteq P_k(i,j)$; this is the reason we introduce the notation $P_k(i,j)$.

The following lemma gives the recursive relation of our dynamic programming algorithm. 


\begin{lemma}\label{lem:deprel}
For any canonical triple $(i,j,k)$, including the case $k=0$, the following holds (see Figure~\ref{fig:recurrelation}):
\begin{equation}
\label{eq:DepReltriple}
f(i, j, k) = \begin{cases}
\max_{p_l\in P_k(p_i,p_j)} (f(i,l,j) + f(l,j,i) + w_l),& \text{if $P_k(i,j)\neq \emptyset$}\\
0,& \text{otherwise.}
\end{cases}
\end{equation}
\end{lemma}
\begin{proof}
If $k=0$, recall that $D(p_i,p_j,p_k)$ is essentially the halfplane left of $\overrightarrow{p_ip_j}$ and thus $\overline{D(p_i,p_j,p_k)}$ refers to the halfplane right of $\overrightarrow{p_ip_j}$. With this convention, our discussions below are applicable to both $k=0$ and $k\neq 0$. 

Recall that $f(i,j,k)$ is the total weight of a maximum-weight independent set $P'\subseteq P_k(i,j)$. If $P_k(i,j)=\emptyset$, then it is vacuously true that $f(i,j,k)=0$. In what follows, we assume that $P_k(i,j)\neq \emptyset$. 

\begin{figure}[t]
\begin{minipage}[t]{\textwidth}
\begin{center}
\includegraphics[height=1.7in]{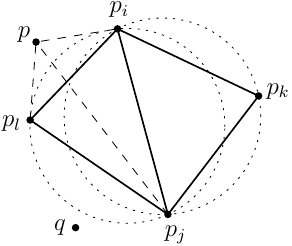}
\caption{\footnotesize Illustrating the proof of Lemma~\ref{lem:deprel}. }
\label{fig:recurrelation}
\end{center}
\end{minipage}
\vspace{-0.1in}
\end{figure}

Let $P'\subseteq P_k(i,j)$ be a maximum-weight independent set. By definition, $f(i,j,k)=w(P')$. We claim that $P'$ must have a point $p_l$ such that the disk $D(p_i,p_l,p_j)$ does not contain any point of $P'\setminus\{p_l\}$. To see this,  for any two points $p,q\in P_k(i,j)$, due to our general position assumption, if $D(p_i,p,p_j)$ contains $q$, then $D(p_i,q,p_j)$ cannot contain $p$. Therefore, such a point $p_l$ must exist (see Figure~\ref{fig:recurrelation}). We partition $P'\setminus\{p_l\}$ into two subsets: $P_1'=P'\cap P(i,l)$ and $P'_2=P'\cap P(l,j)$. Since no point of $P_1'$ lies in $D(p_i,p_l,p_j)$, $P'_1$ is a subset of $P(i,l)\cap \overline{D(p_i,p_l,p_j)}$ and $P'_1\cup \{p_i,p_l\}$ is an independent set (since $P_1'\cup \{p_l\}\subseteq P'$ and $P'\cup \{p_i,p_j\}$ is an independent set). By definition, $f(i,l,j)$ is the total weight of a maximum-weight subset $P''$ of $P(i,l)\cap \overline{D(p_i,p_l,p_j)}$ such that $P''\cup \{p_i,p_l\}$ is an independent set. As such, we have $w(P_1')\leq f(i,l,j)$. Analogously, $w(P_2')\leq f(l,j,i)$. Therefore, $w(P_1'\cup P_2')\leq f(i,l,j) + f(l,j,i)$ holds. Since $f(i,j,k)=w(P')=w(P_1'\cup P_2')+w_l$, we obtain that $f(i,j,k)\leq f(i,l,j) + f(l,j,i)+w_l$, implying that $f(i, j, k) \leq \max_{p_l\in P_k(p_i,p_j)} (f(i,l,j) + f(l,j,i) + w_l)$.

Next, we argue that $\max_{p_l\in P_k(p_i,p_j)} (f(i,l,j) + f(l,j,i) + w_l)\leq f(i, j, k)$. It suffices to show that $f(i,l,j) + f(l,j,i) + w_l \leq f(i, j, k)$ for an arbitrary point $p_l\in P_k(i,j)$. Let $P'(i,l)$ be a maximum-weight subset of $P(i,l)\cap \overline{D(p_i,p_l,p_j)}$ such that $P'(i,l)\cup \{p_i,p_l\}$ is an independent set. Let $P'(l,j)$ be a maximum-weight subset of $P(l,j)\cap \overline{D(p_i,p_l,p_j)}$ such that $P'(l,j)\cup \{p_l,p_j\}$ is an independent set. By definition, $f(i,l,j) = w(P'(i,l))$ and $f(l,j,i) = w(P'(l,j))$. 

We claim the following: (1) $P'(i,l)\cup P'(l,j)\cup \{p_l\}\subseteq P_k(i,j)$; (2) $P'(i,l)\cup P'(l,j)\cup \{p_l\}$ is an independent set. We prove the claim in the following.
\begin{enumerate}
\item 
To prove (1), consider any point $p\in P'(i,l)\cup P'(l,j)\cup \{p_l\}$. Our goal is to prove $p\in P_k(i,j)$.
If $p=p_l$, we already know $p_l\in P_k(i,j)$. Now assume $p\in P'(i,l)$; the case $p\in P'(l,j)$ can be argued analogously. See Figure~\ref{fig:recurrelation}.

Since $p\in P'(i,l)$ and $P'(i,l)\subseteq P(i,j)$, we have $p\in P(i,j)$. To prove $p\in P_k(i,j)$, we need to argue the following: $p\not\in D(p_i,p_j,p_k), |pp_i|> 1, |pp_j|> 1$. Since $P'(i,l)\cup \{p_i,p_l\}$ is an independent set, we have $|pp_i|> 1$. To show $|pp_j|> 1$, notice that $p,p_l,p_j,p_i$ are on their convex hull in counterclockwise order (see Figure~\ref{fig:recurrelation}). 
Furthermore, since $p$ is outside $D(p_i,p_l,p_j)$, one of the two angles $\angle pp_ip_j$ and $\angle pp_lp_j$ must be greater than $60\degree$. Without loss of generality, we assume that $\angle pp_ip_j>60\degree$. Then $|pp_j|\geq \min \{|pp_i|,|p_ip_j|\}$. Since both $|pp_i|$ and $|p_ip_j|$ are larger than $1$, we can derive $|pp_j|> 1$. 


It remains to show that $p\not\in D(i,j,k)$.
By the definition of $P'(i,l)$, all the points of $P'(i,l)$ are outside $D(p_i,p_l,p_j)$. Therefore, $p\not\in D(p_i,p_l,p_j)$. 
Let $H$ denote the open halfplane on the right side of $\overrightarrow{p_ip_j}$. By definition, $p_l\in H$ while $p_k\not\in H$. As $p_l\not\in D(p_i,p_j,p_k)$, we have $H\cap D(p_i,p_j,p_k)\subseteq D(p_i,p_l,p_j)$. Since $p\in H$ and $p\not\in D(p_i,p_l,p_j)$, we obtain $p\not\in D(p_i,p_j,p_k)$. This proves (1).

\item 
We now prove (2), i.e., prove $|pq|> 1$ for any two points $p,q\in P'(i,l)\cup P'(l,j)\cup \{p_l\}$. If $p$ and $q$ are both from $P'(i,l)\cup \{p_l\}$, then since $P'(i,j)\cup \{p_i,p_l\}$ is an independent set, $|pq|> 1$ holds. Similarly, if $p$ and $q$ are both from $P'(l,j)\cup \{p_l\}$, $|pq|> 1$ also holds. The remaining case is when one of the two points, say $p$, is from $P'(i,l)$ while the other point $q$ is from $P'(l,j)$ (see Figure~\ref{fig:recurrelation}). In this case, by Lemma~\ref{lem:rem}, $|pq|\geq \min\{|pp_i|,|pp_l|,|p_ip_j|,|p_lp_j|, |p_lp_i|, |qp_l|, |qp_j|\}$. Since all seven distances on the right-hand side of the above inequality are larger than $1$, we obtain $|pq|> 1$. This proves (2). 
\end{enumerate}

This proves that $P'(i,l)\cup P'(l,j)\cup \{p_l\}\subseteq P_k(i,j)$ and $P'(i,l)\cup P'(l,j)\cup \{p_l\}$ is an independent set. By the definition of $f(i,j,k)$, we have $w(P'(i,l)) + w(P'(l,j)) + w_l \leq f(i,j,k)$. Recall that $f(i,l,j) = w(P'(i,l))$ and $f(l,j,i) = w(P'(l,j))$. Therefore, we obtain $f(i,l,j) + f(l,j,i) + w_l \leq f(i,j,k)$.

The lemma thus follows. 
\end{proof}

\begin{figure}[t]
\begin{minipage}[t]{\textwidth}
\begin{center}
\includegraphics[height=1.7in]{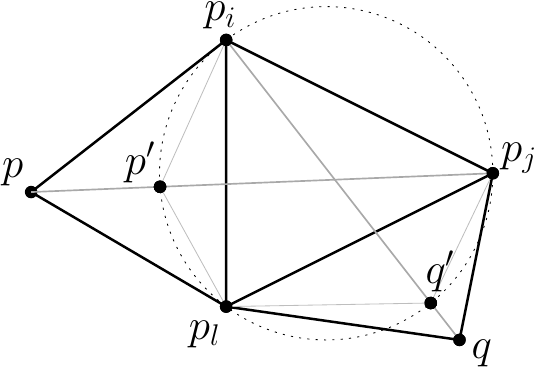}
\caption{\footnotesize Illustrating the proof of Lemma~\ref{lem:rem}. }
\label{fig:rem}
\end{center}
\end{minipage}
\vspace{-0.1in}
\end{figure}

\begin{lemma}
\label{lem:rem}
For any triple of points $p_i,p_l,p_j\in P$ that are ordered counterclockwise on $\calH(P)$, for any two points $p$ and $q$ with $p\in P(i,l)$, $q\in P(l,j)$, and $p,q\not\in D(p_i,p_l,p_j)$, it follows that 
$$|pq|\geq \min \{|pp_i|,|pp_l|,|p_ip_j|,|p_lp_j|, |p_lp_i|, |qp_l|, |qp_j|\}.$$
\end{lemma}

\begin{proof}
We first make the following {\em observation}: Let $a, b, c$ be vertices of a triangle. If $|ac|<\min\{|ab|, |bc|\}$, then the angle $\angle abc < 60\degree $. 

Let $p'$ denote the intersection point between the circle $C=\partial D(p_i, p_j, p_l)$ and the segment $\overline{pp_j}$, and $q'$ the intersection point between $C$ and $\overline{p_iq}$ (see Figure \ref{fig:rem}). Note that since $p,q\not\in D(p_i,p_l,p_j)$, both $p'$ and $q'$ must exist. We assume to the contrary that 
\begin{equation}\label{equ:contrary}
|pq|< \min \{|pp_i|,|pp_l|,|p_ip_j|,|p_lp_j|, |p_lp_i|, |qp_l|, |qp_j|\}.    
\end{equation}

Since $|pq|<\min\{|pp_l|,|qp_l|\}$ due to \eqref{equ:contrary}, we have $\angle pp_lq < 60\degree$ by the above observation. 
Consequently, since $p, p_l,q, p_j, p_i$ are counterclockwise in $P$, we have $\angle p_ip_lq'<\angle p_ip_lq < \angle pp_lq<  60\degree$ and $\angle p'p_lp_i < \angle pp_lp_i <\angle pp_lq < 60\degree$. 

We claim that $|pq|<|p_iq|$. To see this, we first notice that a function $g_{p_i}(x)$ representing the distance from $p_i \in C$ to another point $x\in C$ is \textit{unimodal} with the maximum value achieved when $\overline{p_ix}$ is a diameter of $C$. Therefore, $|p_iq'|>\min\{|p_ip_l|,|p_ip_j|\}$. Since  $\min\{|p_ip_l|,|p_ip_j|\}>|pq|$ due to \eqref{equ:contrary}, we obtain 
$|p_iq'|>|pq|$. As $|p_iq|\geq|p_iq'|$, $|p_iq|>|pq|$ follows.

Since $|pq|<|pp_i|$ by \eqref{equ:contrary}, we have $|pq|<\min\{|pp_i|,|p_iq|\}$, which leads to $\angle pp_iq < 60\degree$ due to the above observation. Analogously, it can be shown that $\angle pp_jq < 60\degree$.
    
From the quadrilateral $p_lq'p_jp_i$, we have  $\angle p_ip_lq' + \angle pp_jq' + \angle p_ip_jp = 180\degree$ by the sum opposite angles of circumscribed quadrilateral. Since $\angle pp_jq'\leq \angle pp_jq$, we can derive
\begin{equation}\label{equ:contrary10}
\angle p_ip_lq' + \angle pp_jq + \angle p_ip_jp \geq 180\degree.    
\end{equation}

Note that $\angle p_ip_jp= \angle p'p_lp_i$ due to the equality of angles subtended by the same arc. We also have proved above that $\angle p'p_lp_i < 60\degree$. Therefore, $\angle p_ip_jp < 60\degree$ holds. Since $\angle pp_jq < 60\degree$, we can derive $\angle p_ip_lq'>60\degree$ from \eqref{equ:contrary10}. However, we have already showed above that $\angle p_ip_lq' < 60\degree$, a contradiction.
\end{proof}

With Lemma~\ref{lem:deprel}, it remains to find an order to solve the subproblems so that when computing $f(i,j,k)$, the values $f(i,l,j)$ and $f(l,j,i)$ for all $p_l\in P_k(p_i,p_j)$ are available. 

For any two points $p_i,p_j\in P$, we call $\overline{p_ip_j}$ a {\em diagonal}. 

We process the diagonals $\overline{p_ip_j}$ for all $1\leq i, j\leq n$ in the following way. For each $j=2,\ldots, n$ in this order, we enumerate $i=j-1,j-2,\ldots, 1$ to process $\overline{p_ip_j}$ as follows. If $|p_ip_j|\leq 1$, then we set $f(i,j)=-(w_i+w_j)$. Otherwise, $(i,j)$ is a canonical pair, and we compute $f(i,j)$, i.e., $f(i,j,0)$, by Equation~\eqref{eq:DepReltriple}; one can check that the values $f(i,l,j)$ and $f(l,j,i)$ for all $p_l\in P_0(p_i,p_j)$ have already been computed. Next, for each point $p_k\in P(j,i)$ with $|p_ip_k|> 1$ and $|p_jp_k|> 1$, $(i,j,k)$ is a canonical triple and we compute $f(i,j,k)$ by Equation~\eqref{eq:DepReltriple}; again, the values $f(i,l,j)$ and $f(l,j,i)$ for all $p_l\in P_k(p_i,p_j)$ have already been computed. Finally, by Lemma~\ref{lem:opt}, we can return the largest $f(i,j)+w_i+w_j$ among all canonical pairs $(i,j)$ as $W^*$. Note that the algorithm only computes the value $W^*$, but by the standard back-tracking technique an optimal solution (i.e., an actual maximum-weight independent set) can also be obtained. 

\subsubsection{Algorithm implementation}
\label{sec:cutting}

We can easily implement the algorithm in $O(n^4)$ time. Indeed, there are $O(n^3)$ subproblems $f(i,j,k)$. Each subproblem can be computed in $O(n)$ time by checking every point $p_l\in P_k(i,j)$. As such, the total time is bounded by $O(n^4)$. In what follows, we provide a more efficient $O(n^{7/2})$ time implementation by computing every subproblem faster using cuttings~\cite{ref:ChazelleCu93}. 

Specifically, we show that for each canonical pair $(i,j)$ we can compute the subproblems $f(i,j,k)$ for all $p_k\in P(j,i)$ in a total of $O(n^{3/2})$ time. 
To this end, we reduce the problem to an {\em offline outside-disk range max-cost query} problem. For each point $p_l\in P_k(i,j)$, we define the {\em cost} of $p_l$ as  $cost(p_l)=f(i,l,j)+f(l,j,i)+w_l$. Recall that $P_0(i,j)=\{p\ |\ p\in P(i,j), |pp_i|> 1, |pp_j|> 1\}$ and $P_k(i,j)=P_0(i,j)\cap \overline{D(p_i,p_j,p_k)}$. As such, computing $f(i,j,k)$ is equivalent to finding the maximum-cost point of $P_0(i,j)$ outside the query disk $D(p_i,p_j,p_k)$. 
Our goal is to answer all such disk queries for all $p_k\in P(j,i)$.\footnote{It is sufficient to consider only those $k$ with $|p_ip_k|\geq 1$ and $|p_jp_k|\geq 1$. Our algorithm simply handles all $p_k\in P(j,i)$.}
We note that this problem can be solved in $O(n^{15/11})$ expected time by applying the recent randomized algorithm of Agarwal, Ezra, and Sharir~\cite{ref:AgarwalSe24}. 
In what follows, we present a deterministic algorithm of $O(n^{3/2})$ time using cuttings~\cite{ref:ChazelleCu93}. 

\paragraph{Cuttings.}
Define $\mathcal{C}$ as the set of circles $\partial D(p_i,p_j,p_k)$ for all $p_k\in P(j,i)$.
Let $m=|\calC|\leq n$. For any compact region $\sigma$ in the plane, let $\mathcal{C}_\sigma$ denote the set of circles that cross the interior of $\sigma$. For a parameter $r$ with $1 \leq r \leq m$, a {\em $(1/r)$-cutting} for $\mathcal{C}$ is a collection $\Xi$ of constant-complexity cells with disjoint interiors whose union covers the entire plane such that the interior of every cell $\sigma\in\Xi$ is intersected by at most $m/r$ circles of $\mathcal{C}$, i.e.,  $|\mathcal{C}_\sigma|\leq m/r$ ($\calC_{\sigma}$ is often called the {\em conflict list} in the literature). The size of $\Xi$ is the number of cells of $\Xi$. 

We say that a cutting $\Xi'$ {\em $c$-refines} a cutting $\Xi$ if every cell of $\Xi'$ is completely contained in a single cell of $\Xi$ and every cell of $\Xi$ contains at most $c$ cells of $\Xi'$. A {\em hierarchical $(1/r)$-cutting} for $\mathcal{C}$ is a sequence of cuttings $\Xi_0,\Xi_1, ..., \Xi_t$, where $\Xi_i$ $c$-refines $\Xi_{i-1}$ for each $1\leq i\leq t$, for some constant $c$, and every $\Xi_i$, $1\leq i\leq t$, is a $(1/\rho^i)$-cutting of size $O(\rho^{2i})$, for some constant $\rho$, and $\Xi_0$ consists of a single cell that is the entire plane. Setting $t=\lceil \log_\rho{r}\rceil$ making the last cutting $\Xi_t$ a $(1/r)$-cutting of size $O(r^2)$. If a cell $\sigma\in\Xi_i$ is contained in a cell $\sigma'\in\Xi_{i-1}$, we consider $\sigma'$ the parent of $\sigma$ and $\sigma$ a child of $\sigma'$. As such, the cells in the hierarchical cutting form a tree structure, with the only cell of $\Xi_0$ as the root. 

\begin{figure}[t]
\begin{minipage}[t]{\textwidth}
\begin{center}
\includegraphics[height=1.0in]{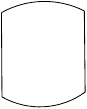}
\caption{\footnotesize Illustrating a pseudo-trapezoid.}
\label{fig:pseudotrap}
\end{center}
\end{minipage}
\vspace{-0.1in}
\end{figure}

A hierarchical $(1/r)$-cutting for $\mathcal{C}$ can be computed in $O(mr)$ time, for any $1\leq r\leq m$, e.g., by the algorithm \cite{ref:WangUn23}, which adapts Chazelle's algorithm for hyperplanes \cite{ref:ChazelleCu93}. The algorithm also produces the conflict lists $\mathcal{C}_\sigma$ for all cells $\sigma\in\Xi_i$ for all $0\leq i \leq t$, which implies that the total size of all conflict lists is $O(mr)$, i.e., $\sum_{i=0}^{t}\sum_{\sigma\in\Xi_i} |\mathcal{C}_\sigma|= O(mr)$. In particular, each cell of the cutting produced by the algorithm of \cite{ref:WangUn23} is a (possibly unbounded) {\em pseudo-trapezoid} that typically has two vertical line segments as left and right sides, an arc of a circle of $\calC$ as a top side (resp., bottom side) (see Figure~\ref{fig:pseudotrap}).

The following lemma gives the algorithm. 

\begin{lemma}\label{lem:SubProb}
Suppose that $cost(p_l)$ for all $p_l\in P_k(i,j)$ are available. Computing the maximum-cost point of $P_0(i,j)$ outside the disk $D(p_i,p_j,p_k)$ for all $p_k\in P(j,i)$ can be done in $O(n^{3/2})$ time.
\end{lemma}
\begin{proof}
We first construct a hierarchical $(1/r)$-cutting $\Xi_0, ..., \Xi_t$ for $\mathcal{C}$, which takes $O(mr)$ time~\cite{ref:ChazelleCu93,ref:WangUn23}. Let $\Xi$ denote the set of all cells $\sigma\in \Xi_i$, for all $0\leq i\leq t$. Note that $\Xi$ has $O(r^2)$ cells. The algorithm also produces conflict lists $\calC_{\sigma}$ for all cells $\sigma\in \Xi$. 

For notational convenience, let $Q=P_0(i,j)$. For each cell $\sigma\in \Xi$, let $Q(\sigma)$ denote the subset of points of $Q$ inside $\sigma$; denote by $cost(\sigma)$ the maximum cost of all points of $Q(\sigma)$ (if $Q(\sigma)=\emptyset$, we let $cost(\sigma)=0$). We wish to compute $cost(\sigma)$ for all cells $\sigma\in \Xi$. This can be done in $O(n\log r+r^2)$ time as follows. Initially, we set $cost(\sigma)=0$ for all cells $\sigma\in \Xi$, which takes $O(r^2)$ time as $\Xi$ has $O(r^2)$ cells. Then, for each point $p\in Q$, we process it by a point location procedure as follows. Starting from the single cell of $\Xi_0$, suppose that $\sigma_i$ is the cell of $\Xi_i$ containing $p$. We update $cost(\sigma_i)=cost(p)$ if $cost(\sigma_i) < cost(p)$. We then find the child $\sigma_{i+1}\in \Xi_{i+1}$ of $\sigma_i$ that contains $p$ and continue the process; note that this takes $O(1)$ time as each cell of $\Xi$ has $O(1)$ children. After processing all points of $Q$ as above, $cost(\sigma)$ for all cells $\sigma\in \Xi$ are correctly computed. As $t=O(\log r)$, processing each point of $Q$ takes $O(\log r)$ time and thus processing all points of $Q$ takes $O(n\log r)$ time. As such, computing $cost(\sigma)$ for all cells $\sigma\in \Xi$ can be done in $O(n\log r+r^2)$ time in total. Note that the above algorithm can also compute the set $Q(\sigma)$ for all cells $\sigma$ in the last cutting $\Xi_t$. 

Let $\calD$ be the set of all disks bounded by the circles of $\calC$. For each disk $D\in \calD$, let $cost(D)$ denote the largest cost of all points of $Q$ outside the disk $D$. Our goal is to compute $cost(D)$ for all disks $D\in \calD$. We proceed as follows. First, we initialize $cost(D)=0$ for all $D\in \calD$. Then, for each $0\leq i\leq t-1$, for each cell $\sigma\in \Xi_i$, for each circle in $\calC_{\sigma}$, for each child $\sigma'$ of $\sigma$, if the bounded disk $D$ of the circle does not intersect $\sigma'$, meaning that $\sigma'$ is completely outside $D$, then we update $cost(D)=cost(\sigma')$ if $cost(D)<cost(\sigma')$. In addition, for each cell $\sigma$ in the last cutting $\Xi_t$, for each circle in $\calC_{\sigma}$, for each point $p\in Q(\sigma)$, if $p$ is outside the bounded disk $D$ of the circle, then we update $cost(D)=cost(p)$ if $cost(D)<cost(p)$. One can verify that $cost(D)$ is computed correctly for all $D\in \calD$. For the runtime of the algorithm, processing cells of $\Xi_i$ for all $0\leq i\leq t-1$ takes $O(mr)$ time since each cell has $O(1)$ children and the total size of the conflict lists of all cells of $\Xi$ is $O(mr)$. Since $|\calD_{\sigma}|\leq m/r$ for each cell $\sigma\in \Xi_t$ and $\sum_{\sigma\in \Xi_t}|Q(\sigma)|=|Q|\leq n$ (this is because the sets $Q(\sigma)$ for all cells $\sigma\in \Xi_t$ are pairwise disjoint), processing all cells $\sigma\in \Xi_t$ takes $O(nm/r)$ time. 

As such, the total time of the algorithm is $O(nm/r + n\log r + r^2+mr)$. Setting $r=\sqrt{m}$ leads to the lemma as $m\leq n$. 
\end{proof}

Plugging Lemma~\ref{lem:SubProb} into our dynamic programming algorithm leads to an algorithm of $O(n^{7/2})$ time. As discussed above, using the randomized result of \cite{ref:AgarwalSe24}, the problem can be solved in $O(n^{37/11})$ expected time. The following theorem summarizes the result. 

\begin{theorem}\label{theo:convexindset}
Given a set $P$ of $n$ weighted points in convex position in the plane, a maximum-weight independent set in the unit-disk graph of $P$ can be computed in $O(n^{7/2})$ deterministic time, or in $O(n^{37/11})$ randomized expected time. 
\end{theorem}

The following corollary will be used in Section~\ref{sec:generaldispersion} to solve the dispersion problem. 

\begin{corollary}\label{coro:10}
Given a set $P$ of $n$ points in convex position in the plane and a number $r>0$, one can find in $O(n^{7/2})$ deterministic time or in $O(n^{37/11})$ randomized expected time a maximum subset of $P$ such that the distance of every two points of the subset is larger than $r$.     
\end{corollary}
\begin{proof}
We redefine the unit-disk graph of $P$ using the parameter $r$: Two points of $P$ have an edge in the graph if their distance is at most $r$. Then, we assign a weight $1$ to every point of $P$ and apply the algorithm of Theorem~\ref{theo:convexindset}. 
\end{proof}




\subsection{Computing an independent set of size $\boldsymbol{3}$}
\label{sec:isconvex3}
To facilitate the discussion in Section~\ref{sec:3dispersion} for the dispersion problem, we consider the following problem: Given a set $P$ of $n$ points in convex position and a number $r>0$, find three points from $P$ whose minimum pairwise distance is larger than or equal to $r$. 

We follow the notation in Section~\ref{sec:pre}. In particular, $P=\langle p_1, p_2, \ldots, p_n\rangle$ is a cyclic list ordered along $\calH(P)$ in counterclockwise order. In the following definition, for each point $p_i\in P$, we define $a_i$ in a similar way to $a_i^i$ in Definition~\ref{def:ab} with respect to $r$ (i.e., change ``$>1$'' to ``$\geq r$''). 

\begin{definition}
For each point $p_i\in P$, define $a_i$ as the index of the first point $p$ of $P$ counterclockwise from $p_i$ such that $|p_ip|\geq r$; similarly, define $b_i\in P$ as the index of the first point $p$ clockwise from $P$ such that $|p_ip|\geq r$. If $|p_ip|< r$ for all points $p\in P$, then let $a_i=b_i=0$.   
\end{definition}


We will make use of the following lemma, which has been proved previously in \cite{ref:KobayashiAn22}.

\begin{lemma}{\em(\cite{ref:KobayashiAn22})}
\label{lem:3points}
$P$ has three points whose minimum pairwise distance is at least $r$ if and only if there exists a point $p_i\in P$ such that $P[a_i,b_i]$ has two points whose distance is at least $r$.
\end{lemma}
\begin{proof}
To make the paper more self-contained, we sketch the proof here; see \cite[Lemma 3]{ref:KobayashiAn22} for details. If $P$ has a subset of three points $\{p_i,p_j,p_k\}$ that is a feasible solution, i.e, the minimum pairwise distance of the three points is larger than $r$, then by the definitions of $a_i$ and $b_i$, both $p_j$ and $p_k$ must be in $P[a_i,b_i]$. 

On the other hand, suppose that for a point $p_i\in P$, $P[a_i,b_i]$ has two points $p_j$ and $p_k$ with $|p_jp_k|\geq r$. We assume that $p_{a_i},p_j,p_k,p_{b_i}$ are ordered counterclockwise in $P$ if $\{p_{a_i},p_{b_i}\}\neq \{p_j,p_k\}$. Then, there are four cases: (1) $|p_ip_j|\geq r$ and $|p_ip_k|\geq r$; (2) $|p_ip_j|< r$ and $|p_ip_k|< r$; (3) $|p_ip_j|< r$ and $|p_ip_k|\geq r$; (4) $|p_ip_j|\geq r$ and $|p_ip_k|< r$. In (1), $\{p_i,p_j,p_k\}$ is a feasible solution. In (2), it can be proved that $\{p_i,p_{a_i},p_{b_i}\}$ is a feasible solution. In (3), it can be proved that $\{p_i,p_{a_i},p_k\}$ is a feasible solution. In (4), it can be proved that $\{p_i,p_{j},p_{b_i}\}$ is a feasible solution. 
\end{proof}

If $p_i$ is a point of $P$ such that $P[a_i,b_i]$ has two points whose distance is at least $r$, we say that $p_i$ is a {\em feasible point}. By Lemma~\ref{lem:3points}, it suffices to find a feasible point (if it exists). 
Our algorithm comprises two procedures. In the first procedure, we compute $a_i$ and $b_i$ for all points $p_i\in P$. This can be done in $O(n\log n)$ time by slightly changing the algorithm of Lemma~\ref{lem:firstout} (e.g., change ``$\leq 1$'' to ``$<r$'' and change ``$> 1$'' to ``$\geq r$''). The second procedure finds a feasible point. In the following, we present an $O(n\log n)$ time algorithm. 
We start with the following easy but crucial observation. 

\begin{observation}\label{obser:20}
A point $p_i\in P$ is a feasible point if and only if there is a point $p_k\in P[a_i,b_i]$ such that $p_{a_k}$ is also in $P[a_i,b_i]$ and $(p_i,p_k,p_{a_k})$ is in counterclockwise order in $P$. 
\end{observation}
\begin{proof}
If there is a point $p_k\in P[a_i,b_i]$ such that $p_{a_k}$ is also in $P[a_i,b_i]$, then by the definition of $a_k$, $|p_kp_{a_k}|\geq r$ holds. Since both $p_k$ and $p_{a_k}$ are in $P[a_i,b_i]$, by Lemma~\ref{lem:3points}, $p_i$ is a feasible point.

On the other hand, suppose that $p_i$ is a feasible point. Then, $P[a_i,b_i]$ has two points whose distance is at least $r$. Let these two points be $p_k$ and $p_{k'}$ so that $(p_i,p_k,p_{k'})$ is in counterclockwise order in $P$. This implies that $a_k$ must be an index of a point in $P[k+1,k']\subseteq P[a_i,b_i]$. Therefore, $p_{a_k}$ must be a point in $P[a_i,b_i]$ and $(p_i,p_k,p_{a_k})$ is in counterclockwise order in $P$.
\end{proof}

For each point $p_i\in P$, note that $a_i\neq i$ must hold; we define
$$a_i'=\begin{cases}
    a_i, & \text{if $i<a_i$}\\
    a_i+n, & \text{otherwise.}
\end{cases}$$
By definition, $i< a_i'$ always holds and $a_i'=a_i$ if $a_i'\leq n$. Note that if $a_i=b_i$, then $P[a_i,b_i]$ has only one point, and therefore $p_i$ cannot be a feasible point. As such, we only need to focus on the points $p_i$ with $a_i\neq b_i$. Our algorithm is based on the following lemma, which in turn relies on Observation~\ref{obser:20}. 

\begin{lemma}\label{lem:feasiblepoint}
For each $p_i\in P$, we have the following. 
\begin{enumerate}
    \item If $a_i< b_i$, then $p_i$ is a feasible point if and only if $\min_{k\in [a_i,b_i]}a_k'\leq b_i$. 
    \item If $a_i>b_i$, then $p_i$ is a feasible point if and only if $\min_{k\in [a_i,n]}a_k'\leq b_i+n$ or $\min_{k\in [1,b_i]}a_k'\leq b_i$. 
\end{enumerate}
\end{lemma}
\begin{proof}
We start with the first case $a_i<b_i$. Suppose that $p_i$ is a feasible point. Then, by Observation~\ref{obser:20}, there is a point $p_k\in P[a_i,b_i]$ such that $p_{a_k}$ is also in $P[a_i,b_i]$ and $(p_i,p_k,p_{a_k})$ is in counterclockwise order. Since $a_i<b_i$, we have $1\leq a_i\leq k<a_k\leq b_i\leq n$. As such, $a_k'=a_k$ and $a_k'\leq b_i$. Therefore, $\min_{k\in [a_i,b_i]}a_k'\leq b_i$ must hold. 

On the other hand, suppose that $\min_{k\in [a_i,b_i]}a_k'\leq b_i$. Then, there exists $k\in [a_i,b_i]$ such that $a_k'\leq b_i$. To show that $p_i$ is a feasible point, it suffices to prove $p_{a_k}\in P[a_i,b_i]$, which is equivalent to proving $a_i\leq a_k\leq b_i$ since $a_i<b_i$. Since $a_k'\leq b_i\leq n$, we obtain $a_k=a_k'\leq b_i$. Furthermore, since $k< a_k'$, we have $k\leq a_k$. As $a_i\leq k$, we obtain $a_i\leq a_k$. Therefore, $a_i\leq a_k\leq b_i$ holds and thus $p_i$ is a feasible point. 

\paragraph{The second case $a_i>b_i$.}
In this case, $P[a_i,b_i]=P[a_i,n]\cup P[1,b_i]$. 

Suppose that $p_i$ is a feasible point. Then, by Observation~\ref{obser:20}, there is a point $p_k\in P[a_i,b_i]$ such that $p_{a_k}$ is also in $P[a_i,b_i]$ and $(p_i,p_k,p_{a_k})$ is in counterclockwise order in $P$. There are two cases depending on whether $p_k\in P[a_i,n]$ or $p_k\in P[1,b_i]$. 

\begin{enumerate}
    \item If $p_k\in P[a_i,n]$, then $a_i\leq k\leq n$ and there are two subcases depending on whether $p_{a_k}\in P[a_i,n]$ or $p_{a_k}\in P[1,b_i]$. If $p_{a_k}\in P[a_i,n]$, then $a_k'=a_k\leq n<b_i+n$, implying that $\min_{k\in [a_i,n]}a_k'\leq b_i+n$ as  $a_i\leq k\leq n$. If $p_{a_k}\in P[1,b_i]$, then we have $1\leq a_k\leq b_i$. Therefore, $a_k'\leq n+a_k\leq b_i+n$, again implying $\min_{k\in [a_i,n]}a_k'\leq b_i+n$.

    \item If $p_k\in P[1,b_i]$, then $1\leq k\leq b_i$. Since $p_{a_k}$ is also in $P[a_i,b_i]$ and $(p_i,p_k,p_{a_k})$ is in counterclockwise order in $P$, we have $a_k\in P[k+1,b_i]$, Therefore, $k\leq a_k\leq b_i$, implying that $a_k'=a_k\leq b_i$. As such, we obtain $\min_{k\in [1,b_i]}a_k'\leq b_i$ as $1\leq k\leq b_i$.
\end{enumerate}

On the other hand, suppose that $\min_{k\in [a_i,n]}a_k'\leq b_i+n$ or $\min_{k\in [1,b_i]}a_k'\leq b_i$ holds. We argue  that $p_i$ must be a feasible point. 
\begin{enumerate}
    \item If $\min_{k\in [a_i,n]}a_k'\leq b_i+n$, there exists $k\in [a_i,n]$ such that $a_k'\leq b_i+n$.  To show that $p_i$ is a feasible point, it suffices to prove that $p_{a_k}\in P[a_i,b_i]=P[a_i,n]\cup P[1,b_i]$. By definition, $a_k'$ is either $a_k$ or $a_k+n$. Recall that $k<a_k'$ always holds. If $a_k'=a_k$, then we have $a_i\leq k<a_k'=a_k\leq n$. Hence, we obtain $p_{a_k}\in P[a_i,n]\subseteq P[a_i,b_i]$. If $a_k'=a_k+n$, then since $a_k'\leq b_i+n$, we have $a_k\in [1,b_i]$ and thus $p_{a_k}\in P[1,b_i]\subseteq P[a_i,b_i]$. 

    \item If $\min_{k\in [1,b_i]}a_k'\leq b_i$, there exists $k\in [1,b_i]$ such that $a_k'\leq b_i$. As above, it suffices to prove that $p_{a_k}\in P[a_i,b_i]=P[a_i,n]\cup P[1,b_i]$. Since $a_k'$ is either $a_k$ or $a_k+n$, and $a_k'\leq b_i\leq n$, $a_k'$ must be $a_k$. Also, recall that $k<a_k'$ always holds. We thus obtain $1\leq k<a_k'=a_k\leq b_i$, implying that $p_{a_k}\in P[1,b_i]$. 
    Therefore, $p_{a_k}\in P[a_i,b_i]$ holds. 
\end{enumerate}

The lemma thus follows. 
\end{proof}

Define an array $A[1\cdots n]$ such that $A[k]=a_k'$ for each $1\leq k\leq n$. 
In light of Lemma~\ref{lem:feasiblepoint}, for each point $p_i\in P$, we can determine whether $p_i$ is a feasible point using at most two {\em range-minima} queries of the following type: Given a range $[i,j]$ with $i\leq j$, find the minimum number in the subarray $A[i\cdots j]$. It is possible to answer each range-minima query in $O(1)$ time after $O(n)$ time preprocessing on $A$~\cite{ref:BenderTh00,ref:HarelFa84}. For our problem, since it suffices to have $O(\log n)$ query time and $O(n\log n)$ preprocessing time, we can use a simple solution by constructing an augmented binary search tree. As such, in $O(n\log n)$ time we can find a feasible point or report that no such point exists. 

In summary, in $O(n\log n)$ time we can determine whether $P$ has three points whose minimum pairwise distance is at least $r$. If the answer is yes, then these three points can also be found within the same time complexity according to the proofs of Lemmas~\ref{lem:3points} and \ref{lem:feasiblepoint}. 
We conclude with the following theorem. 

\begin{theorem}\label{theo:convexindset3}
Given a set $P$ of $n$ points in convex position in the plane and a number $r$, in $O(n\log n)$ time one can find three points of $P$ whose minimum pairwise distance is at least $r$ or report that no such three points exist. 
\end{theorem}

\section{The dispersion problem}
\label{sec:dispersion}

Given a set $P$ of $n$ points in convex position in the plane and a number $k$, the dispersion problem is to find a subset of $k$ points from $P$ so that the minimum pairwise distance of the points of the subset is maximized. We first discuss the general-$k$ case in Section~\ref{sec:generaldispersion} and then the size-3 case in Section~\ref{sec:3dispersion}. 

\subsection{The general-$k$ case}
\label{sec:generaldispersion}
Let $r^*$ be the optimal solution value, that is, $r^*$ is the minimum pairwise distance of the points in an optimal solution subset. It is not difficult to see that $r^*$ is equal to the distance of two points of $P$. Define $R$ as the set of pairwise distances of the points of $P$. We have $r^*\in R$ and $|R|=O(n^2)$. 

Given a value $r$, the {\em decision problem} is to determine whether $r< r^*$, or equivalently, whether $P$ has a subset of $k$ points whose minimum pairwise distance is larger than $r$. By Corollary~\ref{coro:10}, the decision problem can be solved in $O(n^{7/2})$ time or in $O(n^{37/11})$ randomized expected time. Using the decision algorithm and doing binary search on the sorted list of $R$, $r^*$ can be computed in $O(n^{7/2}\log n)$ time or in $O(n^{37/11}\log n)$ randomized expected time. The following theorem summarizes the result.

\begin{theorem}
Given a set of $n$ points in convex position in the plane and a number $k$, one can find a subset of $k$ points whose minimum pairwise distance is maximized in $O(n^{7/2}\log n)$ deterministic time, or in $O(n^{37/11}\log n)$ randomized expected time. 
\end{theorem}
\begin{proof}
We first compute $r^*$ as discussed above. Once $r^*$ is computed, we can apply the algorithm of Corollary~\ref{coro:10} on $r=r^*-\delta$, for an infinitely small symbolic value $\delta$, to find an optimal subset of $k$ points whose minimum pairwise distance is $r^*$. More specifically, whenever the algorithm of Corollary~\ref{coro:10} attempts to compare a value $r'$ with $r$, we assert $r'<r$ if $r'<r^*$, and $r'>r$ if $r'\geq r^*$. 
\end{proof}

\subsection{The size-$3$ case}
\label{sec:3dispersion}
We now consider the case where $k=3$. Given a set $P$ of $n$ points in convex position in the plane, the problem is to find a subset of three points so that their minimum pairwise distance is maximized. 

Let $r^*$ be the optimal solution value, that is, $r^*$ is the minimum pairwise distance of the three points in an optimal solution. It is not difficult to see that $r^*$ is equal to the distance of two points of $P$. Define $R$ as the set of pairwise distances of the points of $P$. We have $r^*\in R$ and $|R|=\Theta(n^2)$. 

Given a number $r$, the {\em decision problem} is to determine whether $r\leq r^*$, or equivalently, whether $P$ has three points whose minimum pairwise distance is at least $r$. By Theorem~\ref{theo:convexindset3}, the decision problem is solvable in $O(n\log n)$ time. Specifically, if we apply the algorithm of Theorem~\ref{theo:convexindset3}, then the algorithm will return with an affirmative answer if and only if $r\leq r^*$. 
In what follows, for convenience, we refer to the algorithm of Theorem~\ref{theo:convexindset3} as the {\em decision algorithm}. 

If we compute $R$ explicitly and then do binary search on $R$ using the decision algorithm, then the total time would be $\Omega(n^2)$ as $|R|=\Theta(n^2)$. Another more efficient solution is to use distance selection algorithms that can find the $k$-th smallest value in $R$ in $O(n^{4/3}\log n)$ time for any given $k$~\cite{ref:KatzAn97,ref:WangIm23}. In fact, if we apply the algorithmic framework proposed by Wang and Zhao in \cite{ref:WangIm23} using our decision algorithm, then $r^*$ can be computed in $O(n^{4/3}\log n)$ time. In the following, we propose an algorithm of $O(n\log^2 n)$ time using parametric search~\cite{ref:ColeSl87,ref:MegiddoAp83}. 



We follow the standard parametric search framework~\cite{ref:MegiddoAp83} and simulate the decision algorithm on the unknown optimal value $r^*$ with an interval $[r_1,r_2)$ that contains $r^*$. At each step, the decision algorithm may be invoked on certain {\em critical values} $r$ to resolve a comparison between $r$ and $r^*$, and specifically, to determine whether $r\leq r^*$; based on the results of the comparison, the algorithm then proceeds accordingly and may shrink the interval $[r_1,r_2)$ by updating one of $r_1$ and $r_2$ to $r$ so that the new interval still contains $r^*$. Once the algorithm finishes, we can show that $r^*=r_1$ must hold. 
Initially, we set $r_1=-\infty$  and $r_2 = \infty$. Clearly, $r^*\in [r_1,r_2)$ holds.


For a parameter $r$, we use $a_i(r)$ and $b_i(r)$ to refer to $a_i$ and $b_i$ defined with respect to $r$. According to our decision algorithm, there are two main procedures. The first one is to compute $a_i(r^*)$ and $b_i(r^*)$ for all points $p_i\in P$ and the second one is to find a feasible point. Observe that the second procedure relies only on $a_i(r^*)$ and $b_i(r^*)$. Therefore, once $a_i(r^*)$ and $b_i(r^*)$ for all $p_i\in P$ are computed, the result of the second procedure is determined; in other words, when we run the second procedure, it does not produce any critical values $r$, meaning that the interval $[r_1,r_2)$ obtained after the first procedure is the final interval for the entire algorithm. Therefore, if $[r_1,r_2)$ is the interval obtained after the first procedure, then we can return $r_1$ as $r^*$. Computing $a_i(r^*)$ and $b_i(r^*)$ for all $p_i\in P$ can be done by following the same algorithm as in Section~\ref{sec:discretecenter} (i.e., the first iteration) by using Theorem~\ref{theo:convexindset3} as the decision algorithm (note the the algorithm in Section~\ref{sec:discretecenter} maintains an interval $(r_1,r_2]$ instead of $[r_1,r_2)$ because the decision algorithm there determines whether $r\geq r^*$ while the decision algorithm here determines whether $r\leq r^*$; correspondingly, the algorithm in Section~\ref{sec:discretecenter} returns $r^*=r_2$ while the algorithm here returns $r^*=r_1$). As such, $r^*$ can be computed in $O(n\log^2 n)$ time. 
To find an actual optimal solution, i.e., three points of $P$ whose minimum pairwise distance is equal to $r^*$, we can simply apply the decision algorithm on $r=r^*$, which will find such three points. 

Additionally, we can solve the problem in $O(n\log n)$ expected time using Chan's randomized technique~\cite{ref:ChanGe99}. We begin by partitioning $P$ into four subsets, $P_1, P_2, P_3,$ and $P_4$, each of size $n/4$, and define $g(P)$ as the maximum minimum pairwise distance of three points in $P$. Then, $g(P) = \max(g(P_1 \cup P_2 \cup P_3), g(P_1 \cup P_2 \cup P_4), g(P_1 \cup P_3 \cup P_4), g(P_2 \cup P_3 \cup P_4))$, since three points determining $g(P)$ must lie within one of these sets. This reduces the problem to a constant number of subproblems, each of size $3n/4$. Consequently, applying Chan's randomized technique~\cite{ref:ChanGe99} with our $O(n\log n)$ time decision algorithm can solve the problem in $O(n\log n)$ expected time.

We summarize our result in the following theorem. 

\begin{theorem}
Given a set of $n$ points in convex position in the plane, one can find three points whose minimum pairwise distance is maximized in $O(n\log^2 n)$ deterministic time or in $O(n\log n)$ randomized expected time. 
\end{theorem}

\section{The size-$\boldsymbol{3}$ weighted independent set for points in arbitrary position}
\label{sec:generalis3weight}

Given a set $P$ of $n$ points in the plane (not necessarily in convex position) such that each point has a weight, the problem is to find a maximum-weight independent set of size $3$ in $G(P)$, or equivalently, find $3$ points of maximum total weight whose minimum pairwise distance is larger than $1$. 
Note that since the size of our target independent set is fixed, we allow points to have negative weights. 
We present an $O(n^{5/3+\delta})$ time algorithm for the problem. 


In the following, we first introduce in Section~\ref{sec:biclique} a new concept, {\em tree-structured biclique partition}, which is critical to the success of our approach; we then describe the algorithm in Section~\ref{sec:algo3weight}. At the very end, we show that our technique can also be used to compute a maximum-weight clique of size $3$ in $G(P)$ within the same time complexity. In addition, we show that computing a maximum-weight independent set or clique of size 2 can be done in $n^{4/3}2^{O(\log^* n)}$ time by using biclique partitions. 

\subsection{Tree-structured biclique partition}
\label{sec:biclique}

Define $\overline{G(P)}$ as the complement graph of $G(P)$. The problem is equivalent to finding a maximum-weight clique of size $3$ in $\overline{G(P)}$. We want to partition $\overline{G(P)}$ into bicliques, i.e., complete bipartite graphs. We give the formal definition below. 

\begin{definition}    
    \label{prob:biclique}
    {\em \bf (Biclique partition)}
    Define a {\em biclique partition} of $\overline{G(P)}$ as a collection of edge-disjoint bicliques (i.e., complete bipartite graphs) $\Gamma(P) = \{A_t \times B_t\ |\ A_t, B_t \subseteq P\}$ such that the following two conditions are satisfied:
    \begin{enumerate}
        \item For each pair $(a, b) \in A_t \times B_t\in \Gamma$, $|ab|> 1$.
        \item For any points $a,b \in P$ with $|ab|> 1$, $\Gamma$ has a unique biclique $A_t \times B_t$ that contains $(a, b)$.
    \end{enumerate}
\end{definition}

Biclique partition has been studied before, e.g., \cite{ref:KatzAn97,ref:WangIm23}. In our problem, we need a stronger version of the partition, called a {\em tree-structured biclique partition} and defined as follows. 

\begin{definition}
{\em \bf (Tree-structured biclique partition)}
A biclique partition $\Gamma(P) = \{A_t \times B_t\ |\ A_t, B_t \subseteq P\}$ is {\em tree-structured} if all the subsets $A_t$'s form a tree $\calT_A$ such that for each internal node $A_t$, all its children subsets form a partition of $A_t$. 
\end{definition}

For convenience, for each node $A_t$ of $\calT_A$, we consider $A_t$ an ancestor of itself. Although biclique partitions have been studied and used extensively in the literature, e.g., \cite{ref:AgarwalCo06,ref:ChanFi23,ref:KatzAn97,ref:WangIm23}, to the best of our knowledge, we are not aware of any previous work on the tree-structured biclique partitions. 
The following lemma explains why we introduce the concept. 

\begin{lemma}\label{lem:treecondition}
Suppose that $\Gamma(P) = \{A_t \times B_t\ |\ A_t, B_t \subseteq P\}$ is a tree-structured biclique partition of $\overline{G(P)}$ and $\calT_A$ is the tree formed by all the subsets $A_t$'s. 
Then, three points $a,b,c\in P$ form an independent set in $G(P)$ if and only if $\Gamma(P)$ has a biclique $(A_t,B_t)$ that contains a pair of these points, say $(a,b)$, and $A_t$ has an ancestor subset $A_{t'}$ in $\calT_A$ such that $c\in B_{t'}$ and $|bc|> 1$.    
\end{lemma}
\begin{proof}
If three points $a,b,c\in P$ form an independent set in $G(P)$, then $|ab|$, $|bc|$, and $|ac|$ are all greater than $1$. Since $|ab|> 1$, by definition, $\Gamma(P)$ must have a biclique $(A_t,B_t)$ that contains $(a,b)$. Similarly, as $|ac|> 1$, $\Gamma(P)$ must have a biclique $(A_{t'},B_{t'})$ that contains $(a,c)$. We argue that one of $A_t$ and $A_{t'}$ must be an ancestor of the other. Indeed, if $t=t'$, then this is obviously true. If $t\neq t'$, then $a$ is in both $A_t$ and $A_{t'}$. 
Let $A_u$ be the highest node of $\calT_A$ that contains $a$. Since for each internal node $A_v$ of $\calT_A$, all its children subsets form a partition of $A_{v}$, $A_u$ must be the root of $\calT_A$. Also, exactly one child of $A_u$ contains $a$. If we follow this argument inductively, since both $A_t$ and $A_{t'}$ contain $a$,
one of them must be an ancestor of the other. If $A_{t'}$ is an ancestor of $A_t$, then the lemma statement is proved; otherwise, we simply switch the notation between $b$ and $c$. 

On the other hand, if $\Gamma(P)$ has a biclique $A_t\times B_t$ that contains $(a,b)$ and $A_t$ has an ancestor subset $A_{t'}$ in $\calT_A$ such that $c\in B_{t'}$ and $|bc|> 1$, we need to show that $\{a,b,c\}$ is an independent set. It suffices to prove $|ab|> 1$ and $|ac|> 1$. Indeed, since $(a,b)\in A_t\times B_t$, by definition $|ab|> 1$ holds. It remains to prove $|ac|> 1$. Since for each internal node $A_v$ of $\calT_A$, all its children subsets form a partition of $A_{v}$, and $A_{t'}$ is an ancestor of $A_t$, we obtain $A_t\subseteq A_{t'}$. As $a\in A_t$, we have $a\in A_{t'}$. Since $c\in B_{t'}$, we have $(a,c)\in A_{t'}\times B_{t'}$. Therefore, $|ac|> 1$ holds. 
\end{proof}

Lemma~\ref{lem:treecondition} suggests the following algorithm. First, we construct a tree-structured biclique partition $\Gamma(P) = \{A_t \times B_t\ |\ A_t, B_t \subseteq P\}$; let $\calT_A$ be the tree formed by all subsets $A_t$'s. For this, we will propose an algorithm in Lemma~\ref{lem:biclique}. 
Second, for each subset $B_t$, for each point $b\in B_t$, for each ancestor subset $A_{t'}$ of $A_t$, according to Lemma~\ref{lem:treecondition}, $\{a,b,c\}$ is an independent subset of $G(P)$ for all points $a\in A_t$ and all points $c\in B_{t'}$ with $|bc|> 1$. Instead of enumerating all these triples and then returning the one with the largest total weight, we find a triple $\{a^*,b,c^*\}$ (and keep it as a {\em candidate solution}) with $a^*$ as the point of $A_t$ with the largest weight and $c^*$ as the largest-weight point among all points $c\in B_{t'}$ with $|bc|> 1$. To compute $a^*$, for each subset $A_t\in \Gamma(P)$, we maintain its largest-weight point. To compute $c^*$, in Lemma~\ref{lem:maxweightquery} we build an {\em outside-unit-disk range max-weight query} data structure for each $B_{t'}$ so that given a query point $b$, such a point $c^*$ can be computed efficiently. Finally, among all candidate solutions, we return the one with the largest total weight. The efficiency of the algorithm hinges on the following: the time to construct $\Gamma(P)$, $\sum_t|B_t|$, i.e., the total size of all the subsets $B_t$'s, the number of ancestors of each $A_t$ in $\calT_A$ (i.e., the height of $\calT_A$), the time for constructing the outside-unit-disk range max-weight query data structure and its query time.

\subsection{Algorithm}
\label{sec:algo3weight}

Although efficient algorithms exist for computing biclique partitions of $\overline{G(P)}$~\cite{ref:KatzAn97,ref:WangIm23}, 
to the best of our knowledge, we are not aware of any previous algorithm to compute a tree-structured biclique partition. Also, we want the height of the tree as small as possible. 
We provide an algorithm in the following lemma using cuttings~\cite{ref:ChazelleCu93,ref:WangUn23}.

\begin{lemma}\label{lem:biclique}
A tree-structured biclique partition $\Gamma(P)=\{A_t \times B_t\ |\ A_t, B_t \subseteq P\}$ for $\overline{G(P)}$ can be computed in $O(n^{3/2})$ time with the following complexities: $|\Gamma(P)|=O(n)$, $\sum_t |A_t|=O(n\log n)$, and $\sum_t |B_t| = O(n^{3/2})$. In addition, the height of the tree formed by the subsets of $A_t$'s is $O(\log n)$.
\end{lemma}    
\begin{proof}
Define $\calC$ as the set of unit circles centered at the points of $P$. For each point $p\in P$, let $D_p$ denote the closed unit disk centered at $p$ and $C_p$ the boundary of $D_p$. 

We follow the notation about cuttings introduced in Section~\ref{sec:cutting}. 
We start by constructing a hierarchical $(1/r)$-cutting $\{\Xi_0, \Xi_1, ..., \Xi_k\}$ for $\calC$, which takes $O(nr)$ time~\cite{ref:ChazelleCu93,ref:WangUn23}, for a parameter $1\leq r\leq n$ to be fixed later. We use $\Xi$ to refer to the set of all cells $\sigma$ in all cuttings $\Xi_i$, $0\leq i\leq k$. 
For each cell $\sigma\in \Xi$, to be consistent with the notation in the lemma statement, denote by $A(\sigma)$ the subset of points of $P$ in $\sigma$. We compute $A(\sigma)$ for all cells $\sigma\in \Xi$. This can be done in $O(n\log r)$ time by processing each point of $P$ using a point location procedure as in the proof of Lemma~\ref{lem:SubProb}. 
Note that $\sum_{\sigma\in \Xi}|A(\sigma)|=O(n\log r)$. 
    
Next, for each cell $\sigma$ of $\Xi$, we compute another subset $B_{\sigma}\subseteq P$. Specifically, a point $p\in P$ is in $B_{\sigma}$ if $\sigma$ is completely outside the unit disk $D_p$ and the unit circle $C_p$ is in the conflict list $\calC_{\sigma'}$ of the parent $\sigma'$ of $\sigma$. The subsets $B_{\sigma}$ for all cells $\sigma$ of $\Xi$ can be computed in $O(nr)$ time as follows. Recall that the cutting algorithm~\cite{ref:ChazelleCu93,ref:WangUn23} already computes the conflict lists $\calC_{\sigma}$ for all cells $\sigma\in \Xi$. For each cutting $\Xi_{i-1}$, $1\leq i\leq k$, for each cell $\sigma'$ of $\Xi_{i-1}$, for each circle $C\in \calC_{\sigma'}$, for each child $\sigma$ of $\sigma'$, if $\sigma$ is completely outside $D_p$, we add the point $p$ to $B_{\sigma}$, where $p$ is the center of $C$. In this way, $B_{\sigma}$ for all cells $\sigma$ of $\Xi$ can be computed in $O(nr)$ time since $\sum_{\sigma'\in \Xi} |\calC_{\sigma'}|=O(nr)$ and each cell $\sigma'$ has $O(1)$ children. As such, $\sum_{\sigma\in \Xi}|B_{\sigma}|=O(nr)$. 

\paragraph{Constructing $\boldsymbol{\Gamma(P)=\{A_t \times B_t\ |\ A_t, B_t \subseteq P\}}$.} We now construct $\Gamma(P)$.
By definition, for each cell $\sigma\in \Xi$, for any point $b\in B_{\sigma}$, $\sigma$ is completely outside the unit disk $D_b$, and therefore $|ab|> 1$ for any point $a\in A(\sigma)$ since $a$ is contained in $\sigma$. We add the biclique $A(\sigma)\times B_{\sigma}$ to $\Gamma(P)$. It is not difficult to see that the bicliques of $\Gamma(P)=\{A(\sigma)\times B_{\sigma}\ |\ \sigma\in \Xi\}$ are edge-disjoint. 
The size of $\Gamma(P)$ is at most the number of cells of $\Xi$, which is $O(r^2)$. Also, we have shown above that $\sum_{\sigma\in \Xi}|A(\sigma)|=O(n\log r)$ and $\sum_{\sigma\in \Xi}|B_{\sigma}|=O(nr)$. 
We add some additional bicliques to $\Gamma(P)$ in the following. 

For each cell $\sigma$ of the last cutting $\Xi_k$, 
for each point $a\in A(\sigma)$, define $A_a=\{a\}$ and $B_a$ as the set of all points $b\in P$ such that $|ab|> 1$ and the unit circle $C_b$ is in the conflict list $\calC_{\sigma}$. We add $A_a\times B_a$ to $\Gamma(P)$ for all points $a\in A(\sigma)$ and for all cells $\sigma\in \Xi_k$. Note that the bicliques of $\Gamma(P)$ are still edge-disjoint.
Computing $B_a$ can be done by simply checking all circles of $\calC_{\sigma}$, which takes $O(n/r)$ time since $|\calC_{\sigma}|\leq n/r$. 
Hence, computing $B_a$ for all points $a\in P$ takes $O(n^2/r)$ time since the sets $A(\sigma)$'s for all cells $\sigma\in \Xi_k$ are pairwise-disjoint. This also implies $\sum_{a\in P}|B_a|=O(n^2/r)$. 

This completes the construction of $\Gamma(P)$. The total time is $O(nr+n^2/r)$. 
Note that $\sum_t|A_t|=\sum_{\sigma\in \Xi}|A(\sigma)|+\sum_{a\in P}|A_a|=O(n\log r)$ and 
$\sum_t|B_t|=\sum_{\sigma\in \Xi}|B_{\sigma}|+\sum_{a\in P}|B_a|=O(nr+n^2/r)$. The size of $|\Gamma(P)|$ is at most the number of cells of $\Xi$, which is $O(r^2)$, plus the number of points of $P$, which is $n$. As such, $|\Gamma(P)|=O(r^2+n)$. Setting $r=\sqrt{n}$ leads to the complexities in the lemma. 

\paragraph{Proving that $\boldsymbol{\Gamma(P)}$ is a biclique partition.} We now prove that $\Gamma(P)$ is indeed a biclique partition of $\overline{G(P)}$. For each pair $(a,b)\in A_t\times B_t\in \Gamma$, by our construction, $|ab|> 1$ always holds. 

Consider two points $a,b\in P$ with $|ab|> 1$. We need to show that $\Gamma(P)$ has a unique biclique $A_t\times B_t$ that contains $(a,b)$. As the bicliques of $\Gamma(P)$ are edge-disjoint, it suffices to show that $\Gamma(P)$ has a biclique $A_t\times B_t$ that contains $(a,b)$. To see this, for each $0\leq i\leq k$, let $\sigma_i$ denote the cell of $\Xi_i$ that contains the point $a$. Let $j$ be the largest index, $0\leq j\leq k$, such that the circle $C_b$ is in the conflict list $\calC_{\sigma_j}$ of $\sigma_j$. Note that such an index $j$ must exist since $C$ must be in $\calC_{\sigma_0}$ (this is because $\sigma_0$ is the only cell of $\Xi_0$, which is the entire plane). If $j=k$, since $|ab|> 1$, the point $b$ must be in $B_a$  by definition and thus $(a,b)\in A_a\times B_a$. If $j\neq k$, then by the definition of $j$, $C_b$ is not in $\calC_{\sigma_{j+1}}$. As such, either $\sigma_{j+1}$ is completely inside the disk $D_b$ or completely outside it. As $|ab|> 1$, $a$ is outside $D_b$. 
Since $a$ is inside $\sigma_{j+1}$, we obtain that $\sigma_{j+1}$ must be completely outside $D_b$. By definition $(a,b)$ must be in $A(\sigma_{j+1})\times B_{\sigma_{j+1}}$. 

This proves that $\Gamma(P)$ is a biclique partition of $\overline{G(P)}$. 

\paragraph{Proving that $\boldsymbol{\Gamma(P)}$ is tree-structured.} All the subsets $A_t$'s of $\Gamma(P)$ can be formed a tree by simply following the tree structure of the hierarchical cutting $\Xi$. Specifically, for two cells $\sigma,\sigma'\in \Xi$ such that $\sigma$ is a child of $\sigma'$, we make $A(\sigma)$ a child of $A(\sigma')$. In addition, for each cell $\sigma$ of the last cutting $\Xi_k$, for each point $a\in A(\sigma)$, we make $A_a$ a child of $A(\sigma)$. It is not difficult to see that all subsets $A_t$'s of $\Gamma(P)$ now form a tree, and furthermore, for each subset $A_t$, all its children form a partition of $A_t$. Clearly, the height of the tree is at most $k+1$, which is $O(\log n)$ since $k=O(\log r)$ and $r=\sqrt{n}$.

The lemma thus follows. 
\end{proof}

The next lemma builds an outside-unit-disk range max-weight query data structure. 

\begin{lemma}\label{lem:maxweightquery}
Let $Q$ be a set of $m$ weighted points in the plane. For any parameter $r$ with $r\leq m/\log^2m$, one can build a data structure in $O(mr(m/r)^{\delta})$ time such that given a query unit disk $D$, the point of $Q$ outside $D$ with the largest weight can be computed in $O(\sqrt{m/r})$ time. 
\end{lemma}
\begin{proof}
We use the following {\em outside-unit-disk range searching} data structure developed in \cite{ref:WangUn23}: For any $r\leq m/\log^2m$, one can build in $O(mr(m/r)^{\delta})$ time a data structure for $Q$, so that for any query unit disk $D$, the number of the points of $Q$ outside $D$ can be computed in $O(\sqrt{m/r})$ time. Note that the paper describes the algorithm for {\em inside-unit-disk} queries, but as discussed in the paper the technique works for the outside-unit-disk queries too. 
The algorithm extends the techniques of the halfplane range searching~\cite{ref:MatousekEf92,ref:MatousekRa93}. The techniques work for other semi-group operations. More specifically, the algorithm maintains the cardinalities of some {\em canonical subsets} of $Q$. For a query disk $D$, cardinalities of certain pairwise-disjoint canonical subsets whose union is $Q\cap \overline{D}$ (recall that $\overline{D}$ is the region of the plane outside $D$) are added together to obtain the answer to the query.  

To solve our problem, we can slightly change the above data structure as follows. For each canonical subset computed by the data structure, instead of maintaining its cardinality, we maintain its largest-weight point. Then, during each query, instead of using the {\em addition} operation on cardinalities of canonical subsets, we take the {\em max} operation on the weights of the largest-weight points of these canonical subsets. These changes do not asymptotically affect the preprocessing time or the query time of the data structure. The lemma thus follows. 
\end{proof}

Following our algorithm discussed before and combining Lemmas~\ref{lem:biclique} and \ref{lem:maxweightquery}, we obtain the following result. 

\begin{theorem}\label{theo:3weightindset}
Given a set $P$ of $n$ weighted points in the plane, one can find a maximum-weight independent set of size $3$ in the unit-disk graph $G(P)$ in $O(n^{5/3+\delta})$ time, for any arbitrarily small constant $\delta>0$. 
\end{theorem}
\begin{proof}
We first compute a tree-structured biclique partition $\Gamma(P)=\{A_t \times B_t\ |\ A_t, B_t \subseteq P\}$ for $\overline{G(P)}$ in $O(n^{3/2})$ time by Lemma~\ref{lem:biclique}. Let $\calT_A$ denote the tree formed by the subsets $A_t$'s of $\Gamma(P)$. 

Second, for each subset $B_t$ of $\Gamma(P)$, we construct an outside-unit-disk range max-weight query data structure by Lemma~\ref{lem:maxweightquery}, with $r=m_t^{1/3}$ and $m_t=|B_t|$. This takes $O(m_t^{4/3+\delta})$ preprocessing time and each query can be answered in $O(m_t^{1/3})$ time (which is $O(n^{1/3})$ as $m_t\leq n$). By Lemma~\ref{lem:biclique}, $\sum_tm_t=O(n^{3/2})$. Therefore, the total time for constructing the data structure for all subsets $B_t$'s of $\Gamma(P)$ is on the order of $\sum_t m_t^{4/3+\delta}=\sum_t m_t\cdot m_t^{1/3+\delta}\leq n^{1/3+\delta}\cdot \sum_t m_t=O(n^{11/6+\delta})$.

For each subset $A_t$ of $\Gamma(P)$, we compute its largest-weight point by simply checking every point of $A_t$. As $\sum_t|A_t|=O(n\log n)$ by Lemma~\ref{lem:biclique}, doing this for all subsets $A_t$'s takes $O(n\log n)$ time. 

Next, for each subset $B_t$ of $\Gamma(P)$, for each point $b\in B_t$, for each ancestor $A_{t'}$ of $A_t$ in $\calT_A$, we compute the largest-weight point $c^*$ among all points $c$ of $B_{t'}$ with $|bc|> 1$ by applying an outside-unit-disk range max-weight query on $B_{t'}$ with the unit disk $D_b$ centered at $b$, and then keep $\{a^*,b,c^*\}$ as a {\em candidate solution}, where $a^*$ is the largest-weight point of $A_t$. We add the point $b$ to a set $S(B_{t'})$ (which is initially $\emptyset$).
As discussed above, each query takes $O(n^{1/3})$ time. Since $A_t$ has $O(\log n)$ ancestors in $\calT_A$ by Lemma~\ref{lem:biclique}, 
processing each point $b\in B_t$ as above takes $O(n^{1/3}\log n)$ time in total and computes $O(\log n)$ candidate solutions. Since $\sum_t|B_t|=O(n^{3/2})$, the total time for processing all points $b\in B_t$ for all subsets $B_t$'s of $\Gamma(P)$ is bounded by $O(n^{11/6}\log n)$. Also, a total of $O(n^{3/2}\log n)$ candidate solutions are found. Finally, we return the candidate solution with the largest total weight as the optimal solution. 

As such, within $O(n^{11/6+\delta})$ time we can find an optimal solution. In the following, we present an improved algorithm of $O(n^{5/3+\delta})$ time. 

According to the above discussion, our goal is to answer outside-unit-disk range max-weight queries on $B_{t'}$ for all points of $S(B_{t'})$, for all $t'$. Note that $\sum_{t'}|S(B_{t'})|=O(n^{3/2}\log n)$. To see this, for each $B_t$ of $\Gamma(P)$, for each point $b\in B_t$, we query the disk $D_b$ on $B_{t'}$ for all ancestors $A_{t'}$ of $A_t$. As $A_t$ has $O(\log n)$ ancestors, the disk $D_b$ is used for $O(\log n)$ queries. Since $\sum_t|B_t|=O(n^{3/2})$, we obtain that $\sum_{t'}|S(B_{t'})|=O(n^{3/2}\log n)$.

For notational convenience, we state the problem as follows: Answer outside-unit-disk range max-weight queries on $B_{t}$ for all points of $S(B_{t})$, for all $t$. The main idea of the improved algorithm is that for each $B_t$, when building the outside-unit-disk range max-weight query data structure, instead of using the parameter $r=m_t^{1/3}$, we use a different parameter so that each max-weight query on $B_t$ takes $O(n^{1/6})$ time. As such, the total query time is $O(n^{5/3}\log n)$ as $\sum_t|S(B_t)|=O(n^{3/2}\log n)$. However, we also need to show that the total preprocessing time for building all data structures for all $B_t$'s can be bounded by $O(n^{5/3+\delta})$. For this, we explore some special properties of the subsets $B_t$'s generated by our algorithm in Lemma~\ref{lem:biclique}. Details are discussed below. 

For each subset $B_t$, 
if $m_t\leq n^{1/3}$, then 
we build a max-weight query data structure by Lemma~\ref{lem:maxweightquery} with the parameter $r=1$. By Lemma~\ref{lem:maxweightquery}, the query time is on the order of $\sqrt{m_t}$, which is $O(n^{1/6})$ since $m_t\leq n^{1/3}$; the preprocessing time is $O(m_t^{1+\delta})$. Since $\sum_t|B_t|=O(n^{3/2})$, the total preprocessing time for all such ``small'' subsets $B_t$ with $|m_t|\leq n^{1/3}$ is $O(n^{3/2+\delta})$, which is $O(n^{5/3+\delta})$. 

It remains to consider the ``large'' subsets $B_t$ with $m_t>n^{1/3}$. For each such $B_t$, we build a max-weight query data structure by Lemma~\ref{lem:maxweightquery} with the parameter $r=m_t/n^{1/3}$. By Lemma~\ref{lem:maxweightquery}, the query time is $O(n^{1/6})$ and 
the preprocessing time is $O(n^{\delta}m_t^2/n^{1/3})$. In the following, we show that the total preprocessing time for all such large $B_t$'s are bounded by  $O(n^{5/3+\delta})$. 

We follow the notation in the proof of Lemma~\ref{lem:biclique}. But to differetiate from the notation $r$ used above, we use $r'$ to refer to the notation $r$ in the proof of Lemma~\ref{lem:biclique}. 
Recall that there are two types of subsets $B_t$'s in $\Gamma(P)$ produced by the algorithm of Lemma~\ref{lem:biclique}: (1) $B_a$ for all points $a\in P$; (1) $B_{\sigma}$ for all cells $\sigma\in \Xi$. 

We first bound the preprocessing time for $B_a$ for all points $a\in P$. Consider a point $a\in P$. Let $\sigma$ be the cell of the last cutting $\Xi_k$ that contains $a$. Recall that a point $b$ is in $B_a$ only if the unit circle $C_b$ is in the conflict list $\calC_{\sigma}$. Therefore, $|B_a|\leq |\calC_{\sigma}|\leq n/r'$, with $r'=\sqrt{n}$. Hence, $|B_a|\leq \sqrt{n}$. Consequently, the total preprocessing time for $B_a$ of all points $a\in P$ is on the order of $\sum_{p\in P}n^{\delta}|B_a|^2/n^{1/3}\leq n^{\delta} \sum_{p\in P}n^{2/3}$, which is $O(n^{5/3+\delta})$. 

We next bound the total preprocessing time for $B_{\sigma}$ for all cells $\sigma\in \Xi$. Recall that each cutting $\Xi_i$ in the hierarchical cutting $\{\Xi_0,\Xi_1,\cdots,\Xi_k\}$ is an $\rho^i$-cutting, for some constant $\rho>0$. 
For any $1\leq i\leq k$, for each cell $\sigma\in \Xi_i$,  recall that a point $p$ is in $B_{\sigma}$ only if the unit circle $C_p$ is in the conflict list $\calC_{\sigma'}$ of the parent $\sigma'$ of $\sigma$. Therefore, $|B_{\sigma}|\leq |\calC_{\sigma'}|$ holds. Since $\sigma'\in \Xi_{i-1}$, $|\calC_{\sigma'}|\leq n/\rho^{i-1}$ and thus $|B_{\sigma}|\leq n/\rho^{i-1}$. Since $\Xi_i$ has $O(\rho^{2i})$ cells and $\rho$ is a constant, we obtain that the total preprocessing time for $B_{\sigma}$ for all cells $\sigma \in \Xi_i$ is on the order of $\sum_{\sigma\in \Xi_i}n^{\delta}|B_{\sigma}|^2/n^{1/3}\leq n^{\delta}\sum_{\sigma\in \Xi_i}n^{5/3}/\rho^{2(i-1)}=n^{\delta}\cdot O(\rho^{2i})\cdot n^{5/3}/\rho^{2(i-1)}=O(n^{5/3+\delta})$. As such, the total preprocessing time for $B(\sigma)$ for all cells $\sigma\in \Xi$ is $O(kn^{5/3+\delta})$, which is $O(n^{5/3+\delta})$ for a slightly larger $\delta$ as $k=O(\log n)$. 

The lemma thus follows. 
\end{proof}

\paragraph{Computing a maximum-weight clique of size $\boldsymbol{3}$.}
Our above algorithm can be easily modified to find a maximum-weight clique of size $3$ in the unit-disk graph $G(P)$. We briefly discuss it. 
First of all, we define the biclique partition for $G(P)$ instead of for $\overline{G(P)}$ in a similar way. Then, we can also define tree-structured biclique partitions for $G(P)$. Suppose $\Gamma(P) = \{A_t \times B_t\ |\ A_t, B_t \subseteq P\}$ is a tree-structured biclique partition of $G(P)$ and $\calT_A$ is the tree formed by all the subsets $A_t$'s. 
Then, by a similar argument to Lemma~\ref{lem:biclique}, we can prove the following: Three points $a,b,c\in P$ form a clique in $G(P)$ if and only if $\Gamma(P)$ has a biclique $(A_t,B_t)$ that contains a pair of these points, say $(a,b)$, and $A_t$ has an ancestor subset $A_{t'}$ in $\calT_A$ such that $c\in B_{t'}$ and $|bc|\leq 1$. Consequently, we can follow a similar algorithm framework as above. 
First, to compute a tree-structured biclique partition $\Gamma(P)$ of $G(P)$, we can slightly modify the algorithm of Lemma~\ref{lem:biclique}. Specifically, for each cell $\sigma\in \Xi$, we define $A(\sigma)$ in the same way as before but define $B_{\sigma}$ as the set of points $p\in P$ such that $\sigma$ is completely inside $D_p$ and $C_p$ is in the conflict list $\calC_{\sigma'}$ of the parent $\sigma'$ of $\sigma$. 
We add $A(\sigma)\times B_{\sigma}$ to $\Gamma(P)$. Similarly, for each cell $\sigma$ of $\Xi_k$, for each point $a\in A(\sigma)$, we now define $B_a$ as the set of all points $b\in P$ such that $|ab|\leq 1$ and $C_b$ is in $\calC_{\sigma}$; we add $A_a\times B_a$ to $\Gamma(P)$. In this way, following the algorithm similarly, we can compute a tree-structured biclique partition $\Gamma(P)$ of $G(P)$ with the same complexities as Lemma~\ref{lem:biclique}. Next, we need to construct an {\em inside-unit-disk range max-weight query} data structure. We can follow the same approach as in Lemma~\ref{lem:maxweightquery} by modifying the disk range searching data structure in \cite{ref:WangUn23}. Finally, using the above results, we can follow the same algorithm as Theorem~\ref{theo:3weightindset} to compute a maximum-weight clique of size $3$ in $G(P)$ in $O(n^{5/3+\delta})$ time.

\paragraph{Computing a maximum-weight independent set (resp., clique) of size $2$.}
To compute a maximum-weight independent set of size $2$ in $G(P)$, we can first compute a biclique partition $\Gamma(P) = \{A_t \times B_t\ |\ A_t, B_t \subseteq P\}$ for $\overline{G(P)}$, but a tree-structured one is not necessary. The algorithm of \cite{ref:WangIm23} can compute $\Gamma(P)$ in $n^{4/3}2^{O(\log^* n)}$ time with the following complexities: (1) $|\Gamma(P)|=O(n^{4/3}\log^*n)$; (2) $\sum_t |A_t|, \sum_t |B_t| = n^{4/3}2^{O(\log^* n)}$. Next, for each biclique $A_t\times B_t$ of $\Gamma(P)$, we find $\{a_t,b_t\}$ as a candidate solution, where $a_t$ is the largest-weight point of $A_t$ and $b_t$ is the largest-weight point of $B_t$. Finding all such candidate solutions can be easily done in $n^{4/3}2^{O(\log^* n)}$ time by brute force since $\sum_t |A_t|, \sum_t |B_t| = n^{4/3}2^{O(\log^* n)}$.  Finally, among all candidate solutions, we return the one with the largest total weight as an optimal solution. 
The total time is $n^{4/3}2^{O(\log^* n)}$. 

Computing a maximum-weight clique of size $2$ can be done similarly. The difference is that we use a biclique partition of $G(P)$. The algorithm of \cite{ref:WangIm23} can also compute such a biclique partition with the same complexities as above. In fact, the original algorithm of \cite{ref:WangIm23} is for computing a biclique partition for $G(P)$ but can be readily adapted to computing a biclique partition for $\overline{G(P)}$ with the complexities stated above. As such, a maximum-weight clique of size $2$ can also be computed in $n^{4/3}2^{O(\log^* n)}$ time.

\paragraph{Minimum-weight problems.}
We can also find a minimum-weight independent set or clique of size 3 (resp., 2) within the same time complexity as above, simply by negating the weight of every point and then applying the corresponding maximum-weight version algorithm discussed above.

\bibliographystyle{plainurl}
\bibliography{refs}
\end{document}